\documentclass[10pt,a4paper,reqno]{amsart}
\usepackage[a4paper,margin=2.5cm]{geometry}
\usepackage{amsmath,amssymb,amsthm,mathtools}
\usepackage{amscd}        
\usepackage[bibliography=common]{apxproof} 
\usepackage{natbib}
\usepackage{bibentry}    
\usepackage[T1]{fontenc}
\usepackage[utf8]{inputenc}
\usepackage{lmodern}  
\usepackage{graphicx}
\graphicspath{{figures/}}
\usepackage{subcaption}
\usepackage{float}
\usepackage{tikz}
\usepackage{booktabs}
\usepackage[table]{xcolor}
\usepackage{array}
\usepackage{multirow}
\usepackage[foot]{amsaddr}
\usepackage{hyperref}
\hypersetup{
  colorlinks      = true,
  pdfauthor       = {G. Custers, S. Karbach, M. Friesen},
  pdftitle        = {Pricing and Hedging of Discretely Monitored Asian Options in the Volterra--Heston Model},
  pdfkeywords     = {Asian options, rough volatility, Volterra--Heston, hedging, affine processes}
}
\usepackage[capitalise]{cleveref}

\usepackage{bm}
\usepackage{bbm}
\usepackage{enumitem}  
\usepackage{empheq}    
\usepackage{microtype} 

\makeatletter
\newsavebox{\@brx}
\newcommand{\llangle}[1][]{\savebox{\@brx}{$\m@th{#1\langle}$}%
  \mathopen{\copy\@brx\kern-0.5\wd\@brx\usebox{\@brx}}}
\newcommand{\rrangle}[1][]{\savebox{\@brx}{$\m@th{#1\rangle}$}%
  \mathclose{\copy\@brx\kern-0.5\wd\@brx\usebox{\@brx}}}
\makeatother

\newcommand*{\E}{\mathop{} \mathrm{e}}

\DeclareMathOperator{\argmin}{argmin}

\newcommand{\R}{\mathbb{R}}
\newcommand{\C}{\mathbb{C}}

\theoremstyle{plain}
\newtheorem{theorem}{Theorem}[section]
\newtheorem{proposition}[theorem]{Proposition}
\newtheorem{lemma}[theorem]{Lemma}
\newtheorem{corollary}[theorem]{Corollary}
\newtheorem{assumption}[theorem]{Assumption}

\theoremstyle{definition}
\newtheorem{definition}[theorem]{Definition}

\theoremstyle{remark}
\newtheorem{remark}[theorem]{Remark}

\newtheoremstyle{example}
  {.3\baselineskip}{.3\baselineskip}
  {\normalfont}{0pt}{\bfseries}{.}{5pt plus 1pt minus 1pt}{}
\theoremstyle{example}
\newtheorem{example}[theorem]{Example}

\title[Asian Options in the Volterra--Heston Model]{Pricing and Hedging of Discretely Monitored Asian Options\\in the Volterra--Heston Model}

\author[G. Custers]{Gijs Custers$^{1}$}

\author[M. Friesen]{Martin Friesen$^{2}$}

\author[S. Karbach]{Sven Karbach$^{3}$}

\address{$^{1}$Department of Mathematical Sciences, Chalmers University of Technology and University of Gothenburg, 
SE--412 96 Gothenburg, Sweden.}
\email{gijs.custers@chalmers.se}

\address{$^{2}$Division of Mathematical Statistics, Lund University, Sweden}
\email{martin.friesen@matstat.lu.se}

\address{$^{3}$Korteweg--de~Vries Institute for Mathematics, University of Amsterdam, 
Science Park~904, 1098~XH Amsterdam, The~Netherlands.}

\email{sven@karbach.org}

\begin{document}
\begin{abstract}
We develop semi-closed pricing formulas and lifted-model hedging methods for
discretely monitored geometric and arithmetic Asian options in the Volterra--Heston stochastic volatility model. Exploiting the affine Volterra structure, we derive a tractable transform for the joint law of the terminal log-price and the discretely monitored geometric average. This transform yields semi-closed pricing formulas for geometric Asian options, which in turn provide effective control variates for Monte Carlo valuation of arithmetic Asian options. Under the stated real-moment and affine-transform hypotheses, we also derive the Galtchouk--Kunita--Watanabe decomposition for Fourier-representable payoffs and obtain a variance-optimal hedge in terms of the Riccati--Volterra equation and the forward-variance curve. Using $N$-factor Markovian approximations, we obtain a finite-dimensional numerical implementation for hedging Asian options. Our numerical experiments document factor convergence for a regular non-Markovian kernel and the effect of rebalancing frequency on hedging error. In the Heston benchmark, geometric Asian controls substantially reduce the variance of arithmetic-Asian price estimates and improve the finite-sample stability of regression-based hedging relative to direct regression.
\medskip

\end{abstract}

\maketitle

\vspace{-4mm}

\section{Introduction}

\subsection{Asian options in stochastic volatility models}

Asian options are path-dependent derivatives whose payoff depends on the average of the underlying asset over a pre-specified time window. They are among the most actively traded exotics in equity, FX, and especially commodity markets. In power and energy markets, for instance, swaps and futures are typically written on arithmetic averages of spot prices over a delivery period. Their popularity stems from their natural risk-reducing properties, where averaging reduces the impact of short-lived price spikes and lowers volatility. This makes them attractive to market participants. From a financial modelling perspective, market conventions distinguish between three key design choices: (i) \emph{geometric} vs.\ \emph{arithmetic} averaging, (ii) \emph{continuous} vs.\ \emph{discrete}
monitoring, and (iii) \emph{fixed} vs.\ \emph{floating} strikes. These distinctions crucially affect
analytical tractability, numerical implementation, pricing and ultimately hedging. 

In the classical Black--Scholes setting, geometric averages are lognormally distributed, enabling
closed-form valuation formulas for both fixed- and floating-strike contracts~\cite{KemnaVorst1990}. In contrast, no closed-form pricing formula analogous to the geometric case is available for arithmetic Asian options. Several approximations and numerical approaches have been developed, such as moment-matching and analytic approximations \cite{TurnbullWakeman1991}, Laplace-transform techniques \cite{GemanYor1993}, and PDE reductions to one spatial dimension \cite{RogersShi1995,Vecer2001}. These results form the methodological benchmark against which more advanced stochastic models should be compared.

It is well known that Black--Scholes pricing fails to reproduce implied volatility smiles and heavy tails. To address this, stochastic-volatility models, and most prominently the Heston model \cite{Heston1993}, have been introduced. The affine structure of the Heston model makes transform methods particularly effective. Combining change-of-numeraire arguments with Fourier inversion~\cite{GemanKarouiRochet1995,GilPelaez1951}, semi-closed-form pricing formulas for continuously monitored geometric Asian options with fixed and floating strikes have been derived in \cite{KimWee2014}. While continuous averaging simplifies analysis, actual Asian contracts are typically settled from a finite set of fixing dates, making discrete monitoring the economically relevant specification. Continuous-monitoring formulas often serve as approximations or control variates, but the difference between discrete and continuous sampling may be material, especially in volatile markets, and is therefore important for both valuation and hedging. Even in Markovian models, working directly with discrete fixings leads to more accurate pricing and risk management. Contributions on discretely monitored Asians include~\cite{MR3508667, MR3181982}.

\subsection{The Volterra--Heston stochastic volatility model}\label{subsec:volterra-hestonsv}

Over the past decade, a new paradigm for stochastic volatility modelling emerged, where literature provides empirical evidence for \emph{rough} or strongly \emph{path-dependent} volatility dynamics, with short-time scaling behaviour consistent with a fractional exponent $H \in (0,\frac{1}{2})$, see \cite{MR4960404, GatheralJaissonRosenbaum2018, MR4625912}. Through the convolution with a Volterra kernel $\mathcal{K}$ in the variance dynamics, one is able to simultaneously capture the path-dependence and roughness in Volterra stochastic volatility models. A canonical example is the rough Heston model, where the variance is driven by a fractional kernel of the form $\mathcal{K}(t)=t^{H-\frac12}/\Gamma(H+\tfrac12)$ for $H \in (0,\frac{1}{2})$. More generally, throughout this work, we shall suppose that the following assumption is satisfied:

\begin{assumption}\label{assump:K}
$\mathcal{K} \in L_{\mathrm{loc}}^2(\R_+) \cap C((0,\infty))$ and there exists $\gamma \in (0,2]$ with
  \begin{align}\label{eq: K increment condition}
    \int_{0}^{h} \mathcal{K}(t)^2\,\mathrm{d}t = O(h^\gamma) \ \ \text{ and } \ \  \int_{0}^{T} \bigl(\mathcal{K}(t+h) - \mathcal{K}(t)\bigr)^2\,\mathrm{d}t = O(h^\gamma)
  \end{align}
  for all $T \in (0,\infty)$ as $h \downarrow 0$. Moreover, for all $h\geq 0$, the shifted kernel $\Delta_h \mathcal{K}(t):= \mathcal{K}(t+h)$ satisfies
  \begin{enumerate}
    \item[(a)] $\Delta_h \mathcal{K}$ is nonnegative, nonincreasing, and not identically zero on $(0,\infty)$;
    \item[(b)] its resolvent of the first kind $L_h$ exists and is nonnegative and nonincreasing in
the sense that $s\longmapsto L_h([s,s+t])$ is nonincreasing for all $t \geq 0$.
  \end{enumerate}
\end{assumption}

Assumption \ref{assump:K} is satisfied for (a) the Gamma kernel $\mathcal{K}(t) = t^{H - \frac{1}{2}}\E^{-\lambda t}$ with $\gamma = 2H$, where $H \in (0, \frac{1}{2})$ and $\lambda \geq 0$; (b) the $\alpha$-$\log$-kernel $\mathcal{K}(t) = \log(1+t^{-\alpha})$ with any choice of $\gamma \in (0,1)$, where $\alpha \in (0,1]$; (c) locally Lipschitz continuous kernels $\mathcal{K}$ that are non-constant and completely monotone (in particular $\mathcal{K}(0) < \infty$) with $\gamma = 1$. The last class of examples covers, e.g., the hyperbolic kernel $\mathcal{K}(t) = c_H (\varepsilon + t)^{H-\frac{1}{2}}$ with $H \in (-\infty, \frac{1}{2})$ and $c_H, \varepsilon > 0$, and the $M$-factor exponential kernels with $c_i, \lambda_i \geq 0$ and $M \geq 1$ defined as $\mathcal{K}(t) = \sum_{i=1}^M c_i \E^{- \lambda_i t}$.

The asset price process $(S_{t})_{t\geq 0}$ is defined on some filtered probability space $(\Omega,\mathcal{F},\mathbb{F},\mathbb{P})$ as the unique weak solution to the coupled stochastic differential equations
\begin{align}
  \mathrm{d}S_t &= r\,S_t\,\mathrm{d}t 
    + \sqrt{\nu_t}\,S_t\,\mathrm{d}W^S_t, 
    \label{eq:VH-S} \\
  \nu_t &= \nu_0 
    + \int_0^t \mathcal{K}(t-s)\,
        \kappa\bigl(\theta - \nu_s\bigr)\,\mathrm{d}s
    + \int_0^t \mathcal{K}(t-s)\,
        \sigma\sqrt{\nu_s}\,\mathrm{d}W^\nu_s,
    \label{eq:VH-nu}
\end{align}
where $(\nu_t)_{t\geq 0}$ is the instantaneous variance. Here, $(\Omega,\mathcal{F},\mathbb{F},\mathbb{P})$ denotes a complete filtered probability space satisfying the usual conditions, $\mathbb{P}$ a risk-neutral measure, and the filtration $\mathbb{F}=(\mathcal{F}_t)_{t\ge 0}$ supports two independent $\mathbb{P}$-Brownian motions $(W^{S},W^{I})$, while the Brownian motion in the volatility dynamics is given by 
\[
  W^{\nu}_t := \rho\,W^{S}_t + \sqrt{1-\rho^{2}}\,W^{I}_t, \qquad t\ge 0,
\]
with correlation parameter $\rho \in [-1,1]$ satisfying $\mathrm{d}\langle W^S, W^\nu\rangle_t = \rho\,\mathrm{d}t$. Finally, $r \geq 0$ denotes the risk-free market rate, $\kappa > 0$ a mean-reversion strength parameter, $\theta > 0$ the long-term volatility level, and $\sigma > 0$ the vol-of-vol coefficient. We specify the model directly under the risk-neutral measure. Hence $\kappa$ and $\theta$ denote risk-neutral parameters. Constructing an equivalent pricing measure from a physical-measure Volterra specification is more delicate and generally requires additional assumptions on the drift and market price of volatility risk, see \cite{FGW26} and \cite{BJ26}.

Existence and uniqueness of a weak, continuous solution $(S,\nu)$ taking values in $(0,\infty) \times \mathbb{R}_+$ to~\eqref{eq:VH-S}-\eqref{eq:VH-nu} is established in~\cite[Thm.~7.1(i)]{abi2019affine}. The first moment of $S$ and higher moments of $\nu_t$ are uniformly bounded on finite time horizons, i.e. $\sup_{t\in[0,T]} \mathbb{E}\bigl[\, S_t + \nu_t^{\,q}\,\bigr] < \infty$ for all $q\ge 2$ and $T>0$, and the discounted price process $(\E^{-rt}S_t)_{t\geq 0}$ is a $\mathbb{P}$-martingale, see~\cite[Lem. 3.1 and Lem.~7.3]{abi2019affine}. In particular, the discounted stock price is a true martingale, and the model is well posed under the chosen pricing measure. As typical for stochastic-volatility models, and in particular when $|\rho|<1$, the stock-and-cash market is incomplete since volatility risk cannot be perfectly hedged solely by trading the underlying asset, see \cite{steinstein1991} for classical discussions of incompleteness in stochastic volatility frameworks, and \cite{euch2018perfect} for the rough Heston model.

\subsection{Pricing and hedging of Asian options}
In this paper, we study the pricing and hedging of geometric and arithmetic \emph{discretely monitored} Asian options with fixed and floating strikes in the Volterra--Heston stochastic volatility model. Given $T > 0$, let $0=t_0<t_1<\dots<t_N=T$ be a discrete monitoring grid for the underlying Asian option. The grid consists of $N+1$ monitoring dates $t_0,\dots,t_N$ (including the spot fixing at $t_0=0$), and throughout $N$ denotes the monitoring index. Define the geometric and arithmetic average
\begin{align}\label{eq: intro 1}
    G_T^N = \exp\left( \frac{1}{N+1}\sum_{j=0}^{N} \log(S_{t_j}) \right) \quad \text{ and } \quad A_T^N = \frac{1}{N+1}\sum_{j=0}^{N} S_{t_j}.
\end{align}
Then $(G^N_T-K)^+$ and $(A^N_T - K)^+$ denote the payoff for a \emph{fixed-strike discretely monitored geometric/arithmetic Asian call option} with strike $K>0$, while $(G^N_T - S_T)^+$ and $(A^N_T - S_T)^+$ denote the payoff for the \emph{floating-strike discretely monitored geometric/arithmetic Asian call option}, where $(x)^+:=\max\{x,0\}$. The corresponding risk-neutral prices are, for $t \in [0,T]$, given by
\begin{align*}
    C^N_{t} := \E^{-r(T-t)}\mathbb{E}\left[(G^N_T-K)^+\middle|\mathcal{F}_{t}\right]
    \ \text{ and } \ \widetilde{C}^N_{t} := \E^{-r(T-t)}\mathbb{E}\left[(G^N_T-S_T)^+\middle|\mathcal{F}_{t}\right].
\end{align*}
For arithmetic Asian options, the corresponding arbitrage-free prices are given in the same way with $G_T^N$ replaced by $A_T^N$.

For European options, semi-closed transform formulas and fast numerical methods are now well established \cite{ElEuchRosenbaum2019}, while American options have been studied in \cite{chevalier2022american}. By contrast, the literature on pricing and hedging exotic derivatives under rough and path-dependent volatility is still comparatively young, and analytical results for path-dependent payoffs remain much scarcer. Using Fourier pricing methods in combination with the affine property, semi-explicit pricing formulas for continuously monitored geometric Asian options were obtained in \cite{aichinger2024pricing} for the Volterra--Heston model. To the best of our knowledge, semi-closed transform formulas for discretely monitored geometric Asian options in the affine Volterra--Heston setting with measure-valued Riccati--Volterra inputs have not been developed. This gap is the focus of the first part of this work, whose key ingredient is a joint Fourier transform formula for the terminal log-price and the log-geometric average, expressed in terms of solutions to Riccati--Volterra equations with distribution-valued inputs. This extends the Markovian approach of~\cite{KimWee2014} to the affine Volterra setting, and complements \cite{aichinger2024pricing} by allowing for discretely observed monitoring dates. For the time-weighted left-Riemann discretisation of the continuous average, we establish a strong convergence rate of order $\Delta_N$ in $L^q$. The latter provides a rigorous justification when exact formulas for discrete observations can be replaced by continuous averages.

On the hedging side, stochastic-volatility markets are incomplete unless one can trade the underlying volatility factor through an additional instrument. In rough Heston models, perfect hedging becomes possible only after enlarging the market to include forward-variance related instruments \cite{euch2018perfect}. In this work, we focus on the practically relevant stock-and-cash market, where hedging Asian options is understood in a mean-variance or risk-minimisation framework. The variance-optimal strategy is characterised by the Galtchouk--Kunita--Watanabe decomposition. In the semimartingale Heston setting, this formalism has been carried out for European claims \cite{Kallsen2010}. Partial hedging under rough volatility is studied in \cite{MEHD2023}. Extending these ideas to discretely monitored Asian payoffs in affine Volterra models is substantially more delicate. We derive blockwise GKW formulas for discretely monitored geometric Asian options in the Volterra--Heston model. The resulting hedge is explicit in terms of the current spot, the forward-variance curve, and the solution of the Riccati--Volterra equation. A Markovian $N$-factor lift then yields a finite-dimensional approximation that is suitable for implementation. For non-singular kernels, we prove that the variance-optimal hedge of the lifted $N$-factor approximation converges to the true variance-optimal hedge. For singular fractional kernels, the numerical accuracy gate shows that the factor approximation used here is not sufficiently accurate to support pathwise hedge claims, and we therefore report no such claims.

Beyond their theoretical interest, our results provide effective computational tools. In particular, the semi-closed formulas for geometric Asians are well suited as control variates for Monte Carlo pricing of arithmetic Asians. On the hedging side, the semi-analytic geometric hedge is explicit in the affine Volterra state variable $(S_t,\xi_t(\cdot))$, while the Markovian lift provides the finite-dimensional representation used in the numerical implementation. We combine the semi-analytic geometric hedge with a backward-regression scheme for arithmetic-Asian hedging. In the Heston benchmark, this substantially stabilises direct regression, although the resulting improvement over the geometric-only hedge is small.

\subsection{Layout of the paper}

Section~\ref{sec:volterra-heston} derives the central measure-valued affine transformation formula, while semi-closed pricing formulas for Asian options, and their continuous limits, are studied in Section~\ref{sec:pricing}. Section~\ref{sec:hedging} develops the quadratic-hedging formulas, discusses the Markovian approximation, and reports reproducible hedging diagnostics. Section~\ref{sec:arithmetic-Asian-options} applies these results to control-variate pricing and hedging of arithmetic Asians.

\section{Joint Fourier-Laplace transform}\label{sec:volterra-heston}

\subsection{Volterra--Riccati equation with measure-valued input}\label{sec:VolterraRiccati}

Let $\mathcal{M}_{\mathrm{lf}}$ be the space of $(\C^2)^*$-valued locally finite measures on $\R_+$, that is, pairs $\mu=(\mu_1,\mu_2)$ of $\mathbb{C}$-valued locally finite measures on $\R_+$, acting on $\mathbb{C}^2$-valued integrands componentwise. For $\mu = (\mu_1, \mu_2) \in \mathcal{M}_{\mathrm{lf}}$, define
\begin{align}\label{eq:phi1}
    \psi_1(t, \mu) = \mu_1([0,t]),
\end{align}
and let $\psi_2 \in L_{\mathrm{loc}}^2(\R_+; \C)$ be given by the \emph{generalised Volterra--Riccati equation}
\begin{align}
    \psi_2(t,\mu) = \int_{[0,t]} \mathcal{K}(t-s)\,\mu_2(\mathrm{d}s) + \int_0^t \mathcal{K}(t-s)R(\psi(s,\mu))\,\mathrm{d}s. \label{eq:VolRic}
\end{align}
The first term in \eqref{eq:VolRic} is defined as an element in $L^2_{\mathrm{loc}}(\R_+)$, and $R$ denotes the quadratic function
\begin{align}\label{eq:R}
  R(x_1, x_2) = \tfrac{1}{2}(x_1^2 - x_1) - \kappa x_2 
       + \tfrac{1}{2}\bigl(\sigma^2x_2^2 + 2\rho\sigma x_1x_2\bigr).
\end{align}
The \emph{fractional Sobolev space} $W^{\eta,2}([0,T]; \C)$ with regularity $\eta\in(0,1)$ is defined as the Banach space of equivalence classes of functions $g:[0,T]\to\mathbb{C}$ with finite norm
\begin{align*}
    \|g\|_{W^{\eta,2}([0,T])} :=
    \left(
        \int_0^T |g(t)|^2\,\mathrm{d}t 
        + \int_0^T \int_0^T 
            \frac{|g(t)-g(s)|^2}{|t-s|^{1+2\eta}}\,\mathrm{d}s\,\mathrm{d}t
    \right)^{1/2}.
\end{align*}

For the vector-valued measure $\mu = (\mu_1, \mu_2) \in \mathcal{M}_{\mathrm{lf}}$ let us write $|\mu|, |\mu_1|, |\mu_2|$ for the corresponding variation measures on $\R_+$. Moreover, let us define
\[
    \mathcal{M}_{\mathrm{ad}} \coloneqq \left\{ \mu = (\mu_1, \mu_2) \in \mathcal{M}_{\mathrm{lf}} \ | \ \Re(\mu_1([0,t])) \in [0,1] \ \text{ for all }\ t \geq 0 \ \text{ and } \ \Re(\mu_2)\leq 0 \right\}.
\]
The next theorem establishes the existence and uniqueness of solutions $(\psi_1,\psi_2)$ to equations~\eqref{eq:phi1}-\eqref{eq:VolRic} for given $\mu \in \mathcal{M}_{\mathrm{ad}}$. 

\begin{theorem}\label{thm:solRicVolmeasure}
    Let $\mu \in\mathcal{M}_{\mathrm{ad}}$. For every $\eta\in(0,\gamma/2)$ there exists a unique solution $\psi_2(\cdot; \mu) \in W_{\mathrm{loc}}^{\eta, 2}(\mathbb{R}_+,\mathbb{C}_-)$ to \eqref{eq:VolRic}. Moreover, for every $p\in[2,\infty)$ with $\mathcal{K}\in L^p([0,T])$, there exists a constant $C > 0$ that only depends on the model parameters $(\rho,\sigma,\kappa)$, $p$, $T$, and $\|\mathcal{K}\|_{L^1([0,T])}$ such that
    \[
        \|\psi(\cdot, \mu)\|_{L^p([0,T])} \leq C\left (1 + \| \mathcal{K}\|_{L^{p}([0,T])}  \right)\Bigl(1 + |\mu|([0,T])\Bigr),
    \]
    and for $[\mathcal{K}]_{\eta, 2,T} := \left( \int_0^T t^{-2\eta}|\mathcal{K}(t)|^2\, \mathrm{d}t + \int_0^T \int_0^T \frac{|\mathcal{K}(t) - \mathcal{K}(s)|^2}{|t-s|^{1 + 2\eta}}\,\mathrm{d}s\,\mathrm{d}t \right)^{1/2}$ with $\eta \in (0,\gamma/2)$, we find
    \begin{align*}
    \|\psi_2(\cdot,\mu)\|_{W^{\eta,2}([0,T])} 
    &\leq \|\psi_2(\cdot, \mu)\|_{L^2([0,T])}+ C\big(1 + [\mathcal{K}]_{\eta,2,T}\big) \\
    &\quad \times
	        \Big(1 + |\mu|([0,T]) + \|\psi(\cdot,\mu)\|_{L^1([0,T])} 
	              + \|\psi(\cdot,\mu)\|_{L^2([0,T])}^2\Big).
	    \end{align*}
\end{theorem}

In appendix~\ref{sec:1} it is shown that condition \eqref{eq: K increment condition} implies $[\mathcal{K}]_{\eta, 2,T} < \infty$ for $\eta \in (0,\gamma/2)$, while a proof of this theorem is given in appendix~\ref{sec:proof-Theorem2.3}. Next, we focus on stability for the solution $\psi$ with respect to the initial condition $\mu$ and the Volterra kernel $\mathcal{K}$. 

\begin{theorem}\label{thm:stability-Lq}
    Let $\mu = (\mu_1, \mu_2), \mu^{(n)} = (\mu^{(n)}_1, \mu^{(n)}_2) \in \mathcal{M}_{\mathrm{ad}}$. Let $\mathcal{K}, \mathcal{K}^{(n)}$ be Volterra kernels that satisfy Assumption~\ref{assump:K}, and let $\psi = (\psi_1, \psi_2)$ and $\psi^{(n)} = (\psi^{(n)}_1, \psi^{(n)}_2)$ be the corresponding solutions of \eqref{eq:phi1} and \eqref{eq:VolRic}. 
    Suppose that $\sup_{n \geq 1}\  [\mathcal{K}^{(n)}]_{\eta,2,T} < \infty$ holds for some $\eta \in (0,\gamma/2)$, that $\sup_{n\ge1}|\mu^{(n)}|([0,T])<\infty$, and
    \begin{align*}
        \lim_{n \to \infty} \left( \| \mathcal{K}^{(n)} - \mathcal{K} \|_{L^2([0,T])} + |\mu_1([0,t]) - \mu^{(n)}_1([0,t])| + |\mu_2 - \mu_2^{(n)}|([0,T]) \right)= 0, \qquad \forall t \in [0,T].
	    \end{align*}
	    Then $\psi_1^{(n)} \longrightarrow \psi_1$ pointwise, and $\psi_2^{(n)} \longrightarrow \psi_2$ in $L^{2}([0,T]; \C)$.
\end{theorem}

The uniform variation bound $\sup_{n\ge1}|\mu^{(n)}|([0,T])<\infty$ is automatic for the monitoring measures used below, all of which have total mass one. The proof is given in Appendix~\ref{sec:proofthm24}. In some cases, one would like to have continuous solutions $\psi_2$, and stability uniformly in $t$. The next theorem provides a sufficient condition.

\begin{theorem}\label{thm:stability-uniform}
    Suppose that $\mu = (\mu_1, \mu_2), \mu^{(n)} = (\mu^{(n)}_1, \mu^{(n)}_2) \in \mathcal{M}_{\mathrm{ad}}$ with $\mu_2(\mathrm{d}s) = f_2(s)\mathrm{d}s$, $\mu_2^{(n)}(\mathrm{d}s) = f_2^{(n)}(s)\mathrm{d}s$, and there exists $q \in [3, \infty]$ such that $\mathcal{K}, \mathcal{K}^{(n)} \in L^q([0,T])$, $f_2, f_2^{(n)} \in L^{\frac{2q}{q-1}}([0,T]; \C)$,
    $$
        \lim_{n \to \infty}\left(\| \mathcal{K} - \mathcal{K}^{(n)}\|_{L^q([0,T])} + \| f_2 - f_2^{(n)}\|_{L^{\frac{q}{q-1}}([0,T]) } + |\mu_1([0,t]) - \mu^{(n)}_1([0,t])| \right) = 0
    $$
    for each $t \in [0,T]$. If $\sup_{n \geq 1}\  [\mathcal{K}^{(n)}]_{\eta,2,T} < \infty$ holds for some $\eta \in (0,\gamma/2)$ and
    \begin{align*}
        &\sup_{n \geq 1}\  \left( |\mu^{(n)}|([0,T]) + \| f_2^{(n)}\|_{L^{\frac{2q}{q-1}}([0,T])}\right) < \infty,
	    \end{align*}
	    and if $\|\psi_2\|_{L^\infty([0,T])}<\infty$ and $\sup_{n\ge1}\|\psi_2^{(n)}\|_{L^\infty([0,T])}<\infty$, where for $q=\infty$ we use the convention $\frac{q}{q-1}=1$ and $\frac{2q}{q-1}=2$, then $\psi_2, \psi_2^{(n)}$ are continuous on $[0,T]$, and $\psi_2^{(n)} \longrightarrow \psi_2$ uniformly on $[0,T]$. 
\end{theorem}

 The proof is given in Appendix~\ref{sec:proofthm25}. For the Volterra kernel, such an approximation could be based, for example, on a discretisation of the associated Bernstein measure. Approximations in $\mu_1, \mu_2$ could be based on approximations of $\mu_1, \mu_2$ through Dirac measures. Finally, concerning continuity of solutions on $(0,\infty)$, the assumption $\mathcal{K} \in L_{\mathrm{loc}}^3(\R_+)$ may be too restrictive, but might be weakened through analytical bootstrapping arguments, see \cite[Section 2]{FKW26}.

\subsection{Generalised affine transformation formula}\label{sec:affine-transform-formula}

Let us define $X_t := (\log S_t,\nu_t)^\top$. Then $X$ is an affine Volterra process in the sense of \cite{abi2019affine}. In particular, it admits a semi-explicit Fourier--Laplace transform in terms of a solution $(\psi_1, \psi_2)$ to a Volterra--Riccati equation as obtained in~\cite[Section 7]{abi2019affine}, see also \cite{ElEuchRosenbaum2019}. Following \cite{friesen2024volterra}, we provide an extension of the affine transformation formula to the case of measure-valued inputs $\mu = (\mu_1, \mu_2) \in \mathcal{M}_{\mathrm{ad}}$. The latter allows us to compute the Fourier--Laplace transform at discrete time points.

\begin{theorem}\label{cor:CorsolRicVolmeasure}
    Fix $T > 0$ and recall that $X_t = (\log S_t, \nu_t)^{\top}$. Let $\mu=(\mu_1,\mu_2)\in\mathcal{M}_{\mathrm{ad}}$. Then, for all $t\in[0,T]$,
\begin{align*}
    \mathbb{E}\left[\exp\left(\int_{[0,T]}\langle X_{T-s},\mu(\mathrm{d}s)\rangle\right)\middle|\mathcal{F}_t\right]
    &= \exp\Bigg( \int_{[0,T]} \langle \mathbb{E}[X_{T-s}\mid\mathcal{F}_t],\mu(\mathrm{d}s)\rangle 
    \\  &\qquad + \tfrac{1}{2}\int_0^{T-t}\psi(s)\, \mathbb{E}[\nu_{T-s}\mid\mathcal{F}_t]\,
        A^2\psi(s)^\top\,\mathrm{d}s
    \Bigg),
	\end{align*}
	where $\psi = (\psi_1, \psi_2)$ is the solution to~\eqref{eq:phi1} and~\eqref{eq:VolRic}, and $A^2 = \begin{pmatrix}1 & \rho \sigma \\ \rho \sigma & \sigma^2 \end{pmatrix}$. 
\end{theorem}

The proof of Theorem~\ref{cor:CorsolRicVolmeasure} is given in the appendix~\ref{sec:proof-affine-formula}. This formula covers the case of continuously monitored time points when $\mu(\mathrm{ds}) = u \delta_0(\mathrm{d}s) + f(s)\mathrm{d}s$, compare with \cite[Thm. 4.3]{aichinger2024pricing}, but also provides semi-explicit formulas for a fixed discrete monitoring grid by choosing $\mu_1$ as a linear combination of Dirac measures and $\mu_2 = 0$. Below, we provide a particular case that is essential for the pricing of different geometric Asian-type payoffs. 

Let $g$ be a finite nonnegative Borel measure on $[0,T]$ with $g([0,T])=1$. Define for $t \in [0,T]$
\begin{align}\label{eq:Gg}
    G_{t,T} := \exp\left( \int_{(t,T]} \log(S_{u})\, g(\mathrm{d}u) \right) \ \text{ and } \ G_t := \exp\left( \int_{[0,t]} \log(S_u)\, g(\mathrm{d}u) \right).
\end{align}
Note that $G_{T} = G_{t}G_{t,T}$. An application of the arithmetic-geometric mean inequality gives $G_{T} \leq \int_{[0,T]}S_{u}\, g(\mathrm{d}u)$ and hence $\mathbb{E}[ G_{T} ] \leq \int_{[0,T]} \mathbb{E}[S_{u}]\, g(\mathrm{d}u) = S_0 \int_{[0,T]} \E^{ru}\, g(\mathrm{d}u) < \infty$, since $(\E^{-ru}S_u)_{u \in [0,T]}$ is a martingale. Similarly, we see that also $\mathbb{E}[G_{t,T}] < \infty$. Below, we compute the joint Fourier--Laplace transform of $\log G_{t,T}$ and $\log S_T$ conditioned on $\mathcal{F}_{t}$ for the Volterra--Heston stochastic volatility model. This will be our key tool when formulating semi-explicit pricing formulas. 

\begin{corollary}\label{thm:condcfG_nS_Texpression}
Let $(s,w)\in\mathbb{C}^2$, $\Re(s), \Re(w) \geq 0$ be such that $\Re(w) + \Re(s) g([0,T]) \leq 1$. Define
\begin{align}\label{eq:psit}
    \Psi_t(s,w) := \mathbb{E}\left[\exp\left(s\log G_{t,T} + w\log S_T\right)\middle|\mathcal{F}_{t}\right],
\end{align}
 and let $\xi_s(t) := \mathbb{E}[\nu_t\mid\mathcal{F}_s]$ be the forward variance process. Then
\begin{align*}
    \Psi_t(s,w) &= \exp\Bigg( \log(S_t) \psi_{1,-}(T-t)  
        \\ &\qquad + r \left( s \int_{(t,T]} (y-t)\, g(\mathrm{d}y) + w(T-t) \right)
        - \frac{1}{2}\int_t^T \psi_{1,-}(T-x) \xi_t(x)\, \mathrm{d}x
        \\ &\qquad + \tfrac{1}{2}\int_t^T\Big( \psi_1(T-y)^2 + 2\rho\sigma \psi_1(T-y)\psi_2(T-y) + \sigma^2\psi_2(T-y)^2 \Big)\,\xi_t(y)\,\mathrm{d}y \Bigg),
\end{align*}
where $\psi=(\psi_1,\psi_2)$ is given by $\psi_1(t) = s g([T-t, T]) + w$, $\psi_{1,-}(t) = s g((T-t, T]) + w$, and $\psi_2$ satisfies \eqref{eq:VolRic} with $\mu_2 = 0$.
\end{corollary}

The proof is given in the appendix~\ref{sec:proof-CFform}. We close this section with the central examples for continuous and discretely monitored Asian options.

\begin{example}[continuous monitoring]\label{example:continuous}
    Suppose that $g_c(\mathrm{d}t) = \frac{1}{T}\mathrm{d}t$. Then 
    $$
        g_c((t,T]) = \frac{T-t}{T}, \quad \int_{(t,T]}(y-t)\, g_c(\mathrm{d}y) = \frac{(T-t)^2}{2T}, \quad g_c([T-t,T]) = \frac{t}{T}.
    $$
\end{example} 
\begin{example}[discrete monitoring]\label{example:discrete}
    Let $0 = t_0 < \dots < t_N = T$ be a fixed discrete monitoring grid with $N \geq 1$, and define
    $$
      g_{N}(\mathrm{d}t) =  \frac{1}{N+1}\sum_{j=0}^N \delta_{t_j}(\mathrm{d}t).
    $$
	For $t \in [0,T]$, let us define $n_t \in \{0, \dots, N-1\}$ by $t_{n_t}\leq t<t_{n_t+1}$ and $n_T=N$. Then 
    $$
        g_N((t,T]) = \frac{N - n_t}{N+1}, \quad g_N([T-t, T]) = \frac{ N - n_{T-t} + \mathbbm{1}_{\{t_{n_{T-t}} = T-t\}}}{N+1}
    $$
    and 
    $$
        \int_{(t,T]}(y-t)\, g_N(\mathrm{d}y) = \frac{1}{N+1}\sum_{j = n_t + 1}^N (t_j - t).
    $$
\end{example} 

Note that similar formulas can also be obtained for \textit{continuous weighted} averages $g(\mathrm{d}t) = g(t)\mathrm{d}t$ with $g(t)\ge0$ and $\int_0^T g(t)\, \mathrm{d}t = 1$, and \textit{discrete weighted} averages on a fixed grid $0 = t_0 < t_1 < \dots < t_N = T$ with weights $c_0, \dots, c_N \geq 0$, where $g(\mathrm{d}t) = \sum_{j=0}^{N}c_j \delta_{t_j}(\mathrm{d}t)$ with $\sum_{j=0}^{N}c_j = 1$.

\section{Pricing \texorpdfstring{$g$}{g}-monitored Asian Options}\label{sec:pricing}

\subsection{Semi-analytic pricing formulas}\label{sec:semi-closed-pricing}

We provide a model-independent pricing formula for $g$-monitored geometric Asian options, which covers the Volterra--Heston model as a particular case. Let $T > 0$ and let $(S_t)_{t \in [0,T]}$ be a continuous semimartingale on $(\Omega, \mathcal{F}, (\mathcal{F}_t)_{t \in [0,T]}, \mathbb{P})$ such that $\mathbb{P}(S_t > 0, \ \forall t \in [0,T]) = 1$, and $(\E^{-r t}S_t)_{t\in[0,T]}$ is a $\mathbb{P}$-martingale. Let $g$ be a Borel probability measure on $[0,T]$, and recall that $G_t, G_{t,T}$ are the geometric averages defined in~\eqref{eq:Gg}. For $(s,w)\in\mathbb{C}^2$, let $\Psi_t(s,w) := \mathbb{E}\left[ G_{t,T}^s S_T^w \,\middle|\,\mathcal{F}_t \right]$ be given as in \eqref{eq:psit}, whenever the conditional expectation is well-defined. In the Volterra--Heston model, this transform is given explicitly by Corollary~\ref{thm:condcfG_nS_Texpression}.

We derive a unified pricing formula for a class of payoffs that contains fixed- and floating-strike geometric Asian calls and puts as particular examples. Let $\mathcal{S}_R := \{R+\mathrm{i}y:y\in\mathbb{R}\}$ denote the vertical line with fixed real part $R\in(0,1)$. We denote by $\lambda_R$ the Lebesgue measure along $\mathcal{S}_R$, that is, $\int_{\mathcal{S}_R} f(z)\,\lambda_R(\mathrm{d}z) := \int_{\mathbb{R}} f(R+\mathrm{i}y)\,\mathrm{d}y$ whenever the integral is well-defined. Suppose that the payoff function $h:\mathbb{R}_+^2\longrightarrow\mathbb{R}$ admits a representation of the form
\begin{align}\label{eq:bromwich-general-payoff}
    h(x,y) &= a+bx+cy + \int_{\mathcal{S}_R}x^z\,\zeta_0(\mathrm{d}z) + \int_{\mathcal{S}_R}x^z y^{1-z}\,  \zeta_1(\mathrm{d}z), \qquad x,y>0,
\end{align}
where $a,b,c\in\mathbb{R}$ and $\zeta_0,\zeta_1$ are finite complex Borel measures on $\mathcal{S}_R$. We assume that the integrals on the right-hand side of \eqref{eq:bromwich-general-payoff} are convergent and real-valued. The corresponding geometric Asian payoff is given by $h(G_T,S_T)$, and contains the usual fixed- and floating-strike options as demonstrated in the next example.

\begin{example}\label{example:putcall}
 Define
\begin{align}\label{eq:zeta-fixed-floating}
    \zeta_K(\mathrm{d}z) = \frac{K^{1-z}}{2\pi z(z-1)} \lambda_R(\mathrm{d}z) \quad \text{ and } \quad  \zeta_{\mathrm{fl}}(\mathrm{d}z) = \frac{1}{2\pi z(z-1)} \lambda_R(\mathrm{d}z).
\end{align}
Since $R\in(0,1)$ and $|K^{1-z}| = K^{1-R}$ on $\mathcal{S}_R$, the densities in \eqref{eq:zeta-fixed-floating} have no poles on $\mathcal{S}_R$ and decay as $|\operatorname{Im}(z)|^{-2}$ when $\operatorname{Im}(z) \to \infty$. Thus both $\zeta_K$ and $\zeta_{\mathrm{fl}}$ are finite complex measures. A standard Mellin--Bromwich inversion, or equivalently a residue calculation, gives
\begin{align}\label{eq:bromwich-min-fixed}
    \int_{\mathcal{S}_R} x^z\,\zeta_K(\mathrm{d}z) = -\min\{x,K\} \ \text{ and } \ \int_{\mathcal{S}_R} x^z y^{1-z}\, \zeta_{\mathrm{fl}}(\mathrm{d}z) = -\min\{x,y\}.
\end{align}
Hence the four standard geometric Asian payoffs arise from
\eqref{eq:bromwich-general-payoff} as follows:
\begin{center}
\begin{tabular}{c|c|c|c|c|c}
    Payoff & $a$ & $b$ & $c$ & $\zeta_0$ & $\zeta_1$
    \\ \hline $(x-K)^+$ & $0$ & $1$ & $0$ & $\zeta_K$ & $0$
    \\ $(K-x)^+$ & $K$ & $0$ & $0$ & $\zeta_K$ & $0$ 
    \\ $(x-y)^+$ & $0$ & $1$ & $0$ & $0$ & $\zeta_{\mathrm{fl}}$ 
    \\ $(y-x)^+$ & $0$ & $0$ & $1$ & $0$ & $\zeta_{\mathrm{fl}}$
\end{tabular}
\end{center}
In particular, setting $(x,y)=(G_T,S_T)$ yields respectively the fixed-strike call, fixed-strike put, floating-strike call, and floating-strike put.
\end{example}

Given $h$ as in \eqref{eq:bromwich-general-payoff}, the risk-neutral price of the payoff $h(G_T, S_T)$ at time $t\in[0,T]$ is given by
\begin{align}\label{eq:general-payoff-price-definition}
    H_t^h := \mathrm{e}^{-r(T-t)} \mathbb{E}\left[ h(G_T,S_T) \,\middle|\, \mathcal{F}_t \right].
\end{align}
The next theorem provides a semi-closed pricing formula. Its proof is given in Appendix~\ref{sec:proof-pricing-formula}.

\begin{theorem}\label{thm:general-geometric-payoff-pricing}
    Let $R\in(0,1)$ and let $h$ admit the representation \eqref{eq:bromwich-general-payoff}, where $\zeta_0$ and $\zeta_1$ are finite complex Borel measures on $\mathcal{S}_R$. Then, for every $t\in[0,T]$,
    \begin{align}\label{eq:general-geometric-payoff-transform}
        H_t^h &= \mathrm{e}^{-r(T-t)}a + \mathrm{e}^{-r(T-t)} bG_t\Psi_t(1,0) + cS_t
        \nonumber
        \\ &\qquad + \mathrm{e}^{-r(T-t)} \int_{\mathcal{S}_R} G_t^z\Psi_t(z,0)\, \zeta_0(\mathrm{d}z) + \mathrm{e}^{-r(T-t)} \int_{\mathcal{S}_R} G_t^z\Psi_t(z,1-z)\, \zeta_1(\mathrm{d}z).
    \end{align}
    The two integrals in \eqref{eq:general-geometric-payoff-transform} are absolutely convergent almost surely.
\end{theorem}

\subsection{Convergence to continuously monitored Asian options}\label{secconvtocont}

In this section, we study the convergence of discretely monitored Asian call options 
to their continuous counterparts. Let $\mathcal{P}_T^N=\{t_0,t_1,\dots,t_N\}$ with $0=t_0<t_1<\dots<t_N=T$ be a monitoring grid and define its mesh size by 
\[
    \Delta_N:=\max_{j\in\{0,\dots,N-1\}}(t_{j+1}-t_j).
\]
Define the weights
\[
    \omega_j^N:=t_{j+1}-t_j,\qquad j=0,\dots,N-1, \qquad \omega_N^N:=t_N-t_{N-1},
\]
and the normalising constant
\[
    T_N:=\sum_{j=0}^N\omega_j^N=T+\omega_N^N.
\]
The corresponding probability measure of monitored time points is
\begin{align}\label{eq:gN-delta-general}
    g_N^\Delta(\mathrm{d}t) := \frac{1}{T_N}\sum_{j=0}^N\omega_j^N\delta_{t_j}(\mathrm{d}t).
\end{align}
The additional terminal weight is included so that on an equidistant grid \(g_N^\Delta\) coincides with the standard equal-weight monitoring measure introduced in the introduction. For $t<T$, recall that $n_t\in\{0,\dots,N-1\}$ is defined by $t_{n_t}\leq t<t_{n_t+1}$, and $n_T=N$. Integration with respect to $g_N^\Delta$ is then given by
\[
    \int_{(t,T]}f(r)\,g_N^\Delta(\mathrm{d}r) = \frac{1}{T_N}\sum_{j=n_t+1}^{N}\omega_j^N f(t_j),
    \qquad \int_{[0,t]}f(r)\,g_N^\Delta(\mathrm{d}r) = \frac{1}{T_N}\sum_{j=0}^{n_t}\omega_j^N f(t_j).
\]
Let $G_T^N$ and $G_T^c$ denote the discretely and continuously sampled
geometric averages
\[
    G_T^N := \exp\left( \frac{1}{T_N}\sum_{j=0}^{N}\omega_j^N\log(S_{t_j}) \right),
    \qquad G_T^c := \exp\left( \frac{1}{T}\int_0^T\log(S_t)\,\mathrm{d}t \right).
\]
Remark that, for an equidistant grid $t_j = j \frac{T}{N}$ we obtain $\Delta_N = \frac{T}{N}$, $\omega_j^N = \frac{T}{N}$, $T_N = T\left( 1 + \frac{1}{N}\right)$. The next proposition provides the first convergence result from discrete to continuously averaged Asian payoffs.

\begin{theorem}\label{thm:general-price-process-convergence}
    Let $(\mathcal P_T^N)_{N\geq1}$ be a sequence of monitoring grids
    such that $\Delta_N\to0$, and let $g_N^\Delta$ be defined by
    \eqref{eq:gN-delta-general}. Assume that $(\mathrm e^{-rt}S_t)_{t\in[0,T]}$ is a square-integrable martingale. Let $R\in(0,1)$ and let $h$ have representation \eqref{eq:bromwich-general-payoff}. Define
    \[
        H_t^{h,N} := \mathrm e^{-r(T-t)} \mathbb E\left[ h(G_T^N,S_T) \,\middle|\, \mathcal F_t
        \right] \ \text{ and } \ H_t^{h,c} := \mathrm e^{-r(T-t)} \mathbb E\left[ h(G_T^c,S_T) \,\middle|\, \mathcal F_t \right].
    \]
    Then 
        \begin{align}\label{eq:price-process-L2-convergence}
            \mathbb E\left[ \sup_{t\in[0,T]} |H_t^{h,N}-H_t^{h,c}|^2 \right] \longrightarrow 0.
        \end{align}
\end{theorem}

While the previous theorem provides a general convergence result, below we provide a convergence rate of order $\Delta_N$ under minor additional assumptions.

\begin{theorem}\label{thm:lipschitz-payoff-convergence-rate}
    Let $(\mathcal P_T^N)_{N\geq1}$ be a sequence of monitoring grids with mesh size $\Delta_N\to0$, and let $g_N^\Delta$ be defined by~\eqref{eq:gN-delta-general}. Suppose that $\exists q>1$ such that
    \[
        \quad \mathbb E[S_T^{2q}]<\infty.
    \]
    Set $X_t:=\log S_t$ and suppose that $X$ admits the semimartingale decomposition $X_t=X_0+A_t+M_t$, where $A$ is a continuous finite-variation process with $A_0=0$ and $M$ is a continuous local martingale with $M_0=0$, such that
    \begin{align}\label{eq:log-semimartingale-rate-assumption}
        \big\||A|_T\big\|_{L^{2q}(\Omega)} + \big\|\langle M\rangle_T^{1/2}\big\|_{L^{2q}(\Omega)} < \infty.
    \end{align}
    Let $h:\mathbb R_+^2\longrightarrow \mathbb R$ satisfy
    \begin{align}\label{eq:lipschitz-first-variable}
        |h(x,y)-h(x',y)| \leq L_h|x-x'|, \qquad x,x',y>0,
    \end{align}
    for some $L_h \in (0,\infty)$, and assume that $h(1,S_T)\in L^q(\Omega)$. Then there exists a constant $\Lambda_{q,h}(T)>0$, independent of $N$, such that
    \begin{align}\label{eq:lipschitz-price-process-rate}
        \mathbb E\left[ \sup_{t\in[0,T]} |H_t^{h,N}-H_t^{h,c}|^q \right] \leq \Lambda_{q,h}(T)\Delta_N^q.
    \end{align}
\end{theorem}

The proofs of Theorem~\ref{thm:general-price-process-convergence} and Theorem~\ref{thm:lipschitz-payoff-convergence-rate} are given in Appendix~\ref{sec:proof-discrete-to-continuous}.
For the Volterra--Heston model, the assumption $S_T \in L^{2q}(\Omega)$ is related to moment-explosions. Such a problem was studied in \cite{gerhold2019moment} for the rough Heston model, and in \cite{MR4099324} for the Volterra--Heston model. In particular, when $\mathcal{K}(t)= t^{\alpha-1}/\Gamma(\alpha)$ with $\alpha \in (\frac{1}{2},1)$, then $S_T \in L^{2q}(\Omega)$ for each $T > 0$ is equivalent to 
\begin{align}\label{eq: moment}
     \kappa - 2\rho \sigma q \geq \sigma \sqrt{2q(2q - 1)} \ \ \Longleftrightarrow \ \ q \leq \frac{1 - 2\rho \frac{\kappa}{\sigma} + \sqrt{1 - 4\rho \frac{\kappa}{\sigma} + 4 \left( \frac{\kappa}{\sigma}\right)^2} }{4(1- \rho^2) },
\end{align}
with the last equivalence valid for $|\rho| < 1$. For the remaining integrability assumptions, note that $A_t = rt - \frac{1}{2}\int_0^t \nu_s\, \mathrm{d}s$ and $\langle M \rangle_t = \int_0^t \nu_s\, \mathrm{d}s$. Hence \eqref{eq:log-semimartingale-rate-assumption} follows from the known moment bound on the variance process $\nu$.

\subsection{Numerical study}\label{sec:numericspricing}

 We validate the fixed-strike geometric-Asian pricing formula in
Theorem~ \ref{thm:general-geometric-payoff-pricing} using the classical Heston case, which permits
an accurate independent simulation benchmark.  The parameters are
\[
 \kappa=1,\qquad \theta=\nu_0=0.04,\qquad \sigma=0.3,\qquad \rho=-0.7,
 \qquad S_0=K=100,\qquad T=1,\qquad r=0,
\]
and the contract is monitored at the 21 equidistant dates \(t_0,\ldots,t_{20}\).

The implementation follows representation of Example~\ref{example:putcall}. For \(0<R<1\),
\[
 (G_T-K)^+=G_T-\min(G_T,K),  \qquad  \min(G_T,K)=-\int_{\mathcal S_R}G_T^z\,\zeta_K(\mathrm dz).
\]
Consequently, both the price and the hedge contain the block $G_T$ in addition to the
contour integral. It is possible to omit this contribution, provided that integration is carried out over the line \(R>1\) with adjusted measure $\zeta_K$. The latter is justified by the Residue theorem, since one passes a singularity at $z = 1$. 

The Volterra--Riccati equation reduces to an ordinary Riccati equation in the Heston benchmark and is solved by classical fourth-order Runge--Kutta on a grid aligned with all monitoring dates. On each interval, the left value of the piecewise-constant monitoring coefficient is used, while the continuous Riccati solution is integrated at interval midpoints. The accepted
transform controls are 384 Gauss--Legendre nodes on \(|\operatorname{Im}z|\leq60\) and 1,000
Runge--Kutta steps.  They are compared with 512 nodes on \(|\operatorname{Im}z|\leq80\) and 1,400 steps.  Across \(R\in\{0.25,0.50,0.75\}\), the finer-grid prices differ by at most \(1.6\times10^{-7}\) and the finer-grid hedge ratios by at most \(3.0\times10^{-9}\). Relative to the accepted grid, all price differences are below  \(1.2\times10^{-5}\) and all hedge differences below \(8.1\times10^{-6}\).

 For the Monte Carlo benchmark, the CIR variance transition is sampled exactly from its
noncentral-\(\chi^2\) law.  Integrated variance is approximated by the trapezoidal rule, the
variance-driver contribution to the stock is reconstructed from the CIR increment, and the
orthogonal stock component is simulated independently.  Thus the variance transition is exact,
whereas the joint stock-variance path remains time discretised.  The price study uses 200,000
paths and fixed random seeds.  Parentheses in Table~\ref{tab:pricing-primary-validation} are 95\%
Monte Carlo confidence-interval half-widths.

\begin{table}[htbp]
\centering
 \small
\begin{tabular}{lrrr}
\toprule
check &  benchmark &  transform/GKW &  absolute gap \\
 \midrule
Heston MC (1000 steps) & 4.253520 (0.025055) & 4.247377 & 6.14e-03 \\
Heston MC (2000 steps) & 4.263444 (0.025137) & 4.247377 & 1.61e-02 \\
Constant variance (price) & 4.36884811 & 4.36884811 & 4.55e-11 \\
Constant variance (delta) & 0.48453809 & 0.48453809 & 3.04e-13 \\
GKW finite difference ($h=0.0100$) & 0.45268194 & 0.45298997 & 3.08e-04 \\
GKW finite difference ($h=0.0050$) & 0.45291288 & 0.45298997 & 7.71e-05 \\
GKW finite difference ($h=0.0025$) & 0.45297069 & 0.45298997 & 1.93e-05 \\
\bottomrule
\end{tabular}
\caption{Validation of the pricing and time-zero hedging implementation. Both Heston confidence
intervals contain the transform price. The constant-variance rows compare the raw integral and its delta with the analytic discrete geometric-Asian Black--Scholes values. The final rows compare the transform GKW hedge with \(\partial_{S_0}C_0+(\rho\sigma/S_0)\partial_{\nu_0}C_0\), using centred finite differences and holding the first fixing fixed. The fourfold reduction in the gap when \(h\) is halved is consistent with the expected second-order accuracy of centred differences.}
 \label{tab:pricing-primary-validation}
\end{table}

 The constant-variance price and delta errors are approximately \(4.6\times10^{-11}\) and
\(3.0\times10^{-13}\), respectively. In the Heston benchmark the
transform price lies inside every reported confidence interval, including the 500- and 4,000-step
checks recorded in the computational release.  The time-discretisation study also compares hedge
MSE at 1,000 and 2,000 simulation steps over 20 independent replications of 20,000 paths.  A
two-one-sided equivalence test on the log-MSE contrast gives a 90\% ratio interval
\([0.99968,1.00758]\), entirely inside the prespecified two-percent equivalence margin.  These
checks validate the numerical implementation at the stated parameter values; they do not replace
the analytic assumptions of the pricing and hedging theorems.

\section{Variance-Optimal Hedging of Discretely-Monitored Geometric Asian Options}\label{sec:hedging}

\subsection{Variance-optimal hedge in transform coordinates}
\label{subsec:volterra-block-hedge}

In this section, we work on the filtered probability space $(\Omega,\mathcal F,\mathbb F,\mathbb P)$ as in Section~\ref{subsec:volterra-hestonsv}, with $r=0$, and suppose that the stock price $S$ is a strictly positive continuous square-integrable martingale. Sufficient conditions for the square-integrability have been studied in \cite{gerhold2019moment, MR4099324}. Let $\mathcal G^2(S)$ denote the space of predictable processes $\vartheta$ with finite norm
\[
    \|\vartheta\|_{\mathcal G^2(S)} := \left( \mathbb E\left[ \int_0^T |\vartheta_t|^2\,\mathrm d\langle S\rangle_t \right] \right)^{1/2}.
\]
Let $\mathcal G^2_{\mathbb C}(S)$ denote its complexification, and write $(\vartheta\cdot S)_t := \int_0^t\vartheta_s\,\mathrm dS_s$ for the stochastic integral of $\vartheta$ against $S$.

Let $g$ be an arbitrary Borel probability measure on $[0,T]$, and recall that $G_T, G_{t,T}$ are defined in \eqref{eq:Gg}. We consider a payoff $h:\mathbb R_+^2\longrightarrow \mathbb R$ admitting the representation~\eqref{eq:bromwich-general-payoff}. As before, we assume that the right-hand side of~\eqref{eq:bromwich-general-payoff} is real-valued. The risk-neutral price martingale of the payoff is given by \eqref{eq:general-payoff-price-definition}, i.e.
\begin{align}\label{eq:Hh-hedging}
    H_t^h := \mathbb E\left[ h(G_T,S_T) \,\middle|\, \mathcal F_t \right].
\end{align}
Since $S$ is square-integrable, it follows that $\mathbb E[G_T^2] \leq \int_{[0,T]}\mathbb E[S_t^2]\,g(\mathrm dt) \leq \mathbb E[S_T^2] <\infty$. Since $\zeta_0$ and $\zeta_1$ have finite total variation, \eqref{eq:bromwich-general-payoff} implies that $h(G_T,S_T)\in L^2(\Omega)$. Hence $H^h$ is a square-integrable martingale.

The variance-optimal hedging problem consists in finding an initial capital $v\in\mathbb R$ and a strategy $\vartheta\in\mathcal G^2(S)$ minimizing
\begin{align}\label{eq:vo-problem}
    \epsilon_h := \inf_{v\in\mathbb R,\, \vartheta\in\mathcal G^2(S)} \mathbb E\left[ \left( v+(\vartheta\cdot S)_T-h(G_T,S_T) \right)^2 \right].
\end{align}
In incomplete markets, this problem dates back to F\"ollmer--Sondermann~\cite{Foellmer1985} and is characterized through the Galtchouk--Kunita--Watanabe decomposition. More precisely, there exists $\vartheta^{*,h}$ and a square-integrable martingale $L^h$ such that
\[
    H_t^h = H_0^h + (\vartheta^{*,h}\cdot S)_t + L_t^h, \qquad \langle L^h,S\rangle\equiv0.
\]
In particular, it follows that $v^{*,h}=H_0^h$ and $\epsilon_h = \mathbb E[(L_T^h)^2]$. Since $S$ is continuous, the optimal strategy is characterized by
\begin{align}\label{eq:var-opt-rn-deriv}
    \vartheta_t^{*,h} = \frac{\mathrm d\langle H^h,S\rangle_t}{\mathrm d\langle S\rangle_t}.
\end{align}

While this formally gives a solution to the variance-optimal hedging problem, it is not very useful for practical implementations. Below we derive a semi-explicit representation of this strategy in the Volterra--Heston model. Our approach extends \cite{DT2017, Kallsen2010} towards affine Volterra processes with particular emphasis on Asian payoffs. Let
\[
    \mathcal D := \left\{ (s,w)\in\mathbb C^2: \Re(s),\Re(w)\geq0,\quad \Re(s)+\Re(w)\leq1 \right\},
\]
and, for $(s,w)\in\mathcal D$, define the elementary pricing martingale
\begin{align}\label{eq:Hsw-hedging}
    H_t(s,w) &:= \mathbb E\left[ G_T^sS_T^w \,\middle|\, \mathcal F_t \right] = G_t^s\Psi_t(s,w),
\end{align}
where $\Psi_t$ is given by Corollary \ref{thm:condcfG_nS_Texpression}. In particular, for every $z\in\mathcal S_R$, we find $(z,0),(z,1-z), (1,0) \in\mathcal D$. Thus all arguments entering \eqref{eq:bromwich-general-payoff} lie in the domain of Corollary~\ref{thm:condcfG_nS_Texpression}. For $(s,w)\in\mathcal D$, let
\[
    \psi_1^{s,w}(u) := w+s\,g([T-u,T]) \quad \text{ and } \quad \psi_{1,-}^{s,w}(u) := w+s\,g((T-u,T]),
\]
and let $\psi_2^{s,w}$ denote the solution of the Volterra--Riccati equation \eqref{eq:VolRic} with $\mu_2 = 0$. Let $A_t^g := \int_{[0,t]}\log(S_y)\,g(\mathrm dy)$ denote the accumulated observed part of the geometric average. In terms of the forward variance $\xi_t(u) := \mathbb E[\nu_u\mid\mathcal F_t]$, $u\in[t,T]$, Corollary~\ref{thm:condcfG_nS_Texpression} gives
\begin{align}\label{eq:Y-general-hedge}
    H_t(s,w) &= \exp\bigl(Y_t(s,w)\bigr),
\end{align}
where
\begin{align}
    Y_t(s,w) &= sA_t^g + \psi_{1,-}^{s,w}(T-t)\log(S_t) \nonumber - \frac12 \int_t^T \psi_{1,-}^{s,w}(T-x) \xi_t(x)\, \mathrm dx 
    \\ &\quad + \frac12 \int_t^T \Big( \psi_1^{s,w}(T-y)^2 + 2\rho\sigma \psi_1^{s,w}(T-y) \psi_2^{s,w}(T-y) + \sigma^2 \psi_2^{s,w}(T-y)^2 \Big) \xi_t(y)\, \mathrm dy. \label{eq:Y-general-hedge-expanded}
\end{align}
For later use, define
\begin{align}\label{eq:beta-sw-hedge}
    \beta_t(s,w) := \psi_{1,-}^{s,w}(T-t) + \rho\sigma \psi_2^{s,w}(T-t) 
    = w+s\,g((t,T]) + \rho\sigma \psi_2^{s,w}(T-t).
\end{align}

The next theorem provides a representation of the solution to the mean-variance problem in terms of $H_t$ and $\beta_t$. The latter can be computed semi-analytically in terms of $\psi_2$. 

\begin{theorem}[Variance-optimal hedge for general $g$-monitored geometric Asian payoffs]
\label{thm:optimal-strategy-geo}
Let $g$ be a Borel probability measure on $[0,T]$, let $R\in(0,1)$, and let $h$ admit the representation \eqref{eq:bromwich-general-payoff}. Then the variance-optimal strategy for the payoff $h(G_T,S_T)$ is
\begin{align}
    \vartheta_t^{*,h} &= c + b\, \frac{H_t(1,0)}{S_t} \beta_t(1,0) + \int_{\mathcal S_R} \frac{H_t(z,0)}{S_t} \beta_t(z,0)\, \zeta_0(\mathrm dz) + \int_{\mathcal S_R} \frac{H_t(z,1-z)}{S_t} \beta_t(z,1-z)\, \zeta_1(\mathrm dz). \label{eq:theta-star-general-g}
\end{align}
The optimal initial capital is
\begin{align}\label{eq:optimal-capital-general-g}
    v^{*,h} = a + bH_0(1,0) + cS_0 + \int_{\mathcal S_R} H_0(z,0)\, \zeta_0(\mathrm dz) + \int_{\mathcal S_R} H_0(z,1-z)\, \zeta_1(\mathrm dz).
\end{align}
Moreover, the minimal mean-squared hedging error is given by
\begin{align}\label{eq:hedging-error-general}
    \epsilon_h = (1-\rho^2)\sigma^2 \int_0^T \mathbb E\left[ \nu_t |Q_t^h|^2 \right] \mathrm dt,
\end{align}
with 
\begin{align}\label{eq:Qh-general-g}
    Q_t^h &:= bH_t(1,0)\psi_2^{1,0}(T-t) + \int_{\mathcal S_R} H_t(z,0) \psi_2^{z,0}(T-t)\, \zeta_0(\mathrm dz) + \int_{\mathcal S_R} H_t(z,1-z) \psi_2^{z,1-z}(T-t)\, \zeta_1(\mathrm dz).
\end{align}
\end{theorem}

The four standard geometric Asian options are obtained from Theorem~\ref{thm:optimal-strategy-geo} by the choices given in the table of Section~\ref{sec:semi-closed-pricing}. Formula~\eqref{eq:theta-star-general-g} shows that the monitoring scheme enters through the remaining monitoring mass $g((t,T])$. For a discrete monitoring measure this quantity is piecewise constant. Roughness enters through both the conditional transforms $H_t(s,w)$ and the Volterra--Riccati solutions $\psi_2^{s,w}$, while leverage contributes the correction $\rho\sigma\psi_2^{s,w}(T-t)$.
The proof of Theorem~\ref{thm:optimal-strategy-geo} is given in Appendix~\ref{sec:proof-optimal-strategy-geo}.

\subsection{Hedging in finite-factor Markovian lifts}
\label{subsec:markapprox}

For numerical implementations, we approximate the Volterra kernel by a finite sum of exponentials. Throughout this subsection, the superscript $n\geq1$ denotes the number of factors in the Markovian lift and is kept distinct from the monitoring-grid index used in Section~\ref{secconvtocont}. Suppose that the kernel $\mathcal K$ is completely monotone. By the Hausdorff--Bernstein--Widder  theorem, it admits the representation
\[
    \mathcal K(t) = \int_{[0,\infty)} \mathrm e^{-xt}\, k(\mathrm dx), \qquad t>0,
\]
for some nonnegative Borel measure $k$. Following the finite-factor lifting approach, see, e.g., \cite{abi2019lifting,AK21, Bayer2023}, we approximate $k$ by a discrete measure. 

\begin{assumption}\label{assumption:nodes-and-weights}
For each $n\geq1$, choose weights and nodes $(w_i^n)_{i=1,\dots, n}, (x_i^n)_{i=1,\dots, n} \subset \mathbb{R}_+$ with $\sum_{i=1}^n w_i^n > 0$ such that the approximation 
\begin{align}\label{eq:finite-factor-kernel}
    \mathcal K^n(t) := \sum_{i=1}^n w_i^n\mathrm e^{-x_i^n t}
\end{align}
satisfies 
\[
    \int_0^T |\mathcal{K}(t) - \mathcal{K}^n(t)|^2\, \mathrm{d}t \longrightarrow 0, \qquad n \to \infty.
\]
\end{assumption}

Remark that $\mathcal{K}^n$ is completely monotone and smooth. In particular it satisfies the standing assumption \ref{assump:K}. We next describe the corresponding Markovian lift of the Volterra--Heston model. Recall that $W_t^\nu = \rho W_t^S + \sqrt{1-\rho^2}\,W_t^I$, where $W^S$ and $W^I$ are independent Brownian motions. For fixed $n\geq1$, let $(S^n,\nu^n)$ solve
\begin{align}
    \mathrm dS_t^n &= S_t^n\sqrt{\nu_t^n}\, \mathrm dW_t^S, \label{eq:lifted-S}
    \\ \nu_t^n &= \nu_0 + \int_0^t \mathcal K^n(t-s) \kappa(\theta-\nu_s^n)\, \mathrm ds + \int_0^t \mathcal K^n(t-s) \sigma\sqrt{\nu_s^n}\, \mathrm dW_s^\nu. \label{eq:lifted-volterra-variance}
\end{align}

Let $g$ be an arbitrary Borel probability measure on $[0,T]$, and set
\begin{align}\label{eq:lifted-general-geometric-average}
    G_{t,T}^n = \exp\left( \int_{(t,T]} \log(S_u^n)\, g(\mathrm{d}u) \right) \ \text{ and } \ G_T^n = \exp\left( \int_{[0,T]} \log(S_u^n)\,g(\mathrm du) \right).
\end{align}
For $(s,w)\in\mathcal D$, recall that $\psi_1^{s,w}(u) := w+s\,g([T-u,T])$ and that $R(p,q) := \frac12(p^2-p) + (\rho\sigma p-\kappa)q + \frac12\sigma^2q^2$ was defined in \eqref{eq:R}. Let $\psi_2^{n;s,w}$ be the unique solution of 
\begin{align}\label{eq:lifted-Volterra-Riccati-equivalence}
    \psi_2^{n;s,w}(u) = \int_0^u \mathcal K^n(u-v) R\left( \psi_1^{s,w}(v), \psi_2^{n;s,w}(v) \right) \mathrm dv.
\end{align}
Then, according to Corollary \ref{thm:condcfG_nS_Texpression}, we find for 
\begin{align}\label{eq:lifted-Hsw-definition}
    \Psi_t^n(s,w) &:= \mathbb E\left[ (G_{t,T}^n)^s(S_T^n)^w \,\middle|\, \mathcal F_t \right]
\end{align}
the affine representation
\begin{align*}
    \Psi_t^n(s,w) &= \exp\Bigg( \log(S^n_t) \psi_{1,-}^{s,w}(T-t) - \frac{1}{2} \int_t^T \psi_{1,-}^{s,w}(T-x)\xi^n_t(x)\, \mathrm{d}x
    \\ &\qquad + \frac{1}{2}\int_t^T \left( \psi_1^{s,w}(T-y)^2 + 2\rho \sigma \psi_1^{s,w}(T-y)\psi_2^{n;s,w}(T-y) + \sigma^2 \psi_2^{n;s,w}(T-y)^2\right)\xi_t^n(y)\, \mathrm{d}y\Bigg),
\end{align*}
where we have set $\psi_{1,-}^{s,w}(t) = w + s g((T-t,T])$ and $\xi^n_t(x) = \mathbb{E}[ \nu_x^n \ | \ \mathcal{F}_t]$. 

The above representation appears to be infinite-dimensional due to the presence of the forward-variance curve $\xi_t^n$. By expressing this curve in terms of finitely many factors, we provide an alternative representation that is more suitable for numerical implementations. Let $Y^1,\dots, Y^n$ be the solution of
\begin{align}\label{eq:lifted-factors-definition}
    Y_t^{n,i} := \int_0^t \mathrm e^{-x_i^n(t-s)} \kappa(\theta-\nu_s^n)\, \mathrm ds + \int_0^t \mathrm e^{-x_i^n(t-s)} \sigma\sqrt{\nu_s^n}\, \mathrm dW_s^\nu.
\end{align}
Then
\begin{align}\label{eq:variance-factor-representation}
    \nu_t^n = \nu_0 + \sum_{i=1}^n w_i^nY_t^{n,i},
\end{align}
and the factors satisfy
\begin{align}\label{eq:factor-SDE}
    \mathrm dY_t^{n,i} &= \left( - x_i^nY_t^{n,i} + \kappa(\theta-\nu_t^n) \right)\mathrm dt + \sigma\sqrt{\nu_t^n}\, \mathrm dW_t^\nu, \qquad Y_0^{n,i}=0.
\end{align}
In particular, setting $X^n = \log S^n$,  $( X_t^n, Y_t^{n,1}, \ldots, Y_t^{n,n}\bigr)_{t\in[0,T]}$ is an $(n+1)$-dimensional affine Markovian realization of the finite-factor Volterra--Heston model. Existence of an admissible lifted solution and nonnegativity of $\nu^n$ follow under the standard finite-factor admissibility conditions, see e.g.~\cite[Appendix~A.1]{abi2019lifting}.
    
\begin{proposition}\label{prop:finite-dimensional-Riccati-lift}
    Fix $n \geq 1$ and $(s,w) \in \mathcal{D}$. Then there exists a unique absolutely continuous solution of
    \begin{align}\label{eq:lifted-Riccati-eta}
        \frac{\mathrm d}{\mathrm du} \eta_i^{n;s,w}(u) &= - x_i^n\eta_i^{n;s,w}(u) + R\left( \psi_1^{s,w}(u), \psi_2^{n;s,w}(u) \right), \qquad \eta_i^{n;s,w}(0) = 0, \ \ \text{a.e.} 
    \end{align}
    where $i = 1,\dots, n$, and
    \begin{align}\label{eq:lifted-psi2-factor-sum}
        \psi_2^{n;s,w}(u) = \sum_{i=1}^n w_i^n\eta_i^{n;s,w}(u).
    \end{align}
    Moreover, let $\phi^{n;s,w}$ be the unique absolutely continuous solution of
    \begin{align}\label{eq:lifted-Riccati-phi}
        \frac{\mathrm d}{\mathrm du} \phi^{n;s,w}(u) &= \kappa\theta\, \psi_2^{n;s,w}(u) + \nu_0 R\left(\psi_1^{s,w}(u), \psi_2^{n;s,w}(u) \right), \qquad \phi^{n;s,w}(0)=0.
    \end{align}
    Then the conditional transform of $H_t^n(s,w) := \mathbb E\left[ (G_{T}^n)^s(S_T^n)^w \,\middle|\, \mathcal F_t \right]$ is given by
    \begin{align}\label{eq:lifted-affine-transform}
        H_t^n(s,w) &= \exp\Bigg( sA_t^{g,n} + \psi_{1,-}^{s,w}(T-t)X_t^n + \phi^{n;s,w}(T-t) + \sum_{i=1} ^n w_i^n \eta_i^{n;s,w}(T-t) Y_t^{n,i} \Bigg)
    \end{align}
    where $A_t^{g,n} = \int_{[0,t]}X_u^n\, g(\mathrm{d}u)$
\end{proposition}

The proof is given in Appendix~\ref{sec:proof-finite-dimensional-Riccati-lift}.
Define 
\begin{align}\label{eq:lifted-beta}
    \beta_t^n(s,w) := \psi_{1,-}^{s,w}(T-t) + \rho\sigma \psi_2^{n;s,w}(T-t) 
    = w+s\,g((t,T]) + \rho\sigma \psi_2^{n;s,w}(T-t).
\end{align}
The following is a particular case of Theorem \ref{thm:optimal-strategy-geo} applied for $\mathcal{K}^n$.

\begin{theorem}[Variance-optimal hedge for multi-factor approximation]\label{thm:lifted-general-hedge}
Let $g$ be a Borel probability measure on $[0,T]$, let $R\in(0,1)$, and let $h$ admit the representation \eqref{eq:bromwich-general-payoff}. If $S^n$ is a square-integrable martingale, then the variance-optimal strategy for the payoff $h(G^n_T,S^n_T)$ is
\begin{align}
    \vartheta_t^{*,h,n} &= c + b\, \frac{H^n_t(1,0)}{S^n_t} \beta^n_t(1,0) + \int_{\mathcal S_R} \frac{H^n_t(z,0)}{S^n_t} \beta^n_t(z,0)\, \zeta_0(\mathrm dz) + \int_{\mathcal S_R} \frac{H^n_t(z,1-z)}{S^n_t} \beta^n_t(z,1-z)\, \zeta_1(\mathrm dz). 
\end{align}
The optimal initial capital is
\begin{align*}
    v^{*,h,n} = a + bH^n_0(1,0) + cS^n_0 + \int_{\mathcal S_R} H^n_0(z,0)\, \zeta_0(\mathrm dz) + \int_{\mathcal S_R} H^n_0(z,1-z)\, \zeta_1(\mathrm dz).
\end{align*}
Moreover, the minimal mean-squared hedging error is given by
\begin{align*}
    \epsilon_{h,n} = (1-\rho^2)\sigma^2 \int_0^T \mathbb E\left[ \nu^n_t |Q_t^{h,n}|^2 \right] \mathrm dt,
\end{align*}
with 
\begin{align}\label{eq:Qh-general-g-lifted}
    Q_t^{h,n} &:= bH^n_t(1,0)\psi_2^{n;1,0}(T-t) + \int_{\mathcal S_R} H^n_t(z,0) \psi_2^{n;z,0}(T-t)\, \zeta_0(\mathrm dz) 
    \\ \nonumber &\qquad \qquad + \int_{\mathcal S_R} H^n_t(z,1-z) \psi_2^{n;z,1-z}(T-t)\, \zeta_1(\mathrm dz).
\end{align}
\end{theorem}

\subsection{Convergence of the finite-factor hedge}
\label{subsec:optimality}

We next study the convergence of the variance-optimal hedge in the
finite-factor models to its Volterra counterpart. Throughout this
subsection, the factor index is denoted by $n$, and the lifted and
limiting models are constructed on the same filtered probability space
and driven by the same Brownian motions $(W^S,W^I)$. As in the preceding
hedging sections, the short rate is zero.

As before, let $g$ be an arbitrary Borel probability measure on $[0,T]$ and let $G_T^n, G_{t,T}^n$, $G_T, G_{t,T}$ be the corresponding geometric $g$-weighted averages. Let $R\in(0,1)$ and let $h$ be the real-valued function given by \eqref{eq:bromwich-general-payoff}. Denote by 
\[
    H_t^{h,n} := \mathbb E[ h(G_T^n, S_T^n) \mid\mathcal F_t],
    \qquad
    H_t^h := \mathbb E[ h(G_T, S_T) \mid\mathcal F_t]
\]
the corresponding pricing martingales.  

\begin{assumption}\label{assump:strong-lift-convergence}
The lifted models and the Volterra model are defined on the same
Brownian probability space and satisfy the following conditions.
\begin{enumerate}
    \item[(i)] The variance processes satisfy
    \begin{align}\label{eq:strong-variance-convergence}
        \int_0^T \mathbb E[ |\nu_t^n-\nu_t|^2 ]\,\mathrm dt \longrightarrow0.
    \end{align}

    \item[(ii)] $S^n$ and $S$ are martingales, and there exists $\delta>0$ such that
    \begin{align}\label{eq:uniform-stock-moments-lift}
        \sup_{n\geq1} \mathbb E[ (S_T^n)^{2+\delta} ] + \mathbb E[ S_T^{2+\delta} ] <\infty.
    \end{align}
\end{enumerate}
\end{assumption}

The next proposition establishes a basic convergence result for the processes under consideration.

\begin{proposition}\label{prop:lifted-payoff-L2-convergence}
Under Assumption~\ref{assump:strong-lift-convergence}, it holds that
\begin{align}\label{eq:price-process-lifted-convergence}
    \left\| \sup_{t\in[0,T]} |H_t^{h,n}-H_t^h| \right\|_{L^2(\Omega)} \leq 2\|h(G_T^n, S_T^n) - h(G_T, S_T)\|_{L^2(\Omega)} \longrightarrow 0.
\end{align}
\end{proposition}

To shorten the notation, let us write $\delta_n := \|h(G_T^n, S_T^n) - h(G_T, S_T)\|_{L^2(\Omega)}$. The previous proposition shows that, under Assumption \ref{assump:strong-lift-convergence}, $\delta_n \to 0$. Its proof is given in Appendix~\ref{sec:proof-lifted-payoff-L2-convergence}. The next theorem addresses the convergence of the variance-optimal hedge. 

\begin{theorem}[Stability of the variance-optimal hedge]
\label{thm:lifted-hedge-stability}
Under Assumption~\ref{assump:strong-lift-convergence}, write
\begin{align*}
    H_t^{h,n} &= v^{*,h,n} + (\vartheta^{h,n} \cdot S^n)_t + L_t^{h,n},
    \\ H_t^{h} &= v^{*,h} + (\vartheta^{h} \cdot S)_t + L_t^{h}
\end{align*}
for the GKW representations of the corresponding pricing martingales, where $v^{*,h,n}, v^{*,h}$ denote the optimal initial capital, $\vartheta^{h,n}, \vartheta^h$ the variance-optimal hedging portfolio, and $L^{h,n}, L^h$ the non-hedgable residuals. Then 
\begin{align}\label{eq:gains-convergence-lift}
    \left\| \sup_{t \leq T}|(\vartheta^{*,h,n}\cdot S^n)_t - (\vartheta^{*,h}\cdot S)_t| \right\|_{L^2(\Omega)} + \ \|\sup_{t \leq T}|L_t^{h,n}-L_t^h|\|_{L^2(\Omega)} &\leq 2\delta_n,
\end{align}
Moreover, if $\epsilon_{h,n} := \mathbb E[ (L_T^{h,n})^2 ]$ and $\epsilon_h := \mathbb E[ (L_T^h)^2 ]$ denote the minimal mean-squared hedging errors, then
\begin{align}\label{eq:minimal-error-bound}
    |\epsilon_{h,n}-\epsilon_h| \leq \delta_n \left( 2\|h(G_T, S_T)\|_{L^2(\Omega)} + \delta_n \right).
\end{align}
\end{theorem}

The proof of Theorem~\ref{thm:lifted-hedge-stability} is given in Appendix~\ref{sec:proof-lifted-hedge-stability}.
The uniform moment condition \eqref{eq:uniform-stock-moments-lift} can be verified from the results obtained in \cite{gerhold2019moment, MR4099324}; see also \eqref{eq: moment}. Theorem \ref{thm:strong-variance-convergence-regular} in Appendix~\ref{sec:strong-stability} shows that Assumption \ref{assump:strong-lift-convergence}(i) is typically satisfied for completely monotone Volterra kernels with $\mathcal K\in C^1(\mathbb{R}_+)$.

 \subsection{Numerical hedging experiments}\label{sec:numericshedging}

We first check the transform hedge in the Markovian benchmark of
Section~\ref{sec:numericspricing}.  Table~\ref{tab:contour_validation} varies the real part of the
safe-strip contour while keeping the accepted quadrature and ODE controls fixed.  The displayed
variation is below \(2.0\times10^{-6}\) for the price and \(1.2\times10^{-6}\) for the initial
hedge. 

 \begin{table}[htbp]
\centering
\small
\begin{tabular}{rrr}
\toprule
\(R\) & price & initial hedge \\
\midrule
0.25 & 4.24737598 & 0.45299052 \\
0.50 & 4.24737699 & 0.45298997 \\
0.75 & 4.24737793 & 0.45298938 \\
\bottomrule
\end{tabular}
\caption{Contour check for the Heston benchmark.  The price includes the \(z=1\)
linear term and the contour term representing \(-\min(G_T,K)\); the hedge uses the same
decomposition.}
\label{tab:contour_validation}
\end{table}

For comparison with discretely rebalanced gains, we also report the deterministic frozen-state
proxy obtained from the quadratic-error integrand in Theorem~\ref{thm:optimal-strategy-geo}:
\begin{equation}\label{eq:eps-tilde-def}
 \widetilde\varepsilon
  =(1-\rho^2)\sigma^2\nu_0\int_0^T|A(\tau)|^2\,\mathrm d\tau,
 \qquad
 A(\tau)=H_0(1,0)\psi_2^{1,0}(\tau)
 +\int_{\mathcal S_R}H_0(z,0)\psi_2^{z,0}(\tau)\,\zeta_K(\mathrm dz).
\end{equation}
The first summand corresponds to the linear term $G_T$ in the pricing formula. Freezing the stochastic state at time zero makes \eqref{eq:eps-tilde-def} a leading-order proxy, not the exact continuous-time minimal error.

 \subsubsection*{A regular non-Markovian kernel.}
To assess the finite-factor theory without the singular boundary of a fractional kernel, we use
the completely monotone kernel
\[
 \mathcal K(t)= \int_0^1\mathrm e^{-xt}\,\mathrm dx = \frac{1-\mathrm e^{-t}}{t},\qquad \mathcal K(0)=1.
\]
Midpoint quadrature of the representing measure gives
\[
 x_i=\frac{i-\frac{1}2}{N_{\rm fac}},\qquad
 w_i=\frac{1}{N_{\rm fac}},\qquad
 \mathcal K_{N_{\rm fac}}(t)=\sum_{i=1}^{N_{\rm fac}}w_i\mathrm e^{-x_it}.
\]
We set \(\kappa=2\) and retain
\(\theta=\nu_0=0.04\), \(\sigma=0.3\), \(\rho=-0.7\),
\(S_0=K=100\), \(T=1\), and 21 monitoring dates.  This example satisfies the approximation
condition in Assumption~\ref{assumption:nodes-and-weights} explicitly.  Table~ \ref{tab:regular-kernel-deterministic}
separates factor error from transform-solver error.  The deterministic calculation uses 1,001
Fourier nodes on \(|\operatorname{Im}z|\leq80\) and 500 aligned Runge--Kutta steps.  At 256
factors, changing from 801 to 1,001 Fourier nodes moves the  price by \(9.83\times10^{-6}\) and the
hedge by \(4.70\times10^{-8}\); changing from 500 to 800 ODE steps moves the price by less than
\(10^{-12}\).

\begin{table}[htbp]
\centering
 \scriptsize
\begin{tabular}{rrrrrr}
\toprule
 \(N_{\rm fac}\) &  \(\|\mathcal K_{N_{\rm fac}}-\mathcal K\|_2\) &  price &  \(|\Delta C|\) & hedge & \(|\Delta\vartheta_0|\) \\
\midrule
2 & 3.16e-03 & 4.28638844 & 4.06e-05 & 0.46172385 & 2.03e-05 \\
4 & 7.93e-04 & 4.28635802 & 1.01e-05 & 0.46170861 & 5.09e-06 \\
8 & 1.98e-04 & 4.28635042 & 2.53e-06 & 0.46170479 & 1.27e-06 \\
16 & 4.96e-05 & 4.28634851 & 6.31e-07 & 0.46170384 & 3.17e-07 \\
32 & 1.24e-05 & 4.28634804 & 1.56e-07 & 0.46170360 & 7.83e-08 \\
64 & 3.10e-06 & 4.28634792 & 3.71e-08 & 0.46170354 & 1.86e-08 \\
128 & 7.75e-07 & 4.28634789 & 7.42e-09 & 0.46170353 & 3.73e-09 \\
256 & 1.94e-07 & 4.28634788 & -- & 0.46170352 & -- \\
\bottomrule
\end{tabular}
 \caption{Kernel, price, and time-zero hedge convergence for
\(\mathcal K(t)=(1-\mathrm e^{-t})/t\).  Price and hedge errors are measured against the
256-factor deterministic reference.}
\label{tab:regular-kernel-deterministic}
\end{table}

The common-noise experiment below compares every factor count with a 128-factor reference on the same
Brownian paths.  It uses 2,000 paths, a fixed 1,000-step full-truncation grid, and 20 synchronised
batches.  Besides the log-average and payoff, we report the 20-interval gain
\[
 \Gamma_T^{(20)}=\sum_{k=0}^{19}\vartheta_{t_k}
                 (S_{t_{k+1}}-S_{t_k}).
\]
Theorem~\ref{thm:lifted-hedge-stability} concerns the continuous-time gain under its stated strong
convergence assumption; the last column of Table~\ref{tab:regular-kernel-pathwise} is the
corresponding fixed-grid numerical diagnostic.

\begin{table}[htbp]
\centering
\small
 \begin{tabular}{rrrr}
\toprule
 \(N_{\rm fac}\) &  RMSE\((\log G_T)\) &  RMSE\((h_G)\) &  RMSE\((\Gamma_T^{(20)})\) \\
\midrule
4 & 6.33e-06 (3.3e-07) & 4.87e-04 (3.8e-05) & 4.71e-04 (3.7e-05) \\
8 & 1.58e-06 (8.2e-08) & 1.22e-04 (9.4e-06) & 1.17e-04 (9.1e-06) \\
16 & 3.90e-07 (2.0e-08) & 3.00e-05 (2.3e-06) & 2.90e-05 (2.3e-06) \\
32 & 9.29e-08 (4.8e-09) & 7.15e-06 (5.5e-07) & 6.91e-06 (5.4e-07) \\
64 & 1.86e-08 (9.6e-10) & 1.43e-06 (1.1e-07) & 1.38e-06 (1.1e-07) \\
\bottomrule
\end{tabular}
\caption{ Common-noise factor convergence for the  regular non-Markovian kernel.  Entries are mean
batch RMSEs,  with Student-\(t\) 95\% half-widths across 20 synchronised batches in parentheses. }
 \label{tab:regular-kernel-pathwise}
\end{table}

 We also ran the prespecified positive-exponential factor/rate-range sweep for fractional kernels.
At \(\alpha=0.6\), the smallest relative \(L^2(0,T)\) kernel error was 19.1\%, above the 5\%
admission threshold; at \(\alpha=0.75\) it was 1.44\%.  Because the singular kernel has
\(\mathcal K(0)=\infty\), whereas every finite positive exponential sum has a finite value at zero,
the regular-kernel evidence cannot be transferred to the rough case merely by increasing the
factor count.  We therefore exclude fractional-kernel pathwise hedge results from the paper rather
than presenting an unconverged lift as evidence for the limiting Volterra hedge.

\subsubsection*{Rebalancing frequency.}
For the Heston benchmark, the dynamic strategy is recomputed on each trading date from the current
state and the already observed fixings.  Table~\ref{tab:rebalancing_frequency} and Figure~\ref{fig:rebalancing_frequency} use 20 independent
replications of 5,000 paths each. The simulation grid has 1,000 steps; the preceding equivalence
test validates that choice against 2,000 steps for the hedging-MSE estimand.

\begin{table}[htbp]
\centering
\small
 \begin{tabular}{rcc}
\toprule
 rebalancing intervals &  MSE &  paired contrast versus 5 intervals \\
\midrule
5 & 5.538034 (0.078115) & 0.000000 (0.000000) \\
10 & 2.716779 (0.041592) & -2.821256 (0.072467) \\
20 & 1.543345 (0.019965) & -3.994689 (0.077175) \\
50 & 0.979941 (0.011020) & -4.558093 (0.075709) \\
100 & 0.749233 (0.009791) & -4.788801 (0.076793) \\
200 & 0.645082 (0.007471) & -4.892952 (0.075860) \\
\bottomrule
\end{tabular}
 \caption{Out-of-sample hedging MSE as the trading grid is refined.  Parentheses contain
Student-\(t\) 95\% half-widths across independent replications.  Contrasts are formed within each
replication.}
\label{tab:rebalancing_frequency}
\end{table}

 \begin{figure}[htbp]
 \centering
  \includegraphics[width=0.68\linewidth]{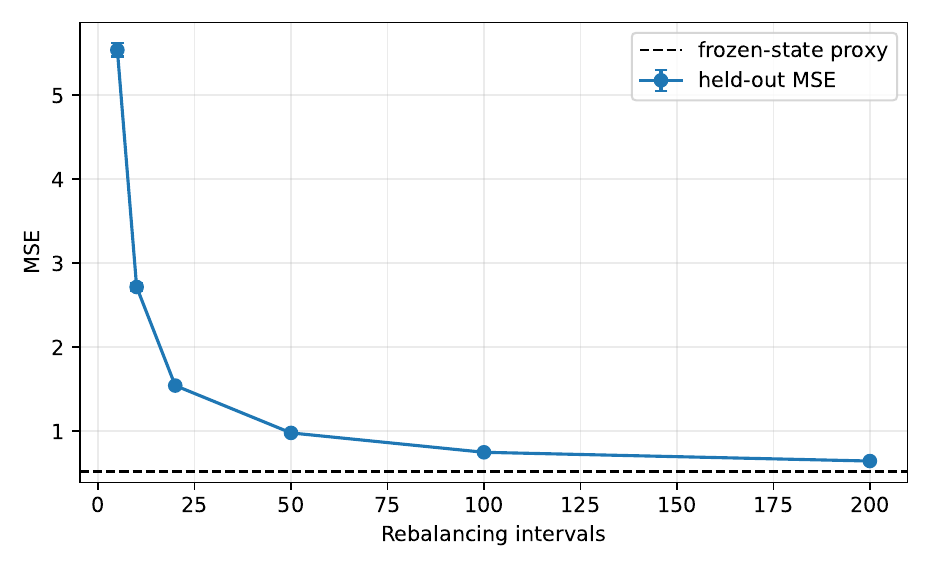}
 \caption{ Hedging MSE against the  number of rebalancing intervals.   Error bars are Student-\(t\)
 95\% intervals across 20 independent replications; the  dashed line is  the  frozen-state proxy
  \(\widetilde\varepsilon=0.518308\) from \eqref{eq:eps-tilde-def}. }
  \label{fig:rebalancing_frequency}
\end{figure}

 The MSE decreases monotonically from 5.538 to 0.645 as the grid is refined from 5 to 200 intervals.
The remaining difference from the frozen-state proxy reflects both discrete trading and the fact
that \eqref{eq:eps-tilde-def} freezes a genuinely stochastic integrand.  All transform controls,
seeds, simulation schemes, factor nodes, and full replication-level outputs are recorded in the
computational release.

\section{Pricing and Hedging of Arithmetic Asian Options}\label{sec:arithmetic-Asian-options}

In this section, we use a control variate approach to accurately and efficiently price and hedge discretely monitored arithmetic Asian options. Recall that both the price and variance-optimal hedge for the geometric Asian payoff $(\bar{G}-K)^+$ are available in semi-closed form from Theorem~\ref{thm:general-geometric-payoff-pricing} and Theorem~\ref{thm:optimal-strategy-geo} respectively. For the arithmetic Asian payoff $h_A(\bar{A}) = (\bar{A}-K)^+$, no comparable affine Fourier representation is available in the present framework, which motivates numerical approaches to the pricing and hedging of such arithmetic Asian payoffs.

\subsection{Pricing of arithmetic Asian options using control variates}\label{sec:CV-pricing-arithmetic}

 Let the arithmetic and geometric averages over \(t_0,\ldots,t_N\) be
\[
 \bar A=\frac{1}{N+1}\sum_{j=0}^N S_{t_j},
 \qquad
 \bar G=\exp\!\left(\frac{1}{N+1}\sum_{j=0}^N\log S_{t_j}\right),
\]
with fixed-strike call payoffs \(h_A(\bar A)=(\bar A-K)^+\) and
 \(h_G(\bar G)=(\bar G-K)^+\).  The arithmetic payoff is not affine, but the geometric payoff has
the semi-analytic value \(C_G\) from Theorem~\ref{thm:general-geometric-payoff-pricing}.  It is
therefore a natural control variate.

For a coefficient \(\lambda\), the  held-out estimator is
\begin{equation}\label{eq:arith-cv-estimator}
 \widehat C_A^{\rm CV}
 =\frac{1}{m}\sum_{i=1}^m h_A(\bar A^{(i)})
 -\widehat\lambda^{\rm tr}
 \left(\frac{1}{m}\sum_{i=1}^m h_G(\bar G^{(i)})-C_G\right),
\end{equation}
where
 \[
 \widehat\lambda^{\rm tr}
 =\frac{\widehat{\operatorname{Cov}}_{\rm tr}(h_A,h_G)}{\widehat{\operatorname{Var}}_{\rm tr}(h_G)}
\]
 is estimated on an independent training sample.  Conditional on the training sample,
\eqref{eq:arith-cv-estimator} is unbiased when \(C_G\) is exact.  The same coefficient is obtained
from discounted or undiscounted payoffs because the common deterministic discount factor cancels.

 The following experiment uses the Heston parameters and 21-date contract from
Section~\ref{sec:numericspricing}.  Each of the 20 independent replications uses 3,000 training paths
and 5,000 held-out paths, simulated with the exact-CIR/trapezoidal scheme and 1,000 time steps.
Direct and control-variate estimators share the held-out paths within a replication.  The geometric
control value is computed with the accepted safe-strip transform controls.  Table~\ref{tab:arithmetic-pricing}
reports replication means and Student-\(t\) 95\% half-widths.

 \begin{table}[htbp]
\centering
 \small
\begin{tabular}{lrr}
\toprule
quantity &  replication mean &  95\% half-width \\
 \midrule
Direct & 4.366431 & 0.043844 \\
Control variate & 4.368634 & 0.001186 \\
CV minus Direct & 0.002203 & 0.043892 \\
Variance reduction (fraction) & 0.999363 & 0.000030 \\
\bottomrule
\end{tabular}
\caption{ Held-out arithmetic-Asian pricing experiment.  The last row is one minus the  ratio of
the  adjusted-payoff variance to the  direct-payoff variance.   Intervals use the  20 independent
replications as  sampling units.}
 \label{tab:arithmetic-pricing}
\end{table}

 The direct and control-variate estimates agree within Monte Carlo uncertainty.  The geometric
control reduces payoff variance by 99.936\%, and the half-width for the mean price falls from
0.0438 to 0.00119.  This comparison includes the cost of estimating \(\lambda\) on a separate
sample and does not rely on the same-sample approximation that can introduce a finite-sample
control-variate bias.

\subsection{Approximate hedging of arithmetic Asian options via control-variate regression}\label{sec:CV-hedging-arithmetic}

We now use the semi-analytic geometric hedge as a control signal in a \emph{finite-basis backward-regression} scheme for arithmetic Asian hedging. The strategy approximates the stock-only GKW hedge on the chosen rebalancing grid; it is not a closed-form expression or a consistency theorem for the continuous-time arithmetic-Asian hedge. Throughout this subsection we work in discounted units. The experiments use $r=0$, so the stock is a $\mathbb P$-martingale; for non-zero rates the same construction applies to discounted stock prices and payoffs.

\subsubsection*{GKW decompositions.}
Write $H_t^G:=\mathbb{E}[h_G(\bar G)\mid\mathcal{F}_t]$ and $H_t^A:=\mathbb{E}[h_A(\bar A)\mid\mathcal{F}_t]$ for the pricing martingales. Assume $h_A(\bar A),h_G(\bar G)\in L^2$ and that the discounted stock $S$ is a square-integrable martingale. Each admits a GKW decomposition with respect to~$S$:
\begin{align}\label{eq:GKW-A-G}
  H_T^G = H_0^G + (\vartheta^{*,G}\cdot S)_T + L_T^G,\qquad
  H_T^A = H_0^A + (\vartheta^{*,A}\cdot S)_T + L_T^A,
\end{align}
where $\vartheta^{*,G}\in \mathcal{G}^2(S)$ and $L^G\perp S$ are given explicitly by Theorem~\ref{thm:optimal-strategy-geo}. The GKW integrand $\vartheta^{*,A}$ and the residual $L^A$ have no closed-form expression.

\subsubsection*{Control variate decomposition.}
For a scalar $\lambda\in\mathbb{R}$, define the \emph{residual payoff}
\begin{align}\label{eq:Delta-payoff}
  \Delta_\lambda := h_A(\bar A) - \lambda\, h_G(\bar G).
\end{align}
For fixed $\lambda$, the linearity of the GKW projection gives
\begin{align}\label{eq:CV-hedge-decomposition}
  \vartheta^{*,A}_t = \lambda\,\vartheta^{*,G}_t + \vartheta^{\Delta}_t,
\end{align}
where $\vartheta^{\Delta}\in \mathcal{G}^2(S)$ is the exact GKW integrand of $\Delta_\lambda$ with respect to~$S$ (to be distinguished from its regression estimate $\hat\vartheta^{\Delta}$ below). The key observation is that, when $\lambda$ is well chosen, the basis-risk payoff $\Delta_\lambda$ has substantially lower variance in the parameter regimes considered below; estimating $\vartheta^{\Delta}$ by regression is therefore empirically more stable than a direct regression for $\vartheta^{*,A}$.

\begin{proposition}[Optimal control variate coefficient]\label{prop:optimal-lambda}
Assume $\mathrm{Var}(h_G(\bar G))>0$. The coefficient $\lambda^*$ minimizing $\mathrm{Var}(\Delta_\lambda)$ is
\begin{align}\label{eq:lambda-star}
  \lambda^* = \frac{\mathrm{Cov}\big(h_A(\bar A),\, h_G(\bar G)\big)}{\mathrm{Var}\big(h_G(\bar G)\big)},
\end{align}
and the minimal residual variance satisfies
\begin{align}\label{eq:residual-var}
  \mathrm{Var}(\Delta_{\lambda^*}) = (1-\varrho_{AG}^2)\,\mathrm{Var}\big(h_A(\bar A)\big),
\end{align}
where $\varrho_{AG}:=\mathrm{Corr}(h_A(\bar A),h_G(\bar G))$.
\end{proposition}

\begin{proof}
This is a standard $L^2$-projection; see also~\cite[Section~4.1]{glasserman2004monte} for the pricing analogue. Minimizing $\mathbb{E}[\Delta_\lambda^2]-(\mathbb{E}[\Delta_\lambda])^2$ over $\lambda$ gives \eqref{eq:lambda-star} by first-order conditions. Substituting back yields $\mathrm{Var}(\Delta_{\lambda^*})=\mathrm{Var}(h_A)-\frac{\mathrm{Cov}(h_A,h_G)^2}{\mathrm{Var}(h_G)}=\mathrm{Var}(h_A)-(\lambda^*)^2\mathrm{Var}(h_G)=\mathrm{Var}(h_A)(1-\varrho_{AG}^2)$.
\end{proof}

Since the arithmetic and geometric averages satisfy $\bar G\leq \bar A$ (AM--GM inequality) and the two averages are close when path dispersion is moderate, the two call payoffs are often highly correlated. This high correlation is an empirical property of the parameter regimes considered below.

\begin{remark}[What $\lambda^*$ optimises]\label{rem:lambda-scope}
Proposition~\ref{prop:optimal-lambda} is an $L^2$ statement about the payoff residual $\Delta_\lambda$: $\lambda^*$ minimizes $\mathrm{Var}(\Delta_\lambda)$. It is not claimed here that the same coefficient minimizes the full hedging error of $\hat\vartheta^A=\lambda\vartheta^{*,G}+\hat\vartheta^\Delta$ over all choices of $\lambda$ and $\hat\vartheta^\Delta$.
\end{remark}

\subsubsection*{Hedging error decomposition.}
We quantify the quality of any strategy of the form $\hat\vartheta^A = \lambda\,\vartheta^{*,G}+\hat\vartheta^{\Delta}$ where $\hat\vartheta^{\Delta}$ is an estimate of the residual hedge.

\begin{proposition}[Hedging error decomposition]\label{prop:hedge-error-decomp}
Let $\hat\vartheta^{\Delta}\in \mathcal{G}^2(S)$ be any predictable process. Define
$$
  \hat\vartheta_t^A := \lambda\,\vartheta_t^{*,G} + \hat\vartheta_t^{\Delta}.
$$
Then the mean-squared hedging error satisfies
\begin{align}\label{eq:mse-decomposition}
  \mathbb{E}\big[\big(h_A(\bar A) - H_0^A - (\hat\vartheta^A\cdot S)_T\big)^2\big]
  = \lambda^2\,\epsilon_G + \mathbb{E}\big[\big(e^\Delta\big)^2\big] + 2\lambda\,\mathbb{E}[L_T^G\, e^{\Delta}],
\end{align}
where $e^\Delta := \Delta_\lambda - \mathbb{E}[\Delta_\lambda] - (\hat\vartheta^{\Delta}\cdot S)_T$ is the residual hedging error, and $\epsilon_G:=\mathbb{E}[(L_T^G)^2]$ is the irreducible hedging error for the geometric Asian, given semi-analytically by Theorem~\ref{thm:optimal-strategy-geo}.
\end{proposition}

\begin{proof}
Write $e^A:= h_A -H_0^A -(\hat\vartheta^A\cdot S)_T$. From the decomposition $h_A=\Delta_\lambda+\lambda h_G$, the GKW identity $h_G=H_0^G+(\vartheta^{*,G}\cdot S)_T+L_T^G$, and the definition $\hat\vartheta^A=\lambda\vartheta^{*,G}+\hat\vartheta^\Delta$, we obtain
\begin{align*}
  e^A &= \Delta_\lambda + \lambda h_G - H_0^A - \lambda(\vartheta^{*,G}\cdot S)_T - (\hat\vartheta^\Delta\cdot S)_T \\
  &= \lambda\bigl[h_G-H_0^G-(\vartheta^{*,G}\cdot S)_T\bigr] + \bigl[\Delta_\lambda-\mathbb{E}[\Delta_\lambda]-(\hat\vartheta^\Delta\cdot S)_T\bigr] + \bigl[\lambda H_0^G+\mathbb{E}[\Delta_\lambda]-H_0^A\bigr] \\
  &= \lambda L_T^G + e^\Delta,
\end{align*}
since the third bracket vanishes by $\mathbb{E}[\Delta_\lambda]=H_0^A-\lambda H_0^G$ (in the discounted, $r=0$ units of this subsection, where $H_0^A=\mathbb{E}[h_A]$ and $H_0^G=\mathbb{E}[h_G]$). Squaring path-by-path and taking expectations gives
$$
    \mathbb{E}[(e^A)^2] = \lambda^2\,\mathbb{E}[(L_T^G)^2] + \mathbb{E}[(e^\Delta)^2] + 2\lambda\,\mathbb{E}[L_T^G\,e^\Delta] = \lambda^2\,\epsilon_G + \mathbb{E}[(e^\Delta)^2] + 2\lambda\,\mathbb{E}[L_T^G\,e^\Delta],
$$
which is~\eqref{eq:mse-decomposition}.
\end{proof}

\begin{remark}[The cross-term]\label{rem:mse-bound}
The cross-term $\mathbb{E}[L_T^G\,e^\Delta]$ in~\eqref{eq:mse-decomposition} does not vanish in general. Orthogonality of $L^G$ to $S$-stochastic integrals gives $\mathbb{E}[L_T^G\,(\hat\vartheta^\Delta\cdot S)_T]=0$, but $\Delta_\lambda-\mathbb{E}[\Delta_\lambda]$ and $L^G$ may both contain volatility-driven components. Even for the exact residual GKW hedge, $L_T^\Delta$ is orthogonal to $S$ but need not be orthogonal to $L_T^G$. Thus \eqref{eq:mse-decomposition} is used only to motivate the low-variance residual regression; the experiments report realised out-of-sample mean-squared errors directly.
\end{remark}

 \subsubsection*{Backward regression scheme for \(\vartheta^\Delta\).}
We use the rebalancing grid \(0=t_0<\cdots<t_N=T\), chosen equal to the monitoring grid.  The
reduced state at date \(t_k\) is
\[
 \mathbf x_k=(\log S_{t_k},\nu_{t_k}^+,A_k,\bar X_k,\vartheta_{t_k}^{*,G}),
 \qquad \nu_{t_k}^+=\max(\nu_{t_k},0),
\]
where
 \[
 A_k=\frac{1}{N+1}\sum_{j=0}^kS_{t_j},
 \qquad
 \bar X_k=\frac{1}{N+1}\sum_{j=0}^k\log S_{t_j}.
\]
 The current fixing is included before the  hedge at  that date is chosen.  The partial averages are
both payoff-bookkeeping variables and regression covariates.

 \begin{definition}[Training-centred control-variate regression hedge]
\label{def:CV-backward-hedge}
Let  \(I_{\rm tr}\) be the  training paths and let
\(\phi=(\phi_1,\ldots,\phi_d)^\top\) be a fixed finite basis.  Estimate
\[
 \widehat\lambda^{\rm tr}
 =\frac{\widehat{\operatorname{Cov}}_{\rm tr}(h_A,h_G)}{\widehat{\operatorname{Var}}_{\rm tr}(h_G)}
\]
 whenever the training variance of \(h_G\) is positive, and define
\[
 \bar h_A^{\rm tr}=\frac{1}{|I_{\rm tr}|}\sum_{i\in I_{\rm tr}}h_A(\bar A^{(i)}),
 \qquad
 \bar\Delta^{\rm tr}=\frac{1}{|I_{\rm tr}|}\sum_{i\in I_{\rm tr}}
 \left(h_A(\bar A^{(i)})-\widehat\lambda^{\rm tr}h_G(\bar G^{(i)})\right).
\]
 For the residual regression, initialise the training target by
\[
 \widehat V_N^{(i)}=h_A(\bar A^{(i)})-
 \widehat\lambda^{\rm tr}h_G(\bar G^{(i)})-\bar\Delta^{\rm tr}.
\]
 For the direct regression use \(h_A(\bar A^{(i)})-\bar h_A^{\rm tr}\) instead.  For
\(k=N-1,\ldots,0\), solve
\begin{align}\label{eq:hedge-ratio-regression}
 \widehat\gamma_k
 :=\argmin_{\gamma\in\mathbb R^d}\frac{1}{|I_{\rm tr}|}
 \sum_{i\in I_{\rm tr}}
 \left(\widehat V_{k+1}^{(i)}-
 \gamma^\top\phi(\mathbf x_k^{(i)})
 (S_{t_{k+1}}^{(i)}-S_{t_k}^{(i)})\right)^2,
\end{align}
 with the stated ridge penalty when used numerically.  Set
\[
 \widehat\vartheta_{t_k}^{\Delta,(i)}
 =\widehat\gamma_k^\top\phi(\mathbf x_k^{(i)}),
 \qquad
 \widehat V_k^{(i)}=\widehat V_{k+1}^{(i)}-
 \widehat\vartheta_{t_k}^{\Delta,(i)}
 (S_{t_{k+1}}^{(i)}-S_{t_k}^{(i)}).
\]
The full control-variate hedge on training and held-out paths is
\begin{align}\label{eq:full-CV-hedge}
 \widehat\vartheta^A_{t_k}
 =\widehat\lambda^{\rm tr}\vartheta^{*,G}_{t_k}
 +\widehat\vartheta^\Delta_{t_k}.
\end{align}
All held-out errors use the training-only arithmetic initial capital
\(\bar h_A^{\rm tr}\).
\end{definition}

 This recursion estimates local gain coefficients rather than continuation values.  It is a
finite-basis, discrete-time local-risk-minimisation procedure, not a closed-form arithmetic-Asian
hedge.

\begin{remark}[Variance reduction in the regression step]
\label{rem:regression-variance-reduction}
The  residual target in \eqref{eq:hedge-ratio-regression}  can have substantially smaller variance
than the arithmetic payoff.  This improves finite-sample stability, but the realised hedging-error
comparison still depends on the fitted strategy and on the cross term in
Proposition~\ref{prop:hedge-error-decomp}.
\end{remark}

\begin{remark}[Choice of basis functions]\label{rem:basis-functions}
The  principal comparison uses the same basis for the direct and control-variate  fits:
\[
 B_{\rm path}=\{1,\log S_{t_k},\nu_{t_k}^+,
 \log S_{t_k}\nu_{t_k}^+,A_k,\bar X_k,\vartheta_{t_k}^{*,G}\}.
\]
Nonconstant columns are standardised using training data only. Because
\(\vartheta_{t_k}^{*,G}\) is included in both fits, the unpenalised methods search the same
finite-dimensional class of discretely rebalanced terminal gains. With ridge regularisation, the
penalty is not invariant under recentering the strategy around the geometric hedge. The common
ridge-sensitivity experiment therefore checks that the comparison is not driven by this
parametrisation effect. Subject to that qualification, the numerical difference measures the
finite-sample effect of changing the regression target rather than a basis-class advantage.
\end{remark}

\begin{proposition}\label{prop:CV-hedge-convergence}
Fix a date $t_k$ in the rebalancing grid and the basis functions $\{\phi_q\}_{q=1}^d$.  Condition on the previously constructed target $\hat V_{k+1}$ and treat it as fixed during the OLS step at date $t_k$ (no quantity already entering $\hat V_{k+1}$ is re-estimated at $t_k$). Suppose:
\begin{enumerate}
\item[(i)] the moment matrix $\mathbb{E}\bigl[\phi(\mathbf{x}_k)\phi(\mathbf{x}_k)^\top(\Delta S_k)^2\bigr]$ is finite and nonsingular, where $\phi=(\phi_1,\dots,\phi_d)^\top$ and $\Delta S_k:=S_{t_{k+1}}-S_{t_k}$;
\item[(ii)] $\mathbb{E}[\hat V_{k+1}^{\,2}]<\infty$.
\end{enumerate}
Then the least-squares estimator $\hat\gamma_k$ defined by~\eqref{eq:hedge-ratio-regression} converges in probability as the sample size $m\to\infty$ to the population projection coefficient
$$
  \gamma_k^* := \arg\min_{\gamma\in\mathbb{R}^d}\mathbb{E}\Big[\big(\hat V_{k+1}-\sum_{q=1}^d\gamma_q\phi_q(\mathbf{x}_k)\Delta S_k\big)^2\Big].
$$
\end{proposition}

\begin{proof}
With $\hat V_{k+1}$ held fixed, \eqref{eq:hedge-ratio-regression} is an ordinary least-squares problem with regressors $\phi(\mathbf{x}_k)\Delta S_k$ and target $\hat V_{k+1}$. Assumptions~(i)--(ii) yield consistency of the OLS estimator for $\gamma_k^*$ by standard least-squares theory; see, e.g., \cite[Chapter 1]{Gyorfi2002}.
\end{proof}

This is only a per-step OLS consistency statement with fixed target. It does not prove consistency of the full backward recursion, because $\hat V_{k+1}$ is itself built from later estimated coefficients on the same sample; propagation of error across dates is left open.

\subsubsection*{Numerical experiments.}
 The arithmetic hedging experiment is restricted to the classical Heston benchmark
\(\alpha=1\).  The fractional-kernel accuracy gate reported in
Section~\ref{sec:numericshedging}  rules out using the available rough lift as evidence for the
limiting Volterra hedge.  The contract and model parameters are otherwise unchanged: 21 monitoring
dates, \(K=100\), \(T=1\), and \(r=0\).  The rebalancing grid coincides with the monitoring grid.

 Each of 20 independent replications uses 3,000 training paths and 5,000 held-out paths.  Paths are
simulated with the exact-CIR/trapezoidal scheme on 1,000 time steps.  The training sample alone
determines the payoff coefficient \(\widehat\lambda^{\rm tr}\), initial-capital estimate, feature
centres and scales, active columns, and regression coefficients.  Nonconstant features are
standardised using training moments, a ridge parameter \(10^{-6}\) is applied after scaling, and a
rank-revealing singular-value decomposition removes numerically constant columns.  At time zero
only the intercept remains active.

Both the direct and control-variate regressions use the same path-aware basis
\[
 B_{\rm path}=\{1,\log S_{t_k},\nu_{t_k}^+,\log S_{t_k}\nu_{t_k}^+,
                  A_k,\bar X_k,\vartheta_{t_k}^{*,G}\}.
\]
Thus the comparison does not give the control-variate method an enlarged strategy class.  In Table~\ref{tab:cv_hedge_comparison}, we
compare:
\begin{enumerate}[label=\emph{(\roman*)}, leftmargin=*]
  \item \textbf{Geometric-only}: \(\widehat\lambda^{\rm tr}\vartheta^{*,G}\), with no
  residual regression;
  \item \textbf{Direct-\(B_{\rm path}\)}: backward local-gain regression on the centred
  arithmetic payoff;
  \item \textbf{CV-\(B_{\rm path}\)}: the geometric component plus backward regression
  on the centred payoff residual.
\end{enumerate}
 All methods use the same training-only arithmetic initial capital.  The reported residual is
\(h_A(\bar A)-\widehat H_0^{A,\rm tr}-(\widehat\vartheta\cdot S)_T\), so the MSE includes any
out-of-sample mean error rather than silently recentering the test sample.

 \begin{table}[htbp]
\centering
 \scriptsize
\begin{tabular}{lrrr}
\toprule
 Method &  MSE &  RMSE &  mean residual \\
\midrule
Geometric-only & 1.703314 (0.021421) & 1.304998 (0.008233) & 0.013306 (0.049354) \\
Direct-\(B_{\rm path}\) & 3.548026 (0.154026) & 1.881735 (0.040459) & 0.007559 (0.050256) \\
CV-\(B_{\rm path}\) & 1.695464 (0.022119) & 1.301979 (0.008518) & 0.014082 (0.049135) \\
CV minus Direct & -1.852562 (0.150718) & -- & -- \\
CV minus Geometric & -0.007850 (0.002740) & -- & -- \\
\bottomrule
\end{tabular}
\caption{ Held-out arithmetic-Asian hedging  results.   Entries are  replication means with
Student-\(t\) 95\% half-widths in parentheses.  Contrasts are  formed within each independent
 replication. }
\label{tab:cv_hedge_comparison}
\end{table}

 The control-variate regression lowers held-out MSE by 52.2\% relative to the direct regression.
The two unpenalised regressions have the same attainable terminal gains. With ridge regularisation,
their parametrisations differ, but the common ridge grid preserves the sign and approximate
magnitude of the paired contrast. The result therefore supports a finite-sample benefit from the
lower-variance target rather than a basis-class advantage. The improvement over
the geometric-only strategy is only 0.46\%, although the paired MSE contrast is negative throughout
its 95\% interval.  Hence the semi-analytic  geometric hedge accounts for nearly all of the benefit
in this benchmark, while the learned correction is small but statistically resolved.  The mean
residual intervals contain zero for all three methods.

 \begin{figure}[htbp]
 \centering
 \includegraphics[width=0.70\linewidth]{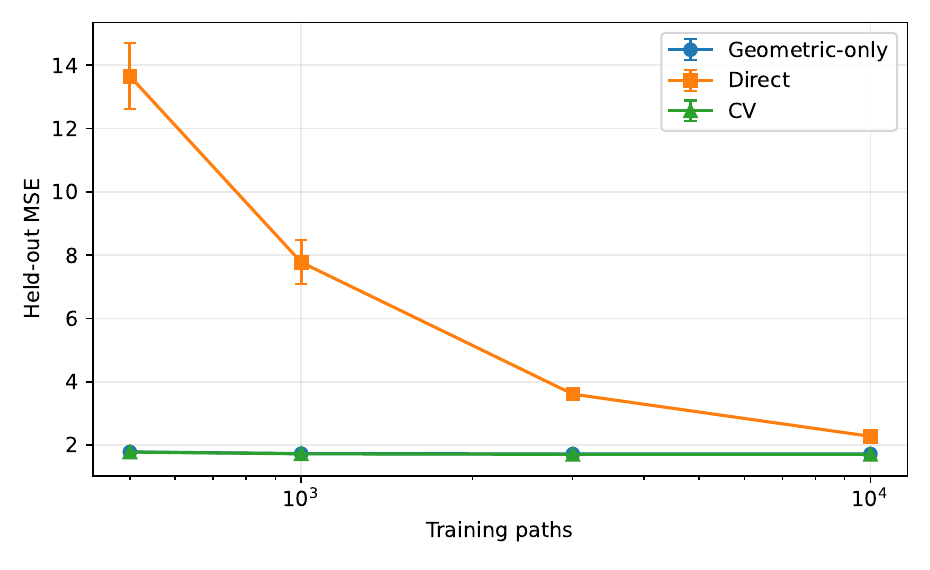}
 \caption{Held-out hedging MSE against the number of training paths.  Direct and control-variate
 regressions use the same basis and held-out paths.  Points are means over 20 independent
 replications and error bars are Student-\(t\) 95\% intervals.}
 \label{fig:arithmetic-learning-curve}
\end{figure}

 Figure~\ref{fig:arithmetic-learning-curve} confirms the finite-sample interpretation.  When the
training size grows from 500 to 10,000, direct-regression MSE falls from 13.66 to 2.27, whereas the
control-variate  MSE remains close to 1.70.  Over the common ridge grid
\(\{10^{-8},10^{-6},10^{-4}\}\), the  paired CV-minus-direct contrast keeps the same sign and its
variation is smaller than two standard errors.  These checks support robustness of the reported
comparison; they do not turn the finite-basis, discrete-time strategy into an exact formula for the
continuous-time arithmetic-Asian GKW hedge.

\medskip
\noindent\textbf{ Reproducibility. }
The full numerical release records the configuration, random seeds, package versions, transform
and simulation controls, all replication-level observations, generated tables and figures, and
SHA-256 hashes.  Quick-mode calculations are execution checks only and are excluded from the
evidentiary chain. The source code and the full release directory are supplied as ancillary files
with the preprint; the archive is regenerated from the documented command in
\texttt{Code/README.md}.

\section{Conclusion}\label{sec:conclusion}

We derived affine-transform pricing formulas and variance-optimal stock hedges for discretely
monitored geometric Asian claims in the Volterra--Heston model. The representation keeps all complex arguments inside the stated affine domain. The hedge formulas are valid under natural moment assumptions stated in the corresponding theorems.

The numerical evidence distinguishes three settings. In the classical Heston benchmark, transform
prices agree with Monte Carlo estimates and geometric Asian options provide highly effective
controls for arithmetic-Asian valuation and regression hedging. For the regular completely
monotone kernel, deterministic and common-noise experiments show stable factor convergence of
prices, initial hedges, and discretely rebalanced gains. By contrast, the prescribed approximation
gate rejects the available finite-factor approximation for the more singular fractional kernel at
\(\alpha=0.6\); no pathwise rough-kernel hedge claim is based on that approximation. Extending the
computational evidence to singular kernels will require factor constructions that resolve the
short-time boundary layer while retaining a stable simulation and hedge implementation.

\appendix
\section{Proofs from Section~\ref{sec:volterra-heston}}\label{sec:proofs-from-section}

\subsection{Fractional Sobolev inequality}\label{sec:1}

\begin{lemma}\label{lemma: fractional Sobolev}
    Let $\mathcal{K} \in L_{\mathrm{loc}}^2(\R_+)$ satisfy 
    $$
        \int_0^h |\mathcal{K}(t)|^2\, \mathrm{d}t + \int_0^T |\mathcal{K}(t+h) - \mathcal{K}(t)|^2\, \mathrm{d}t \leq C_T h^{\gamma}, \qquad h > 0
    $$
    for some constant $\gamma \in (0,2]$. Then for each $\eta \in (0,\gamma/2)$, we find
    $$
        \int_0^T t^{-2\eta}|\mathcal{K}(t)|^2\, \mathrm{d}t + \int_0^T \int_0^T \frac{|\mathcal{K}(t) - \mathcal{K}(s)|^2}{|t-s|^{1+2\eta}}\, \mathrm{d}s\mathrm{d}t < \infty
    $$
\end{lemma}
\begin{proof}
    We first bound the weighted $L^2$-norm. Let $F(t) = \int_0^t \mathcal{K}(s)^2 \, \mathrm{d}s$, then $F(t) \leq C_T t^{\gamma}$ for $t \in [0,T]$ and some constant $C_T > 0$ by assumption. Using integration by parts, we obtain
    \begin{align*}
        \int_0^T t^{-2\eta} |\mathcal{K}(t)|^2 \, \mathrm{d}t 
        &= \int_0^T t^{-2\eta} \, \mathrm{d}F(t) = \left[ t^{-2\eta} F(t) \right]_0^T + 2\eta \int_0^T t^{-2\eta-1} F(t) \, \mathrm{d}t
        \\ &\leq T^{-2\eta}F(T) + 2\eta C_T \int_0^T t^{\gamma - 2\eta - 1} \, \mathrm{d}t
        = T^{-2\eta}F(T) + C_T \frac{2\eta}{\gamma - 2\eta} T^{\gamma - 2\eta} < \infty
    \end{align*}
    where we have used $F(t) \leq C_T t^{\gamma}$ and $\gamma/2 > \eta$ so that the boundary term at zero vanishes. For the second term, we use symmetry to restrict the integral onto $\{h < t\}$, and then the substitution $h = t - s$ to find
    \begin{align*}
        \int_0^T \int_0^T \frac{|\mathcal{K}(t) - \mathcal{K}(s)|^2}{|t-s|^{1+2\eta}} \, \mathrm{d}s \mathrm{d}t 
        &= 2 \int_0^T \int_0^t \frac{|\mathcal{K}(t) - \mathcal{K}(t-h)|^2}{h^{1+2\eta}} \, \mathrm{d}h \mathrm{d}t
        \\ &= 2 \int_0^T \frac{1}{h^{1+2\eta}} \left( \int_h^T |\mathcal{K}(t) - \mathcal{K}(t-h)|^2 \, \mathrm{d}t \right) \mathrm{d}h
        \\ &= 2 \int_0^T \frac{1}{h^{1+2\eta}} \left( \int_0^{T-h} |\mathcal{K}(u+h) - \mathcal{K}(u)|^2 \, \mathrm{d}u \right)\, \mathrm{d}h
        \\ &\leq C_T\int_0^T h^{\gamma - 1 - 2\eta}\, \mathrm{d}h.
    \end{align*}
    The last integral is finite since $\gamma/2 > \eta$, which proves the assertion.
\end{proof}

\subsection{Proof of Theorem~\ref{thm:solRicVolmeasure}}\label{sec:proof-Theorem2.3}

Firstly, let us remark that by Lemma \ref{lemma: fractional Sobolev} 
$$
    [\mathcal{K}]_{\eta, 2, T} = \left( \int_0^T t^{-2\eta}|\mathcal{K}(t)|^2\, \mathrm{d}t + \int_0^T \int_0^T \frac{|\mathcal{K}(t) - \mathcal{K}(s)|^2}{|t-s|^{1+2\eta}}\, \mathrm{d}s\mathrm{d}t \right)^{1/2} < \infty
$$
is finite for each $\eta \in (0,\gamma/2)$ and $T > 0$. If $\mu_i = u_i \delta_0 + f_i \mathrm{d}s$, then \eqref{eq:VolRic} coincides with the equations studied in \cite[Thm. 7.1(ii)]{abi2019affine}. In particular, if $\Re(\psi_1) \in [0,1]$, $\Re(u_2) \leq 0$, and $\Re(f_2) \leq 0$, then $\psi_1(\cdot;u,f) = \mu_1([0,t]) = u_1 + \int_0^t f_1(s)\, \mathrm{d}s$, and there exists a unique solution $\psi_2(\cdot; u,f) \in L_{\mathrm{loc}}^2(\R_+; \C)$ such that $\mathrm{Re}(\psi_2) \leq 0$. Below, we establish the following a-priori $L^p$-bound for the solution.

\begin{lemma}\label{lem:Lpest}
    Let $u\in\mathbb{C}^2$ and $f\in L^2_{\mathrm{loc}}(\mathbb{R}_+,\mathbb{C}^2)$ be such that $\Re(\psi_1)\in[0,1]$, $\Re(u_2)\leq 0$, and $\Re(f_2)\leq 0$. Let $\psi_1, \psi_2$ be the unique solutions of \eqref{eq:phi1} and \eqref{eq:VolRic} with input $\mu_i = u_i \delta_0 + f_i \mathrm{d}s$. Then, for each $T>0$,
    $$
        \|\psi(\cdot, u,f)\|_{L^p([0,T])} 
        \leq C\Bigl(1 + |u| + \|f\|_{L^1([0,T])}\Bigr),
    $$
    where the constant $C\in(0,\infty)$ depends only on the model parameters $(\rho,\sigma,\kappa)$, $p, \gamma$ and $T$. Moreover, it holds that 
    \begin{align*}
    \|\psi_2(\cdot,u,f)\|_{W^{\eta,2}([0,T])} 
    &\leq \|\psi_2(\cdot, u,f)\|_{L^2([0,T])}+ C\big(1 + [\mathcal{K}]_{\eta,2,T}\big) \\
    &\quad \times
        \Big(1 + |u| + \int_0^T |f(t)|\, \mathrm{d}t + \|\psi(\cdot,u,f)\|_{L^1([0,T])} 
              + \|\psi(\cdot,u,f)\|_{L^2([0,T])}^2\Big).
    \end{align*}
\end{lemma}
\begin{proof}
    We first prove the $L^p$-bound. It follows from the proof of \cite[Lem. 7.4]{abi2019affine} that $l\leq\Re(\psi_2(\cdot,u,f))\leq0$ and $|\Im(\psi_2(\cdot,u,f))|\leq h+|\frac{\rho}{\sigma}\Im(\psi_1(\cdot,u,f))|$, with $l$ and $h$ the unique solutions to 
    \begin{align*}
        h(t) &= |\Im(u_2)|\mathcal{K}(t) + |\frac{\rho}{\sigma}\Im(u_1)| + \int_0^t\mathcal{K}(t - s)\Bigg(\Bigg|\frac{\rho}{\sigma}(L \ast \Im(f_1))(s) + \Im(f_2(s)) \\
        &\quad + \frac{1}{2}\Im(\psi_1(s))\left(2(1 - \rho^2)\Re(\psi_1(s)) - 1 + \frac{2\kappa\rho}{\sigma}\right)\Bigg| + (\rho\sigma\Re(\psi_1(s)) - \kappa)h(s)\Bigg)\mathrm{d} s, \\
        l(t) &= \Re(u_2)\mathcal{K}(t) + \int_0^t\mathcal{K}(t - s)\Bigg(\Re(f_2(s)) + \frac{1}{2}\left(\Re(\psi_1(s))^2 - \Re(\psi_1(s)) - \Im(\psi_1(s))^2\right) \\
        &\quad - |\rho\sigma\Im(\psi_1(s))|\left(h(s) + |\frac{\rho}{\sigma}\Im(\psi_1(s))|\right) - \frac{\sigma^2}{2}\left(h(s) + |\frac{\rho}{\sigma}\Im(\psi_1(s))|\right)^2 \\
        &\quad + (\rho\sigma\Re(\psi_1(s)) - \kappa)l(s)\Bigg)\mathrm{d} s,
    \end{align*}
    where $L$ denotes the resolvent of the first kind of $\mathcal{K}$. It is also stated in the proof of \cite[Lem. 7.4]{abi2019affine} that $|\frac{\rho}{\sigma}\Im(\psi_1)|\leq h$. Independently of the value of $\rho\in[-1,1]$, the pointwise bound $|\psi_1(t,u,f)|=|\mu_1([0,t])|\le |u_1|+\int_0^t|f_1(s)|\,\mathrm{d}s$ gives the $\rho$-free estimate $\|\psi_1(\cdot,u,f)\|_{L^p([0,T])} \le T^{1/p}\bigl(|u_1|+\|f_1\|_{L^1([0,T])}\bigr)$, which subsumes and replaces the $\rho\ne 0$ bound $\|\psi_1\|_{L^p}\le T^{1/p}+(\sigma/|\rho|)\|h\|_{L^p}$ of \cite[Lem. 7.4]{abi2019affine}. In particular, no case distinction between $\rho\ne 0$ and $\rho=0$ is required. Similarly, $\|\psi_2(\cdot,u,f)\|_{L^p([0,T])} \leq \|l\|_{L^p([0,T])} + 2\|h\|_{L^p([0,T])}$. The result now follows using the same approach as in \cite[Lem. 3.1]{friesen2024volterra}, noting that one needs to use a comparison principle like the one in \cite[Lem. 1]{gerhold2019moment} to estimate $\|l\|_{L^p([0,T])}$ and $\|h\|_{L^p([0,T])}$, since they cannot be solved explicitly. This proves the $L^p$-bound.

    For the bound in the fractional Sobolev norm, we follow the proof of \cite[Thm. 3.7]{friesen2024volterra}. Firstly, there exists a constant $C > 0$ such that $R$ given by \eqref{eq:R} satisfies $|R(x)| \leq C(|x| + |x|^2)$. The assertion can now be deduced from a modification of \cite[Thm. 3.7]{friesen2024volterra}, using the first part of Lemma~\ref{lem:Lpest} instead of \cite[Lem. 3.1]{friesen2024volterra}.
\end{proof}

\begin{proof}[Proof of Theorem~\ref{thm:solRicVolmeasure}]
     We follow the proof of \cite[Thm. 3.10]{friesen2024volterra}, adjusted to the Volterra--Heston Riccati equation. Fix $T>0$. By \cite[Lem. 3.8]{friesen2024volterra}, there exist sequences $(g_n)_{n\ge1}\subset L^1_{\mathrm{loc}}(\mathbb{R}_+,\mathbb{C})$ and $(f_n)_{n\ge1}\subset L^1_{\mathrm{loc}}(\mathbb{R}_+,\mathbb{C}_-)$ such that $\|g_n\|_{L^1([0,T])}\le |\mu_1|([0,T])$ and $\|f_n\|_{L^1([0,T])}\le |\mu_2|([0,T])$ holds for all $n \geq 1$, and for every $h\in L^q([0,T];\mathbb{C})$, $q\in[1,\infty)$,
     $$
        \int_0^\cdot h(\cdot-s)g_n(s)\,\mathrm{d}s \longrightarrow \int_{[0,\cdot]} h(\cdot-s)\,\mu_1(\mathrm{d}s),
        \qquad
        \int_0^\cdot h(\cdot-s)f_n(s)\,\mathrm{d}s \longrightarrow \int_{[0,\cdot]} h(\cdot-s)\,\mu_2(\mathrm{d}s)
     $$
     in $L^q([0,T];\mathbb{C})$, and pointwise in $t\in[0,T]$ whenever $h\in C([0,T])$ and $h(0)=0$. Moreover, by the construction therein, it follows that $0 \leq \mathrm{Re}(\mu_1^{(n)}([0,t])) \leq 1$. 

     Let $\psi^n=(\psi_1^n,\psi_2^n)$ be the unique solution corresponding to the absolutely continuous input $\mu^{(n)} = \bigl(g_n(s)\,\mathrm{d}s,\ f_n(s)\,\mathrm{d}s\bigr)$, whose existence follows from \cite[Thm. 7.1(ii)]{abi2019affine}. Then $\psi_1^n = 1\ast g_n \longrightarrow 1\ast\mu_1 = \psi_1$ in $L^2([0,T])$. By Lemma~\ref{lem:Lpest}, $\|\psi_2^n\|_{L^2([0,T])}\le C\bigl(1+|\mu|([0,T])\bigr)$, and
     \begin{align*}
         \|\psi_2^n\|_{W^{\eta,2}([0,T])}
         \le \|\psi_2^n\|_{L^2([0,T])} + C\bigl(1+[\mathcal{K}]_{\eta,2,T}\bigr)
             \Bigl(1+|\mu|([0,T])+\|\psi^n\|_{L^1([0,T])}
             +\|\psi^n\|_{L^2([0,T])}^2\Bigr).
     \end{align*}
     Hence $\sup_{n\geq1}\|\psi_2^n\|_{W^{\eta,2}([0,T])}<\infty$. By the compact embedding
     $W^{\eta,2}([0,T];\mathbb{C})\hookrightarrow L^2([0,T];\mathbb{C})$, there exists a subsequence,
     still denoted by $(\psi_2^n)_{n\geq1}$ for simplicity, and a limit $\psi_2\in L^2([0,T];\mathbb{C})$ such that $\psi_2^n\to\psi_2$ in $L^2([0,T];\mathbb{C})$.

     Set $\psi=(\psi_1,\psi_2)$. Since $\psi_2^n = \mathcal{K}\ast f_n + \mathcal{K}\ast R(\psi^n)$, and $\mathcal{K}\ast f_n\to \mathcal{K}\ast\mu_2$ in $L^2([0,T];\mathbb{C})$, it remains to show that
     $\mathcal{K}\ast R(\psi^n)\to \mathcal{K}\ast R(\psi)$ in $L^2([0,T];\mathbb{C})$. Using that
     $|R(u)-R(v)|\le C(1+|u|+|v|)(|u_1-v_1|+|u_2-v_2|)$ for some $C>0$, Young's convolution inequality, and Cauchy--Schwarz, we obtain
     \begin{align*}
         \|\mathcal{K}\ast(R(\psi^n)-R(\psi))\|_{L^2([0,T])}
         &\le \|\mathcal{K}\|_{L^2([0,T])}\|R(\psi^n)-R(\psi)\|_{L^1([0,T])} \\
         &\le C\|\mathcal{K}\|_{L^2([0,T])}
             \bigl(1+\|\psi^n\|_{L^2([0,T])}+\|\psi\|_{L^2([0,T])}\bigr) \\
         &\qquad \times \bigl(\|\psi_1^n-\psi_1\|_{L^2([0,T])}
             +\|\psi_2^n-\psi_2\|_{L^2([0,T])}\bigr).
     \end{align*}
     The right-hand side converges to zero because $\|\mathcal{K}\|_{L^2([0,T])}<\infty$, $\psi_1^n\to\psi_1$ and $\psi_2^n\to\psi_2$ in $L^2([0,T])$, and hence $\|\psi_1^n\|_{L^2([0,T])}$ and $\|\psi_2^n\|_{L^2([0,T])}$ are uniformly bounded in $n$. Therefore $\psi_2$ solves \eqref{eq:VolRic} on $[0,T]$. Uniqueness follows from the local Lipschitz continuity of $R$.

     Finally, the asserted a-priori bounds follow from the corresponding bounds for $\psi^n$, the convergence $\psi^n\to\psi$, and Fatou's lemma.
 \end{proof}

\subsection{Proof of Theorem \ref{thm:stability-Lq}}\label{sec:proofthm24}
\begin{proof}[Proof of Theorem~\ref{thm:stability-Lq}]
    Since $\psi_1^{(n)}(t) = \mu_1^{(n)}([0,t])$ and $\psi_1(t) = \mu_1([0,t])$ by \eqref{eq:phi1}, the pointwise convergence $\psi_1^{(n)} \longrightarrow \psi_1$ is evident. For the convergence of $\psi_2^{(n)}$, let us first note that $\|\mathcal{K}^{(n)}\|_{L^2([0,T])} \leq T^{\eta}[\mathcal{K}^{(n)}]_{\eta,2,T}$, and hence the $L^2([0,T])$-norm of $\mathcal{K}^{(n)}$ is uniformly bounded in $n$. An application of the a-priori bounds from Theorem \ref{thm:solRicVolmeasure} yields
    $$
        \| \psi_2^{(n)}\|_{W^{\eta,2}([0,T])} \leq C\left( 1 + \| \mathcal{K}\|_{L^2([0,T])} + \sup_{n \geq 1}\| \mathcal{K}^{(n)}\|_{L^2([0,T])}\right)\left( 1 + |\mu|([0,T]) + \sup_{n \geq 1}|\mu^{(n)}|([0,T])\right)
    $$
    with a finite right-hand side for all $n$. By the compact embedding $W^{\eta, 2}([0,T]; \C) \hookrightarrow L^2([0,T]; \C)$, we can find for each subsequence $\psi^{(n_k)} = (\psi_1^{(n_k)}, \psi_2^{(n_k)})$ another subsequence, again denoted by $\psi^{(n_k)}$, such that $\psi_2^{(n_k)} \longrightarrow \widetilde{\psi}_2$ in $L^2([0,T]; \C)$ for some limit $\widetilde{\psi}_2$. Let us show that $\widetilde{\psi} = (\psi_1, \widetilde{\psi}_2)$ is a solution to \eqref{eq:VolRic}. 
    
    Since $\psi_2^{(n_k)} = \mathcal{K}^{(n_k)} \ast \mu_2^{(n_k)} + \mathcal{K}^{(n_k)} \ast R(\psi^{(n_k)})$, it suffices to prove that $\mathcal{K}^{(n_k)} \ast \mu_2^{(n_k)} \longrightarrow \mathcal{K} \ast \mu_2$, and $\mathcal{K}^{(n_k)} \ast R(\psi^{(n_k)}) \longrightarrow \mathcal{K} \ast R(\widetilde{\psi})$ in $L^2([0,T]; \C)$. The first claim follows from
    \begin{align*}
        \| \mathcal{K}^{(n_k)} \ast \mu_2^{(n_k)} - \mathcal{K} \ast \mu_2 \|_2
        &\leq \| \mathcal{K}^{(n_k)} \ast (\mu_2^{(n_k)} - \mu_2) \|_2 + \| (\mathcal{K}^{(n_k)} - \mathcal{K})\ast \mu_2\|_2
        \\ &\leq \| \mathcal{K}^{(n_k)}\|_2 |\mu_2^{(n_k)} - \mu_2|([0,T]) + |\mu_2|([0,T])\| \mathcal{K}^{(n_k)} - \mathcal{K}\|_2,
    \end{align*}
    where we write $\| \cdot \|_2 = \| \cdot \|_{L^2([0,T])}$. For the second assertion, let us estimate 
    \begin{align*}
        \| \mathcal{K}^{(n_k)} \ast R(\psi^{(n_k)})  - \mathcal{K} \ast R(\widetilde{\psi})\|_2 
        &\leq \| \mathcal{K}\ast (R(\widetilde{\psi}) - R(\psi^{(n_k)}))\|_2 + \| (\mathcal{K} - \mathcal{K}^{(n_k)}) \ast R(\psi^{(n_k)})\|_2
        \\ &\leq C\| \mathcal{K} \ast (1 + |\widetilde{\psi}| + |\psi^{(n_k)}|)|\widetilde{\psi} - \psi^{(n_k)}| \|_2
        \\ &\quad+ C\| (\mathcal{K} - \mathcal{K}^{(n_k)}) \ast (|\psi^{(n_k)}| + |\psi^{(n_k)}|^2) \|_2,
    \end{align*}
    where we have used $|R(x)| \leq C(|x| + |x|^2)$ and $|R(x) - R(y)| \leq C(1 + |x| + |y|) |x-y|$ for some constant $C > 0$. Using Young's first inequality and then Cauchy--Schwarz, we obtain
    \begin{align*}
        \| \mathcal{K} \ast (1 + |\widetilde{\psi}| + |\psi^{(n_k)}|)|\psi - \psi^{(n_k)}| \|_2
        &\leq \| \mathcal{K}\|_2 \| (1 + |\widetilde{\psi}| + |\psi^{(n_k)}|)|\psi - \psi^{(n_k)}| \|_1
        \\ &\leq \| \mathcal{K}\|_2 \| 1 + |\widetilde{\psi}| + |\psi^{(n_k)}|\|_2 \| \widetilde{\psi} - \psi^{(n_k)} \|_2
        \\ &\leq \| \mathcal{K}\|_2  \left(1 + \|\widetilde{\psi} \|_2 + \sup_{n \geq 1}\|\psi^{(n)}\|_2 \right) \| \widetilde{\psi} - \psi^{(n_k)} \|_2.
    \end{align*}
    Similarly, we show that
    $$
    \| (\mathcal{K} - \mathcal{K}^{(n_k)}) \ast (|\widetilde{\psi}| + |\psi^{(n_k)}|)^2 \|_2
    \leq \| \mathcal{K} - \mathcal{K}^{(n_k)}\|_2 \left( \sup_{n \geq 1} \|\psi^{(n)} \|_1 + \sup_{n \geq 1}\|\psi^{(n)}\|_2^2 \right).
    $$
    This proves that $\mathcal{K}^{(n_k)} \ast R(\psi^{(n_k)}) \longrightarrow \mathcal{K} \ast R(\widetilde{\psi})$ in $L^2([0,T]; \C)$. Hence $\widetilde{\psi}_2$ solves \eqref{eq:VolRic}. By uniqueness, we then conclude that $\psi_2 = \widetilde{\psi}_2$. Since the subsequence was arbitrary, we have shown that any subsequence of $\psi_2^{(n)}$ converges to the same limit $\psi_2$. This implies that $\psi_2^{(n)} \longrightarrow \psi_2$ in $L^2([0,T]; \C)$. 
\end{proof}

\subsection{Proof of Theorem \ref{thm:stability-uniform}}\label{sec:proofthm25}
\begin{proof}[Proof of Theorem~\ref{thm:stability-uniform}]
    Let $q' = \frac{q}{q-1}$ for $q<\infty$ and $q'=1$ for $q=\infty$. Then $\frac{1}{q} + \frac{1}{q'} = 1$ and $2q' \leq q$. The a-priori bounds from Theorem \ref{thm:solRicVolmeasure} yield
    \begin{align}\label{eq:uniform-stability-bound}
        \| \psi \|_{2q'} + \sup_{n \geq 1}\|\psi^{(n)}\|_{2q'} \leq C\left(1 + \| \mathcal{K}\|_{2q'}  + \sup_{n \geq 1}\|\mathcal{K}^{(n)}\|_{2q'} \right)\left( 1 + |\mu|([0,T]) + \sup_{n \geq 1}|\mu^{(n)}|([0,T])\right) < \infty.
    \end{align}
     Write $\psi_2 = \mathcal{K} \ast f_2 + \mathcal{K} \ast R(\psi)$.  Since $f_2 \in L^{q'}([0,T]; \C)$ and $\mathcal{K} \in L^q([0,T])$, Young's convolution inequality shows that $\mathcal{K}\ast f_2 \in C([0,T]; \C)$. Similarly, using $|R(\psi)| \leq C(|\psi| + |\psi|^2) \in L^{q'}([0,T]; \C)$ since $\psi \in L^{2q'}([0,T]; \C)$, we also conclude that $\mathcal{K} \ast R(\psi) \in C([0,T]; \C)$. Thus $\psi_2$ is continuous. The same argument also shows that $\psi_2^{(n)}$ is continuous for each $n \geq 1$. For the uniform convergence of $\psi_2^{(n)} \longrightarrow \psi_2$ on $[0,T]$, arguing as above, we can estimate
    \begin{align*}
        \| \psi_2^{(n)} - \psi_2\|_{\infty} &\leq \| \mathcal{K} \ast (f_2^{(n)} - f_2) \|_{\infty} + \| (\mathcal{K}^{(n)} - \mathcal{K})\ast f_2\|_{\infty}
        \\ &\qquad + C\| \mathcal{K} \ast (1 + |\psi| + |\psi^{(n)}|)|\psi - \psi^{(n)}| \|_{\infty}
        + C\| (\mathcal{K} - \mathcal{K}^{(n)}) \ast (|\psi^{(n)}| + |\psi^{(n)}|^2) \|_{\infty}
        \\ &\leq \| \mathcal{K}\|_{q} \| f_2^{(n)} - f_2\|_{q'} + \| \mathcal{K}^{(n)} - \mathcal{K}\|_q \| f_2\|_{q'}
        \\ &\qquad + C\| \mathcal{K} \|_q \left \| \left( 1 + |\psi| + |\psi^{(n)}|\right)|\psi - \psi^{(n)}| \right\|_{q'} + C \| \mathcal{K} - \mathcal{K}^{(n)}\|_q \left\| (|\psi^{(n)}| + |\psi^{(n)}|^2) \right \|_{q'}
        \\ &\leq \| \mathcal{K}\|_{q} \| f_2^{(n)} - f_2\|_{q'} + \| \mathcal{K}^{(n)} - \mathcal{K}\|_q \| f_2\|_{q'}
        \\ &\qquad + C\| \mathcal{K} \|_q \left( 1 + \| \psi \|_{2q'}^2 + \sup_{n \geq 1}\|\psi^{(n)}\|_{2q'}^2 \right) \left \| \psi - \psi^{(n)} \right\|_{2q'} 
        \\ &\qquad + C \| \mathcal{K} - \mathcal{K}^{(n)}\|_q \left( \sup_{n \geq 1}\|\psi^{(n)}\|_{q'} + \sup_{n \geq 1}\| \psi^{(n)}\|_{2q'}^2 \right)
    \end{align*}
    where the constants on the right-hand side are finite by $q' \leq 2q'$ combined with \eqref{eq:uniform-stability-bound}. The first, second, and last term converge to zero by assumption. The third term tend to zero due to
    \[
        \| \psi_2^{(n)} - \psi_2\|_{2q'} \leq \| \psi_2^{(n)} - \psi_2\|_{2}^{\frac{1}{q'}} \| \psi_2^{(n)} - \psi_2\|_{\infty}^{1 - \frac{1}{q'}},
    \]
    an application of Theorem~\ref{thm:stability-Lq} since $|\mu_2 - \mu_2^{(n)}|([0,T]) \leq \| f_2 - f_2^{(n)}\|_1 \leq T^{\frac{1}{q}}\| f_2 - f_2^{(n)}\|_{q'} \longrightarrow 0$, and the assumption $\sup_{n \geq 1}\| \psi^{(n)}\|_{\infty} < \infty$. Finally, the pointwise convergence $\psi_1^{(n)}(t) = \mu^{(n)}_1([0,t]) \longrightarrow \mu_1([0,t]) = \psi_1(t)$ holds by assumption, which proves the assertion.
\end{proof}

\subsection{Proof of the generalised affine transformation formula}\label{sec:proof-affine-formula}

In this section, we prove the generalised affine transformation formula. Recall that $X_t = (\log S_t, \nu_t)^{\top}$. Then $\mathrm{d}(\log S_t) = \Bigl(r - \tfrac{1}{2}\nu_t\Bigr)\,\mathrm{d}t + \sqrt{\nu_t}\,\mathrm{d}W^S_t$, and $X$ satisfies the bivariate stochastic Volterra equation
\[
  X_t = X_0 + \int_0^t \bm{\mathcal{K}}(t-s)b(X_s)\,\mathrm{d}s
        + \int_0^t \bm{\mathcal{K}}(t-s)\sigma(X_s)\,\mathrm{d}(W^S_s,W^I_s)^\top,
\]
with Volterra kernel $\bm{\mathcal{K}}$, drift, and diffusion coefficients given by
\begin{align}\label{eq:vector}
  \bm{\mathcal{K}} := 
  \begin{pmatrix}
    1 & 0 \\
    0 & \mathcal{K}
  \end{pmatrix}, \qquad b(X) = b^0 + \log S\,b^1 + \nu\,b^2,
  \qquad 
  a(X) = \sigma(X)\sigma(X)^\top = A^0 + \log S\,A^1 + \nu\,A^2,
\end{align}
where $b^0 = (r, \kappa\theta)^{\top}$, $b^1 = \bm{0}$, $b^2 = - (\tfrac{1}{2}, \kappa)^{\top}$, $A^0=A^1=\bm{0}$, and $A^2 = \begin{pmatrix}
    1 & \rho\sigma \\
    \rho\sigma & \sigma^2
  \end{pmatrix}$.
As a first step, we prove the following simple analogue of~\cite[Thm.~4.3]{abi2019affine} for measure-valued initial conditions.

\begin{lemma}\label{lem:VolRicequiv}
Let $B = (b^1,b^2)$ and $A(x) = (xA^1x^\top, xA^2x^\top)$, with $b^i$ and $A^i$ introduced in \eqref{eq:vector}. Then $\psi = (\psi_1, \psi_2)$ given by \eqref{eq:phi1}-\eqref{eq:VolRic} takes the form 
\begin{align}\label{eq:VolRic1}
    \psi = \bigl(\mu + \psi B + \tfrac{1}{2}A(\psi)\bigr)\ast \bm{\mathcal{K}},
\end{align}
where $\bm{\mathcal{K}}$ is defined in equation~\eqref{eq:vector}. Let $E_B = \bm{\mathcal{K}} - R_B\ast \bm{\mathcal{K}}$, where $R_B$ denotes the resolvent of the second kind of $-\bm{\mathcal{K}}B$. Then \eqref{eq:VolRic1} is equivalent to 
\begin{align}\label{eq:VolRic2}
    \psi = \bigl(\mu + \tfrac{1}{2}A(\psi)\bigr)\ast E_B.
\end{align}
\end{lemma}
\begin{proof}
    By definition of $\bm{\mathcal{K}}, B$, and $A(x)$, it is clear that~\eqref{eq:phi1}-\eqref{eq:VolRic} is equivalent to \eqref{eq:VolRic1}. Assume that $\psi$ satisfies \eqref{eq:VolRic2}. Since $E_B\ast(B\bm{\mathcal{K}}) = -R_B\ast\bm{\mathcal{K}}$, we obtain
    $$
        \psi - \psi\ast(B\bm{\mathcal{K}}) 
        = \bigl(\mu + \tfrac{1}{2}A(\psi)\bigr)\ast \bigl(E_B + R_B\ast\bm{\mathcal{K}}\bigr) = \bigl(\mu + \tfrac{1}{2}A(\psi)\bigr)\ast \bm{\mathcal{K}}.
    $$
    Conversely, assume that $\psi$ satisfies \eqref{eq:VolRic1}. Let $\widetilde{R}_B$ denote the resolvent of the second kind of $-B\bm{\mathcal{K}}$. Then
$$
    \psi - \psi\ast \widetilde{R}_B 
    = \bigl(\mu + \tfrac{1}{2}A(\psi)\bigr)\ast 
      \bigl(\bm{\mathcal{K}} - \bm{\mathcal{K}}\ast \widetilde{R}_B\bigr) - \psi\ast \widetilde{R}_B.
$$
Since $\bm{\mathcal{K}}\ast \widetilde{R}_B = R_B\ast \bm{\mathcal{K}}$ 
(see the proof of~\cite[Lem.~4.4]{abi2019affine}), the result follows.
\end{proof}

The following is an analogue of \cite[Thm. 4.3]{abi2019affine}, where function convolutions are replaced by convolutions against $\mu$.
  
\begin{theorem}\label{thm:expafftransf}
   Let $\mu = (\mu_1, \mu_2) \in\mathcal{M}_{\mathrm{ad}}$, and let $\psi = (\psi_1, \psi_2)$ be the unique solution to the Volterra--Riccati equations \eqref{eq:phi1}-\eqref{eq:VolRic}. Recall the definitions of $a$, $b$ from \eqref{eq:vector}. Define the process $(Y_t)_{t\in[0,T]}$ by
\begin{align*}
    Y_t &= Y_0 + \int_0^t \psi(T-s)\,\sigma(X_s)\,\mathrm{d}W_s 
            - \tfrac{1}{2}\int_0^t \psi(T-s)\,a(X_s)\,\psi(T-s)^\top \,\mathrm{d}s, \\
    Y_0 &= \langle X_0, \mu([0,T]) \rangle + \int_0^T \Bigl(\psi(s)b(X_0) 
              + \tfrac{1}{2}\psi(s)a(X_0)\psi(s)^\top\Bigr)\,\mathrm{d}s,
\end{align*}
where $W = (W^S,W^I)^\top$. Then, for all $t \in [0,T]$,
\begin{align}\label{eq:Y}
    Y_t = \mathbb{E}\left[ \int_{[0,T]} \langle X_{T-s}, \mu(\mathrm{d}s) \rangle\,\middle|\,\mathcal{F}_t\right] 
          + \tfrac{1}{2}\int_t^T 
              \psi(T-s)\,a\bigl(\mathbb{E}[X_s|\mathcal{F}_t]\bigr)\,
              \psi(T-s)^\top \,\mathrm{d}s.
\end{align}
Moreover, the process $(\exp(Y_t))_{t\in[0,T]}$ is a local martingale. If it is in fact a true martingale, the following exponential-affine transform formula holds:
\begin{align*}
    \mathbb{E}\left[\exp \left( \int_{[0,T]} \langle X_{T-s}, \mu(\mathrm{d}s) \rangle \right)\,\middle|\,\mathcal{F}_t\right] = \exp(Y_t), \qquad t\in[0,T].
\end{align*}
\end{theorem}
\begin{proof}
    The key steps parallel those in \cite[Thm. 4.3]{abi2019affine}, with function-valued inputs replaced by measure-valued inputs. Concretely, our Lemma~\ref{lem:VolRicequiv} provides the direct analogue of the supplementary \cite[Lem. 4.4]{abi2019affine}. In particular, given $Y_t$, one applies It\^{o}'s formula to $\exp(Y_t)$, notes that the resulting drift term cancels by the Volterra--Riccati equation~\eqref{eq:VolRic}, and hence proves that $\exp(Y_t)$ is a local martingale. The conditional-expectation formula then follows if $\exp(Y)$ is a true martingale.
\end{proof}

To apply Theorem~\ref{thm:expafftransf} for the Volterra--Heston model, we need to verify that $\exp(Y)$ is a true martingale. This is done by the following analogue of~\cite[Thm.~4.5(ii)]{abi2019affine} with respect to measures $\mu$. Recall that $L$ is the resolvent of the first kind of $\mathcal{K}$, and let $\bm{L}$ be the resolvent of the first kind for $\bm{\mathcal{K}}=\mathrm{diag}(1, \mathcal{K})$. Note that $\bm{L} = \mathrm{diag}(\delta_0, L)$, which allows us to further simplify the formulas below.

\begin{theorem}\label{thm:Yexpress}
 Let $\mu = (\mu_1, \mu_2) \in\mathcal{M}_{\mathrm{ad}}$. Then the scalar function $\pi_h = \Delta_h \psi_2 \ast L - \Delta_h(\psi_2 \ast L)$ is right-continuous and of bounded variation on $[0,T-h]$ for every $h\leq T-t$. Moreover, in the notation of Theorem~\ref{thm:expafftransf}, the process $Y$ admits the representation
\begin{align}\label{eq:Y-representation}
        Y_t = \int_{(T-t,T]} \langle X_{T-s}, \,\mu(\mathrm{d}s)\rangle + g(T-t) 
          + (\Delta_{T-t}\psi\ast \bm{L})(0)\,X_t 
          - \bm{\pi}_{T-t}(t)\,X_0 
          + (\mathrm{d}\bm{\pi}_{T-t}\ast X)(t),
\end{align}
where $\bm{\pi}_h = (0, \pi_h)$, and $g(t) = \int_0^t \Bigl(\psi(s)b^0 + \tfrac{1}{2}\psi(s)A^0\psi(s)^\top\Bigr)\,\mathrm{d}s$.
\end{theorem}
\begin{proof}
Remark that $\Delta_h(f\ast \mu)(t) = (\Delta_h f\ast \mu)(t) + (f\ast\Delta_t \mu)(h)$, where $(f \ast \Delta_t \mu)(h) = \int_{[t,t+h]} f_{t+h -r}\, \mu(\mathrm{d}r)$. Hence, using that $\psi = \left( \mu + \frac{1}{2}A(\psi)\right)\ast E_B$ by Lemma \ref{lem:VolRicequiv}, we obtain
     \begin{align*}
         \Delta_{T-t}\psi(x) &= \Delta_{T-t}(\mu\ast E_B)(x) + \Delta_{T-t}\left(\frac{1}{2}A(\psi)\ast E_B\right)(x) \\
         &= \int_0^{T-t+x}E_B(T-t+x-s)\mu(\mathrm{d} s)  + \left(\frac{1}{2}A(\Delta_{T-t}\psi)\ast E_B\right)(x) + \left(\frac{1}{2}A(\psi)\ast\Delta_xE_B\right)(T-t).
     \end{align*}
     Convolving with $\bm{L}$ and applying the Fubini theorem yields
     \begin{align*}
         (\Delta_{T-t}\psi\ast \bm{L})(x) &= \left(\frac{1}{2}A(\Delta_{T-t}\psi)\ast E_B\ast \bm{L}\right)(x) + \left[\left(\mu + \frac{1}{2}A(\psi)\right)\ast(\Delta_\cdot E_B\ast \bm{L})(x)\right](T-t).
     \end{align*}
     Similarly,
     \begin{align*}
         \Delta_{T-t}(\psi\ast \bm{L})(x) &= \left(\frac{1}{2}A(\Delta_{T-t}\psi)\ast E_B\ast \bm{L}\right)(x) + \left[\left(\mu + \frac{1}{2}A(\psi)\right)\ast\Delta_\cdot(E_B\ast \bm{L})(x)\right](T-t).
     \end{align*}
     Computing the difference between these two expressions gives
    \begin{equation*}
         \bm{\pi}_{T-t}(x) = \left[\left(\mu + \frac{1}{2}A(\psi)\right)\ast\Pi_\cdot(x)\right](T-t),
     \end{equation*}
     where $\Pi_h(t) = \Delta_h E_B \ast \bm{L} - \Delta_h (E_B \ast \bm{L})$. By \cite[Thm. 4.5(i)]{abi2019affine}, this function is right-continuous and each term is of locally bounded variation. Hence, also $\bm{\pi}_h$ is right-continuous and of locally bounded variation. Now, by the Fubini theorem,
     \begin{equation*}
         \mathbb{E}[(\mu\ast X)_T|\mathcal{F}_t] = \int_{(T-t,T]} \langle X_{T-s}, \mu(\mathrm{d} s)\rangle + \int_{[0,T-t]} \langle \mathbb{E}[X_{T-s}|\mathcal{F}_t],\, \mu(\mathrm{d} s) \rangle.
     \end{equation*}
     By this, \eqref{eq:Y}, and using \cite[Thm. 4.5(i) and eq. (4.16)]{abi2019affine}, the assertion now follows along the lines of \cite[Thm. 4.5.(ii)]{abi2019affine}. 
\end{proof}    

\begin{proof}[Proof of Theorem~\ref{cor:CorsolRicVolmeasure}]
    We would like to apply the exponential-affine transformation formula from Theorem~\ref{thm:expafftransf}. For this we need to show that $\exp(Y)$ with $Y$ as in (\ref{eq:Y}) is a true martingale, which is done following the proof of \cite[Thm. 7.1]{abi2019affine}. Note that $\exp(Y)$ is a local martingale by Theorem~\ref{thm:expafftransf}, and Theorem~\ref{thm:Yexpress} is applicable, so
    \begin{align*}
        \Re Y_t &= \int_{(T-t,T]}\log S_{T-s}\Re\mu_1(\mathrm{d} s) + \int_{(T-t,T]}\nu_{T-s}\Re\mu_2(\mathrm{d} s) + \Re g(T-t) 
        \\ &\qquad + (\Delta_{T-t}\Re\psi_1\ast\delta_0)(0)\log S_t + (\Delta_{T-t}\Re\psi_2\ast L)(0)\nu_t - \Re\bm\pi_{T-t,1}(t)\log S_0 \\
        &\qquad - \Re\bm\pi_{T-t,2}(t)\nu_0 + (\mathrm{d}\Re\bm\pi_{T-t,1}\ast\log S)_t + (\mathrm{d}\Re\bm\pi_{T-t,2}\ast\nu)(t) \\
        &= \int_{(T - t,T]}\log S_{T - s}\Re\mu_1(\mathrm{d} s) + \int_{(T - t,T]}\nu_{T - s}\Re\mu_2(\mathrm{d} s) + \Re g(T - t) + \Re\psi_1(T - t)\log S_t \\
        &\qquad + (\Delta_{T-t}\Re\psi_2\ast L)(0)\nu_t - \Re\pi_{T-t}(t)\nu_0 + (\mathrm{d}\Re\pi_{T-t}\ast\nu)(t),
    \end{align*}
    since $\bm{\mathcal{K}}_{11} =1$, so its resolvent of the first kind is $\delta_0$. Since $\Re\psi_1\in[0,1]$ integration by parts for measures, which can be applied since $\mu\in\mathcal{M}_{lf}$, yields
    \begin{align*}
        \Re\psi_1&(T-t)\log S_t + \int_{(T-t,T]}\log S_{T-s}\Re\mu_1(\mathrm{d} s) \\
        &= \Re\psi_1(T-t)\log S_t + \Re\mu_1([0,T])\log S_0 - \Re\mu_1([0,T-t])\log S_t - \int_{T-t}^T\Re\mu_1([0,s])\mathrm{d}\log S_{T-s} \\
        &= \Re\psi_1(T)\log S_0 + \int_{0}^t\Re\psi_1(T-s)\mathrm{d}\log S_s \\
        &\leq \Re\psi_1(T)\log S_0 + r T + U_t - \frac{1}{2}\langle U\rangle_t,
    \end{align*}
    with $U_t = \int_0^t\Re\psi_1(T-s)\sqrt{\nu_s}\mathrm{d} B^S_s$. Recall that $L$ is nonnegative and non-increasing by assumption and that $\Re\mu_2\leq0$, $\Re\psi_2\leq0$ and $\nu_t\geq0$, so $\int_{(T-t,T]}\nu_{T-s}\Re\mu_2(\mathrm{d} s)\leq0$, $\Re g(T-t)\leq rT$ and $(\Delta_{T-t}\Re\psi_2\ast L)(0)\nu_t\leq0$. Note also that $\Re\pi_{h}(t) = -\int_0^h\Re\psi_2(h-s)L(t+\mathrm{d} s)$, so $\Re\pi_{h}$ is nonnegative and non-increasing. Hence both $-\Re\pi_{T-t}(t)\nu_0\leq0$ and $(\mathrm{d}\Re\pi_{T-t}\ast\nu)(t)\leq0$. Combining all of the above yields
    \begin{align*}
        |\exp(Y_t)| = \exp(\Re Y_t) \leq \E^{2r T}S_0^{\Re\psi_1(T)}\exp\left(U_t-\frac{1}{2}\langle U\rangle_t\right).
    \end{align*}
    Since the right-hand side is a martingale by \cite[Lem. 7.3]{abi2019affine}, $\exp(Y)$ is also a true martingale. Therefore, the exponential-affine transformation formula from Theorem~\ref{thm:expafftransf} holds, and so the result follows.
\end{proof}

\subsection[Fourier--Laplace transform for g-generalised geometric averages]{Fourier--Laplace transform for \texorpdfstring{$g$}{g}-generalised geometric averages}\label{sec:proof-CFform}

\begin{proof}[Proof of Corollary~\ref{thm:condcfG_nS_Texpression}]
    Define $\Theta_T: [0,T] \longrightarrow [0,T]$ by $\Theta_T(u) = T - u$ and set $g_T = g \circ \Theta_T^{-1}$. Then $\int_{[0,T]}f(u)\, g(\mathrm{d}u) = \int_{[0,T]} f(T-u)g_T(\mathrm{d}u)$ for any $g$-integrable function $f$. Let us define $\mu = (\mu_1, \mu_2)$ by $\mu_1 = s\,g_T + w\,\delta_0$, and $\mu_2 = 0$. Let $\psi = (\psi_1, \psi_2)$ be given by \eqref{eq:phi1}-\eqref{eq:VolRic}. Then $\psi_2 = \mathcal{K} \ast R(\psi)$ and
    $$
        \psi_1(t) = \mu_1([0,t]) = s g_T([0,t]) + w = s g([T-t,T]) + w.
    $$
    Recall that $G_T = G_t\,G_{t,T}$ by \eqref{eq:Gg}, so $\log G_T = \log G_t + \log G_{t,T}$, and so
    \begin{equation}\label{eq:psit-GT-relation}
        \Psi_t(s,w)
        = \E^{-s\log G_t}\,\mathbb{E}\!\left[\exp\bigl( s \log G_T + w\log S_T \bigr)\,\middle|\, \mathcal{F}_t\right],
        \qquad \log G_t = \int_{[0,t]} \log S_y\,g(\mathrm{d}y).
    \end{equation}
    Therefore it suffices to compute the transform of the full average. We have
    \begin{align*}
        \mathbb{E}&\left[\exp \left( s \log(G_T) + w\log(S_T) \right)\middle| \mathcal{F}_{t}\right]
        \\ &= \mathbb{E}\left[\exp \left(\int_{[0,T]}\langle X_{T-y},\mu(\mathrm{d} y)\rangle \right)\middle|\mathcal{F}_{t}\right]
        \\ &= \exp \left(\int_{[0,T]}\langle\mathbb{E}[X_{T-y}|\mathcal{F}_{t}],\mu(\mathrm{d} y)\rangle + \frac{1}{2}\int_{t}^{T}\psi(T-y)A^2\xi_{t}(y)\psi(T-y)^\top\mathrm{d} y \right),
    \end{align*}
    where the last equality holds by Corollary~\ref{cor:CorsolRicVolmeasure}. For the second term, using the definition of $A^2$, we obtain
    \begin{align*}
        \int_{t}^{T}&\psi(T-y)A^2\xi_{t}(y)\psi(T-y)^\top\mathrm{d} y \\
        &= \int_{t}^{T}\left(\psi_1(T-y)^2+2\rho\sigma\psi_1(T-y)\psi_2(T-y)+\sigma^2\psi_2(T-y)^2\right)\xi_{t}(y)\mathrm{d} y.
    \end{align*}
    For the first term, we use the stochastic differential equation satisfied by $\log(S_y)$, to find
    \begin{align*}
        \mathbb{E}[\log(S_y)|\mathcal{F}_{t}] = \log(S_{t}) + r(y-t) - \frac{1}{2}\int_{t}^y\xi_{t}(x)\mathrm{d} x, \qquad y \in [t,T].
    \end{align*}    
    Hence, we obtain
    \begin{align*}
        \int_{[0,T]}\langle\mathbb{E}[X_{T-y}|\mathcal{F}_{t}],\mu(\mathrm{d} y)\rangle 
        &= s \int_{[0,T]} \mathbb{E}[\log(S_{y})|\mathcal{F}_{t}]\, g(\mathrm{d}y) + w\mathbb{E}[\log(S_{T})|\mathcal{F}_{t}] 
        \\ &= s \int_{[0,t]} \log(S_y)\, g(\mathrm{d}y) + s \int_{(t,T]}\left( \log(S_t) + r(y-t) - \frac{1}{2}\int_t^y \xi_t(x)\, \mathrm{d}x\right)\, g(\mathrm{d}y)
        \\ &\qquad + w \left( \log(S_t) + r(T-t) - \frac{1}{2} \int_t^T \xi_t(x)\, \mathrm{d}x\right)
        \\ &= \log(S_t)\big( s g((t,T]) + w \big) + s \int_{[0,t]}\log(S_y)\, g(\mathrm{d}y) 
        \\ &\qquad + r \left( s \int_{(t,T]} (y-t)\, g(\mathrm{d}y) + w(T-t) \right)
        - \frac{1}{2}\int_t^T \big( w + s g((x,T]) \big) \xi_t(x)\, \mathrm{d}x
    \end{align*}
    where we have used Fubini's theorem to find that
    $$
        \int_{(t,T]} \int_t^y \xi_t(x)\, \mathrm{d}x\, g(\mathrm{d}y)
        = \int_t^T \int_{(x,T]} \xi_t(x)\, g(\mathrm{d}y)\, \mathrm{d}x 
        = \int_t^T g((x,T])\xi_t(x)\, \mathrm{d}x.
    $$
    Substituting into \eqref{eq:psit-GT-relation} and cancelling $\exp\left(s\int_{[0,t]}\log S_y\, g(\mathrm{d}y)\right) = G_t^{\,s}$ against the $s\int_{[0,t]}\log(S_y)\,g(\mathrm{d}y)$ contribution yields the stated expression for $\Psi_t(s,w)$.
\end{proof}

\section{Proofs of the results in Section~\ref{sec:pricing}}\label{sec:proofs-from-section-3}

\subsection{Semi-analytic pricing formula}\label{sec:proof-pricing-formula}

\begin{proof}[Proof of Theorem \ref{thm:general-geometric-payoff-pricing}]
    Firstly, bounding the geometric mean by the corresponding arithmetic mean we find $\mathbb{E}[G_T] \leq \int_{[0,T]} \mathbb{E}[S_t]\, g(\mathrm{d}t) < \infty$, and hence $G_T$ is integrable. Since $G_T^R \leq 1 + G_T$, $G_T^R$ is also integrable. Since $\zeta_0$ is finite and $|G_T^z|=G_T^R$ for $z\in\mathcal{S}_R$, we have
    \begin{align*}
        \int_{\mathcal{S}_R} \mathbb{E}\left[|G_T^z|\right] |\zeta_0|(\mathrm{d}z) &= \mathbb{E}[G_T^R]\, |\zeta_0|(\mathcal{S}_R) < \infty.
    \end{align*}
    Similarly, since $|G_T^zS_T^{1-z}| = G_T^R S_T^{1-R} \leq R G_T + (1-R)S_T$ with $z\in\mathcal{S}_R$, we find
    \begin{align*}
        \int_{\mathcal{S}_R} \mathbb{E}\left[ |G_T^zS_T^{1-z}| \right] |\zeta_1|(\mathrm{d}z) &=
        \mathbb{E}\left[ G_T^R S_T^{1-R} \right] |\zeta_1|(\mathcal{S}_R) \leq \left( R \mathbb{E}[G_T] + (1-R)\mathbb{E}[S_T] \right) |\zeta_1|(\mathcal{S}_R) < \infty.
    \end{align*}
    Hence the conditional Fubini's theorem applies to both integral terms in \eqref{eq:bromwich-general-payoff}. It follows that
    \begin{align*}
        H_t^h &= \mathrm{e}^{-r(T-t)} \Bigg( a + b\,\mathbb{E}\left[ G_T\mid\mathcal{F}_t \right] +  c\,\mathbb{E}\left[ S_T\mid\mathcal{F}_t \right] 
        \\ &\qquad \qquad \qquad + \int_{\mathcal{S}_R} \mathbb{E}\left[ G_T^z\mid\mathcal{F}_t \right] \zeta_0(\mathrm{d}z) + \int_{\mathcal{S}_R} \mathbb{E}\left[ G_T^zS_T^{1-z} \mid\mathcal{F}_t \right] \zeta_1(\mathrm{d}z) \Bigg).
    \end{align*}
    Since $G_T = G_tG_{t,T}$ and $G_t$ is $\mathcal{F}_t$-measurable, we get $\mathbb{E}\left[ G_T^z\mid\mathcal{F}_t \right] = G_t^z \mathbb{E}\left[ G_{t,T}^z\mid\mathcal{F}_t \right] = G_t^z\Psi_t(z,0)$. Likewise, we obtain $\mathbb{E}\left[ G_T^zS_T^{1-z} \mid\mathcal{F}_t \right] = G_t^z \mathbb{E}\left[ G_{t,T} ^zS_T^{1-z} \mid\mathcal{F}_t \right] = G_t^z\Psi_t(z,1-z)$ and $\mathbb{E}\left[ G_T\mid\mathcal{F}_t \right] = G_t\Psi_t(1,0)$. Moreover, since $(\mathrm{e}^{-rt}S_t)_{t\in[0,T]}$ is a martingale, $\mathbb{E}\left[ S_T\mid\mathcal{F}_t \right] = \mathrm{e}^{r(T-t)}S_t$, and therefore $\mathrm{e}^{-r(T-t)} \mathbb{E}\left[ S_T\mid\mathcal{F}_t \right] = S_t$. Substituting these identities gives
    \eqref{eq:general-geometric-payoff-transform}.

    Finally, for the first integral on the right-hand side of
    \eqref{eq:general-geometric-payoff-transform},
    \begin{align*}
        \int_{\mathcal{S}_R} \left| G_t^z\Psi_t(z,0) \right| |\zeta_0|(\mathrm{d}z) &\leq |\zeta_0|(\mathcal{S}_R) \mathbb{E}\left[ G_T^R \mid\mathcal{F}_t \right] < \infty
    \end{align*}
    almost surely. Similarly,
    \begin{align*}
        \int_{\mathcal{S}_R} \left| G_t^z\Psi_t(z,1-z) \right| |\zeta_1|(\mathrm{d}z) &\leq |\zeta_1|(\mathcal{S}_R) \mathbb{E}\left[ G_T^R S_T^{1-R} \mid\mathcal{F}_t \right] < \infty
    \end{align*}
    almost surely. This proves the asserted absolute convergence.
\end{proof}

\subsection{Proofs for discrete to continuous}\label{sec:proof-discrete-to-continuous}

\begin{proof}[Proof of Theorem \ref{thm:general-price-process-convergence}]
    Set $D_N := h(G_T^N,S_T)-h(G_T^c,S_T)$. Then $H_t^{h,N}-H_t^{h,c} = \mathrm e^{-r(T-t)} \mathbb E[D_N\mid\mathcal F_t]$. Since $r\geq0$, $\mathrm e^{-r(T-t)}\leq1$, and therefore
    \[
        \sup_{t\in[0,T]} |H_t^{h,N}-H_t^{h,c}| \leq \sup_{t\in[0,T]} \left| \mathbb E[D_N\mid\mathcal F_t] \right|.
    \]
    As shown below, the process $\left( \mathbb E[D_N\mid\mathcal F_t] \right)_{t\in[0,T]}$ is a square-integrable martingale. By Doob's $L^2$ inequality,
    \[
        \mathbb E\left[ \sup_{t\in[0,T]} |H_t^{h,N}-H_t^{h,c}|^2 \right] \leq \mathbb E\left[ \sup_{t\in[0,T]} \left| \mathbb E[D_N\mid\mathcal F_t] \right|^2 \right]
        \leq 4\mathbb E[|D_N|^2].
    \]
    Thus, it remains to show that $\mathbb E[|D_N|^2]\longrightarrow0$. 

    Set $S_T^*:=\sup_{u\in[0,T]}S_u$. Since $M_t:=\mathrm e^{-rt}S_t$ is a square-integrable martingale and $r\geq0$, $S_T^* \leq \mathrm e^{rT}\sup_{u\in[0,T]}|M_u|$. Hence Doob's inequality gives $\mathbb E[(S_T^*)^2] \leq 4\mathrm e^{2rT}\mathbb E[M_T^2] < \infty$. Since $\Delta_N\to0$, the left-Riemann sums imply $\int_{[0,T]}f(u)\,g_N^\Delta(\mathrm du) \longrightarrow \frac1T\int_0^Tf(u)\,\mathrm du$ for every continuous $f:[0,T]\longrightarrow \mathbb R$. Since $S$ is strictly positive and continuous, the map $u\mapsto\log(S_u)$ is almost surely continuous on $[0,T]$. Consequently,
    \begin{align} \label{eq: 1}
        \int_{[0,T]}\log(S_u)\,g_N^\Delta(\mathrm du) \longrightarrow  \frac1T\int_0^T\log(S_u)\,\mathrm du \qquad\text{a.s.},
    \end{align}
    and therefore $G_T^N\longrightarrow G_T^c$ a.s. By Jensen's inequality, $G_T^N \leq \int_{[0,T]}S_u\,g_N^\Delta(\mathrm du) \leq S_T^*$, and similarly $G_T^c \leq \frac1T\int_0^TS_u\,\mathrm du \leq S_T^*$. Hence, for $z\in\mathcal S_R$, $|(G_T^N)^z| = (G_T^N)^R \leq 1+S_T^*$. Moreover, the weighted arithmetic-geometric mean inequality gives
    \begin{align*}
        |(G_T^N)^zS_T^{1-z}| = (G_T^N)^R S_T^{1-R} \leq R G_T^N+(1-R)S_T \leq S_T^*.
    \end{align*}
    The same estimates hold with $G_T^N$ replaced by $G_T^c$.

    Using \eqref{eq: 1}, we find for every fixed $z\in\mathcal S_R$, $(G_T^N)^z \longrightarrow (G_T^c)^z$ and $(G_T^N)^zS_T^{1-z} \longrightarrow (G_T^c)^zS_T^{1-z}$ a.s. Since $\zeta_0$ and $\zeta_1$ have finite total variation, an application of the dominated convergence theorem gives 
    \[
        \int_{\mathcal S_R} (G_T^N)^z\,\zeta_0(\mathrm dz) \longrightarrow \int_{\mathcal S_R} (G_T^c)^z\,\zeta_0(\mathrm dz)
    \]
    and
    \[
        \int_{\mathcal S_R} (G_T^N)^zS_T^{1-z}\,\zeta_1(\mathrm dz) \longrightarrow \int_{\mathcal S_R}(G_T^c)^zS_T^{1-z}\,\zeta_1(\mathrm dz)
    \]
    almost surely. Therefore $h(G_T^N,S_T) \longrightarrow h(G_T^c,S_T)$ a.s. Furthermore,
    \begin{align*}
        |h(G_T^N,S_T)| \leq |a| + (|b|+|c|)S_T^* + |\zeta_0|(\mathcal S_R)(1+S_T^*) + |\zeta_1|(\mathcal S_R)S_T^* \leq C_h(1+S_T^*)
    \end{align*}
    for a deterministic constant $C_h>0$ independent of $N$. The same estimate holds for $h(G_T^c,S_T)$. Since $(1+S_T^*)^2\in L^1(\Omega)$, dominated convergence yields the claim. 
\end{proof}

\begin{proof}[Proof of Theorem \ref{thm:lipschitz-payoff-convergence-rate}]
    Set $S_T^* := \sup_{u\in[0,T]}S_u$. Since $(\mathrm e^{-rt}S_t)_{t\in[0,T]}$ is a martingale, $r\geq0$, and $S_T\in L^{2q}(\Omega)$, Doob's inequality gives $\|S_T^*\|_{L^{2q}(\Omega)}<\infty$. Note also that the prices are well defined. Indeed, by
    \eqref{eq:lipschitz-first-variable},
    \[
        |h(G_T^N,S_T)| \leq |h(1,S_T)| + L_h|G_T^N-1| \leq |h(1,S_T)| + L_h(1+S_T^*),
    \]
    which belongs to $L^q(\Omega)$. The same argument applies to $h(G_T^c,S_T)$.
    
    Set $L_N:=\log G_T^N$, $L_c:=\log G_T^c$, $\delta_N:=\omega_N^N=t_N-t_{N-1}$ and $T_N=T+\delta_N$. Then $\delta_N\leq\Delta_N$. Define the left-Riemann approximation
    \[
        \bar L_N := \frac1T \sum_{j=0}^{N-1} \omega_j^N X_{t_j}.
    \]
    By the definition of $G_T^N$, $L_N = \frac{T}{T_N}\bar L_N + \frac{\delta_N}{T_N}X_T$. We first estimate $\bar L_N-L_c$. Since $L_c = \frac1T\int_0^T X_s\,\mathrm ds$, we have 
    \begin{align*}
        T(\bar L_N-L_c) &= \sum_{j=0}^{N-1} \int_{t_j}^{t_{j+1}} (X_{t_j}-X_s)\,\mathrm ds.
    \end{align*}
    For $u\in(t_j,t_{j+1}]$, define $\kappa_N(u):=t_{j+1}-u$, such that $0\leq\kappa_N(u)\leq\Delta_N$. Integration by parts on each interval gives
    \begin{align}\label{eq:log-riemann-error}
        T(\bar L_N-L_c) &= -\sum_{j=0}^{N-1} \int_{t_j}^{t_{j+1}} (t_{j+1}-u)\,\mathrm dX_u
        = -\int_0^T\kappa_N(u)\,\mathrm dA_u - \int_0^T\kappa_N(u)\,\mathrm dM_u.
    \end{align}
    Since $A$ has finite variation, 
    \begin{align}\label{eq:A-riemann-bound}
        \left\| \int_0^T\kappa_N(u)\,\mathrm dA_u \right\|_{L^{2q}(\Omega)} \leq \Delta_N \big\||A|_T\big\|_{L^{2q}(\Omega)}.
    \end{align}
    By the Burkholder--Davis--Gundy inequality,
    \begin{align*}
        \left\| \int_0^T\kappa_N(u)\,\mathrm dM_u \right\|_{L^{2q}(\Omega)} \leq C_{2q} \left\| \left( \int_0^T \kappa_N(u)^2\, \mathrm d\langle M\rangle_u \right)^{1/2} \right\|_{L^{2q}(\Omega)} \leq C_{2q}\Delta_N \big\| \langle M\rangle_T^{1/2} \big\|_{L^{2q}(\Omega)}.
    \end{align*}
    Therefore
    \begin{align}\label{eq:left-riemann-log-rate}
        \|\bar L_N-L_c\|_{L^{2q}(\Omega)} \leq \frac{\Delta_N}{T} \left( \big\||A|_T\big\|_{L^{2q}(\Omega)} + C_{2q} \big\|\langle M\rangle_T^{1/2}\big\|_{L^{2q}(\Omega)} \right).
    \end{align}

    It remains to control the remaining terminal observation point. Let us write 
    \begin{align}\label{eq:LN-minus-Lc}
        L_N-L_c &= \frac{T}{T_N}(\bar L_N-L_c) + \frac{\delta_N}{T_N}(X_T-L_c).
    \end{align}
    By integration by parts,
    \begin{align*}
        X_T-L_c = \frac1T \left( TX_T-\int_0^TX_s\,\mathrm ds \right)
        &= \frac1T\int_0^T u\,\mathrm dX_u 
        \\ &= \frac1T\int_0^T u\,\mathrm dA_u + \frac1T\int_0^T u\,\mathrm dM_u.
    \end{align*}
    Using the BDG-inequality we arrive at $\|X_T-L_c\|_{L^{2q}(\Omega)} \leq \big\||A|_T\big\|_{L^{2q}(\Omega)} + C_{2q} \big\|\langle M\rangle_T^{1/2}\big\|_{L^{2q}(\Omega)}$. This, combined with \eqref{eq:left-riemann-log-rate}, \eqref{eq:LN-minus-Lc}, and using $\frac{T}{T_N}\leq1$ and $\frac{\delta_N}{T_N} \leq \frac{\Delta_N}{T}$, yields
    \begin{align}\label{eq:log-average-rate}
        \|L_N-L_c\|_{L^{2q}(\Omega)} \leq \frac{2\Delta_N}{T} \left( \big\||A|_T\big\|_{L^{2q}(\Omega)} + C_{2q} \big\|\langle M\rangle_T^{1/2}\big\|_{L^{2q}(\Omega)} \right).
    \end{align}
    
    We next pass from the logarithmic averages to the geometric averages. By Jensen's inequality,
    \[
        G_T^N \leq \int_{[0,T]}S_u\,g_N^\Delta(\mathrm du) \leq S_T^*\quad \text{ and } \quad G_T^c \leq \frac1T\int_0^T S_u\,\mathrm du \leq S_T^*.
    \]
    Using $ |\mathrm e^x-\mathrm e^y| \leq (\mathrm e^x+\mathrm e^y)|x-y|$ for $x,y\in\mathbb R$, we obtain 
    \begin{align*}
        |G_T^N-G_T^c| \leq (G_T^N+G_T^c)|L_N-L_c| \leq 2S_T^*|L_N-L_c|.
    \end{align*}
    By the Lipschitz property~\eqref{eq:lipschitz-first-variable}, H\"older's inequality, and~\eqref{eq:log-average-rate}, we obtain
    \begin{align*}
        \|h(G_T^N,S_T)-h(G_T^c,S_T)\|_{L^q(\Omega)} &\leq L_h \|G_T^N-G_T^c\|_{L^q(\Omega)} 
        \\ &\leq 2L_h \|S_T^*\|_{L^{2q}(\Omega)}\|L_N-L_c\|_{L^{2q}(\Omega)} 
        \leq 4C_q(T) \Delta_N \|S_T^*\|_{L^{2q}(\Omega)}
    \end{align*}
    for a finite constant $C_q(T)$ independent of $N$. Finally, recalling that $D_N := h(G_T^N,S_T)-h(G_T^c,S_T)$. Then $H_t^{h,N}-H_t^{h,c} = \mathrm e^{-r(T-t)} \mathbb E[D_N\mid\mathcal F_t]$. Since $r\geq0$, $\mathrm e^{-r(T-t)}\leq1$, and therefore $\sup_{t\in[0,T]} |H_t^{h,N}-H_t^{h,c}| \leq \sup_{t\in[0,T]} \left| \mathbb E[D_N\mid\mathcal F_t] \right|$. By Doob's $L^q$ inequality,
    \begin{align*}
        \left\| \sup_{t\in[0,T]} |H_t^{h,N}-H_t^{h,c}| \right\|_{L^q(\Omega)} \leq \frac{q}{q-1} \|D_N\|_{L^q(\Omega)} \leq \frac{q}{q-1} L_h C_q(T)\Delta_N.
    \end{align*}
    Taking the $q$-th power yields \eqref{eq:lipschitz-price-process-rate}.
\end{proof}

\section{Proofs from Section \ref{sec:hedging}}

\subsection{Proof of Theorem \ref{thm:optimal-strategy-geo}}\label{sec:proof-optimal-strategy-geo}

\begin{proof}[Proof of Theorem \ref{thm:optimal-strategy-geo}]
We first derive the GKW decomposition for the auxiliary claim $G_T^sS_T^w$, $(s,w)\in\mathcal D$. By \eqref{eq:Hsw-hedging} and \eqref{eq:Y-general-hedge-expanded}, its pricing martingale is $H_t(s,w) = \exp\bigl(Y_t(s,w)\bigr)$. An application of Theorem \ref{thm:expafftransf} gives
\begin{align*}
    \mathrm dH_t(s,w) = H_t(s,w)\sqrt{\nu_t}\, \left(\psi_{1}^{s,w}(T-t) + \rho\sigma \psi_2^{s,w}(T-t) \right)\, \mathrm dW_t^S + H_t(s,w)\sqrt{\nu_t}\sqrt{1-\rho^2}\, \sigma\psi_2^{s,w}(T-t)\, \mathrm dW_t^I.
\end{align*}
Since $\psi_{1,-}^{s,w}$ and $\psi_{1}^{s,w}$ differ only on atoms of $g$, which are at most countable, we may replace $\psi_{1}^{s,w}(T-t)$ by $\psi_{1,-}^{s,w}(T-t)$. Since both functions are equal $\mathrm{d}t$-a.e., it does not alter the equivalence class, whence
\begin{align}\label{eq:dHsw}
    \mathrm dH_t(s,w) = H_t(s,w)\sqrt{\nu_t}\, \beta_t(s,w)\, \mathrm dW_t^S + H_t(s,w)\sqrt{\nu_t}\sqrt{1-\rho^2}\, \sigma\psi_2^{s,w}(T-t)\, \mathrm dW_t^I.
\end{align} 

Since $\mathrm dS_t = S_t \sqrt{\nu_t}\,\mathrm dW_t^S$, we have $\mathrm d\langle S\rangle_t = S_t^2\nu_t\,\mathrm dt$ and $\mathrm d\langle H(s,w),S\rangle_t = H_t(s,w)S_t\nu_t \beta_t(s,w)\,\mathrm dt$. Consequently, using the representation of the GKW integrand in terms of quadratic variations, we find that for $G_T^sS_T^w$ the optimal portfolio is given by
\begin{align}\label{eq:elementary-gkw-integrand}
    \vartheta_t^*(s,w) = \frac{H_t(s,w)}{S_t} \beta_t(s,w).
\end{align}
Furthermore, if
\begin{align}\label{eq:elementary-residual}
    L_t(s,w) &:= \sqrt{1-\rho^2}\,\sigma \int_0^t H_u(s,w) \sqrt{\nu_u}\, \psi_2^{s,w}(T-u)\, \mathrm dW_u^I,
\end{align}
then
\begin{align}\label{eq: elementary GKW}
    H_t(s,w) = H_0(s,w) + (\vartheta^*(s,w)\cdot S)_t + L_t(s,w), \qquad \langle L(s,w),S\rangle\equiv0.
\end{align}

Next, we show that this decomposition is integrable with respect to the measures $\zeta_0, \zeta_1$, which yields the assertion. For $z\in\mathcal S_R$, define $X_0(z):=G_T^z$ and $X_1(z):=G_T^zS_T^{1-z}$. Their $L^2$-norms satisfy
\begin{align}\label{eq: 2}
    \|X_0(z)\|_{L^2(\Omega)}^2 &= \mathbb E[G_T^{2R}], \\
    \|X_1(z)\|_{L^2(\Omega)}^2
    &= \mathbb E[ G_T^{2R}S_T^{2(1-R)} ]
    \leq R\mathbb E[G_T^2] + (1-R)\mathbb E[S_T^2].
\end{align}
Thus both norms are bounded uniformly in $z\in\mathcal S_R$. Moreover, $z\mapsto X_j(z)$ is strongly measurable as an $L^2(\Omega;\mathbb C)$-valued map. Since $\zeta_0$ and $\zeta_1$ have finite total variation, the corresponding contour integrals are well-defined as Bochner integrals in $L^2(\Omega;\mathbb C)$. The GKW projection is a bounded linear projection in $L^2(\Omega)$. In particular, for the elementary GKW decompositions above, $\|\vartheta^*(z,0)\|_{\mathcal G^2_{\mathbb C}(S)} \leq \|G_T^z\|_{L^2}$, $\mathbb{E}[ |L_T(z,0)|^2]  \leq \|G_T^z\|_{L^2(\Omega)}^2$, and similarly $\|\vartheta^*(z,1-z)\|_{\mathcal G^2_{\mathbb C}(S)} \leq \|G_T^zS_T^{1-z}\|_{L^2(\Omega)}$ and $\mathbb{E}[ |L_T(z,1-z)|^2] \leq \|G_T^zS_T^{1-z}\|_{L^2(\Omega)}^2$. The preceding uniform bounds \eqref{eq: 2} combined with the finiteness of $|\zeta_0|(\mathcal S_R)$ and $|\zeta_1|(\mathcal S_R)$ imply
\[
    \int_{\mathcal S_R} \|\vartheta^*(z,0)\|_{\mathcal G^2_{\mathbb C}(S)} |\zeta_0|(\mathrm dz) + \int_{\mathcal S_R} \|\vartheta^*(z,1-z)\|_{\mathcal G^2_{\mathbb C}(S)} |\zeta_1|(\mathrm dz) < \infty,
\]
and analogous estimates for the residual martingales. Hence the Fubini theorem applies to the Bochner and stochastic integrals. 

In particular, using Theorem \ref{thm:general-geometric-payoff-pricing}, we obtain
\begin{align}\label{eq:Hh-transform-hedging}
    H_t^h &= a + bH_t(1,0) + cS_t + \int_{\mathcal S_R} H_t(z,0)\, \zeta_0(\mathrm dz) + \int_{\mathcal S_R} H_t(z,1-z)\, \zeta_1(\mathrm dz).
\end{align}
Inserting \eqref{eq: elementary GKW} into \eqref{eq:Hh-transform-hedging}, using \eqref{eq:elementary-gkw-integrand} and \eqref{eq:elementary-residual}, and rearranging terms gives 
\[
    H_t^h = v^{*,h} + (\vartheta^{*,h} \cdot S)_t + L_t^h,
\]
where $L_t^h = \sqrt{1-\rho^2}\,\sigma \int_0^t \sqrt{\nu_u}\, Q_u^h\, \mathrm dW_u^I$. Since $W^I$ is orthogonal to $W^S$, we also find $\langle L^h,S\rangle\equiv0$. Thus it is the GKW decomposition of $H^h$ with respect to $S$. By uniqueness, \eqref{eq:theta-star-general-g} is the variance-optimal strategy, and the optimal initial capital is \eqref{eq:optimal-capital-general-g}. Finally, the It\^o isometry yields for the GKW residual error
\[
    \epsilon_h = \mathbb E[|L_T^h|^2] = (1-\rho^2)\sigma^2 \int_0^T \mathbb E\left[ \nu_t |Q_t^h|^2 \right] \mathrm dt.
\]
This proves~\eqref{eq:hedging-error-general}.
\end{proof}

\subsection{Proof of Proposition \ref{prop:finite-dimensional-Riccati-lift}}\label{sec:proof-finite-dimensional-Riccati-lift}

\begin{proof}[Proof of Proposition \ref{prop:finite-dimensional-Riccati-lift}]
Fix $(s,w)\in\mathcal D$ and let $\psi_2^{n;s,w}$ denote the unique solution of
\eqref{eq:lifted-Volterra-Riccati-equivalence}. For $i=1,\dots,n$, define 
\[
    \eta_i^{n;s,w}(u) := \int_0^u \mathrm e^{-x_i^n(u-v)} R\left( \psi_1^{s,w}(v), \psi_2^{n;s,w}(v) \right)\,\mathrm dv .
\]
Then $\eta_i^{n;s,w}$ is absolutely continuous and satisfies \eqref{eq:lifted-Riccati-eta} for almost every $u\in[0,T]$. Set $\widetilde\psi_2^{n;s,w}(u) := \sum_{i=1}^n w_i^n\eta_i^{n;s,w}(u)$. By Fubini's theorem and the definition of $\mathcal K^n$,
\begin{align*}
\widetilde\psi_2^{n;s,w}(u) &= \sum_{i=1}^n w_i^n \int_0^u \mathrm e^{-x_i^n(u-v)} R\left( \psi_1^{s,w}(v), \psi_2^{n;s,w}(v) \right)\,\mathrm dv 
\\ &= \int_0^u \mathcal K^n(u-v) R\left( \psi_1^{s,w}(v), \psi_2^{n;s,w}(v) \right)\,\mathrm dv
= \psi_2^{n;s,w}(u),
\end{align*}
where the last equality follows from \eqref{eq:lifted-Volterra-Riccati-equivalence}. This proves
\eqref{eq:lifted-psi2-factor-sum}. Uniqueness of the functions $\eta_i^{n;s,w}$ follows from uniqueness of the corresponding linear equations once $\psi_2^{n;s,w}$ is fixed.

It remains to derive the finite-dimensional representation of the conditional transform. Since $\psi_1^{s,w}$ and $\psi_{1,-}^{s,w}$ differ only at the atoms of $g$, they agree Lebesgue-almost everywhere. Hence Corollary \ref{thm:condcfG_nS_Texpression}, applied to the kernel $\mathcal K^n$, can be rewritten as
\begin{align*}
\Psi_t^n(s,w) = \exp\Bigg( \psi_{1,-}^{s,w}(T-t)X_t^n + \int_t^T \Big[ R\left( \psi_1^{s,w}(T-y), \psi_2^{n;s,w}(T-y) \right) + \kappa\psi_2^{n;s,w}(T-y) \Big] \xi_t^n(y)\,\mathrm dy \Bigg),
\end{align*}
where we have used $R(p,q)+\kappa q = \frac12(p^2-p) + \rho\sigma pq + \frac12\sigma^2q^2$. We now express this integral in terms of the factors $X^n, Y^1,\dots, Y^n$. Fix $t\in[0,T]$ and set $\tau:=T-t$. For $r\in[0,\tau]$ define
\[ 
    m_i(r) := \mathbb E\left[ Y_{t+r}^{n,i} \,\middle|\, \mathcal F_t \right], \qquad i=1,\dots,n.
\]
Taking conditional expectations in \eqref{eq:factor-SDE} gives 
\begin{align}\label{eq:conditional-factor-ode}
m_i'(r) &= -x_i^n m_i(r) + \kappa\left( \theta-\xi_t^n(t+r) \right), \qquad m_i(0)=Y_t^{n,i},
\end{align}
for almost every $r\in[0,\tau]$. Moreover, by \eqref{eq:variance-factor-representation},
\begin{align}\label{eq:forward-variance-factors}
\xi_t^n(t+r) = \nu_0 + \sum_{i=1}^n w_i^n m_i(r).
\end{align}
Define
\[
    F(r) := \sum_{i=1}^n w_i^n \eta_i^{n;s,w}(\tau-r)m_i(r), \qquad r\in[0,\tau].
\]
Using \eqref{eq:lifted-Riccati-eta}, \eqref{eq:conditional-factor-ode}, and \eqref{eq:lifted-psi2-factor-sum}, we obtain for almost every $r\in[0,\tau]$,
\begin{align*}
F'(r) &= \sum_{i=1}^n w_i^n \Big( -(\eta_i^{n;s,w})'(\tau-r)m_i(r) + \eta_i^{n;s,w}(\tau- r)m_i'(r) \Big)
\\ &= -R\left( \psi_1^{s,w}(\tau-r), \psi_2^{n;s,w}(\tau-r) \right) \sum_{i=1}^n w_i^n m_i(r) + \kappa \psi_2^{n;s,w}(\tau-r) \left( \theta-\xi_t^n(t+r) \right).
\end{align*}
By \eqref{eq:forward-variance-factors}, we find $\sum_{i=1}^n w_i^n m_i(r) = \xi_t^n(t+r)-\nu_0$, and therefore
\begin{align*}
-F'(r) &= \Bigg[ R\left( \psi_1^{s,w}(\tau-r), \psi_2^{n;s,w}(\tau-r) \right) + \kappa\psi_2^{n;s,w}(\tau-r) \Bigg] \xi_t^n(t+r)
\\ &\qquad - \nu_0 R\left( \psi_1^{s,w}(\tau-r), \psi_2^{n;s,w}(\tau-r)\right) - \kappa\theta \psi_2^{n;s,w}(\tau-r).
\end{align*}
Integrating over $r\in[0,\tau]$, and using $\eta_i^{n;s,w}(0)=0$, gives
\begin{align*}
&\int_0^\tau \Bigg[ R\left( \psi_1^{s,w}(\tau-r), \psi_2^{n;s,w}(\tau-r) \right) + \kappa\psi_2^{n;s,w}(\tau-r) \Bigg] \xi_t^n(t+r)\,\mathrm dr 
\\ &\qquad = F(0) + \int_0^\tau \Bigg[ \nu_0 R\left( \psi_1^{s,w}(\tau-r), \psi_2^{n;s,w}(\tau-r) \right) + \kappa\theta \psi_2^{n;s,w}(\tau-r) \Bigg]\,\mathrm dr 
\\ &\qquad = \sum_{i=1}^n w_i^n \eta_i^{n;s,w}(\tau)Y_t^{n,i} + \phi^{n;s,w}(\tau),
\end{align*}
where the last equality follows from \eqref{eq:lifted-Riccati-phi}. After the change of variables
$y=t+r$, we have therefore shown that 
\begin{align*}
&\int_t^T \Bigg[ R\left( \psi_1^{s,w}(T-y), \psi_2^{n;s,w}(T-y) \right) + \kappa\psi_2^{n;s,w}(T-y) \Bigg] \xi_t^n(y)\,\mathrm dy \nonumber
\\ &\qquad = \phi^{n;s,w}(T-t) + \sum_{i=1}^n w_i^n \eta_i^{n;s,w}(T-t)Y_t^{n,i}.
\end{align*}
This yields the alternative representation
\[
 \Psi_t^n(s,w) = \exp\Bigg( \psi_{1,-}^{s,w}(T-t)X_t^n + \phi^{n;s,w}(T-t) + \sum_{i=1}^n w_i^n \eta_i^{n;s,w}(T-t)Y_t^{n,i} \Bigg).
\]
Finally, since $G_T^n = \exp\left(A_t^{g,n}\right)G_{t,T}^n$, and $A_t^{g,n} =  \int_{[0,t]} X_u^n\,g(\mathrm du)$ is $\mathcal F_t$-measurable, we obtain 
\[
    H_t^n(s,w) = \mathbb E\left[ (G_T^n)^s(S_T^n)^w \,\middle|\, \mathcal F_t \right] = \mathrm e^{sA_t^{g,n}}\Psi_t^n(s,w),
\]
which gives \eqref{eq:lifted-affine-transform}.
\end{proof}

\subsection{Proof of Proposition \ref{prop:lifted-payoff-L2-convergence}}\label{sec:proof-lifted-payoff-L2-convergence}

\begin{proof}[Proof of Proposition \ref{prop:lifted-payoff-L2-convergence}]
Since $\mathrm dX_t^n = -\frac12\nu_t^n\,\mathrm dt + \sqrt{\nu_t^n}\,\mathrm dW_t^S$ and $\mathrm dX_t = -\frac12\nu_t\,\mathrm dt + \sqrt{\nu_t}\,\mathrm dW_t^S$, we obtain
\[
    X_t^n-X_t = -\frac12 \int_0^t (\nu_s^n-\nu_s)\,\mathrm ds + \int_0^t (\sqrt{\nu_s^n}-\sqrt{\nu_s})\, \mathrm dW_s^S.
\]
Using $|\sqrt{x}-\sqrt{y}|^2 \leq |x-y|$ for $x,y \geq 0$, together with Cauchy--Schwarz and Doob's $L^2$ inequality, we obtain 
\begin{align*}
    \mathbb E\left[ \sup_{t\in[0,T]} |X_t^n-X_t|^2 \right]
    &\leq \frac{T}{2} \int_0^T \mathbb E[ |\nu_s^n-\nu_s|^2 ]\,\mathrm ds + 8 \int_0^T \mathbb E[ |\nu_s^n-\nu_s| ]\,\mathrm ds
    \\ &\leq \frac{T}{2} \int_0^T \mathbb E[ |\nu_s^n-\nu_s|^2 ]\,\mathrm ds
    + 8\sqrt{T} \left( \int_0^T\mathbb E[ |\nu_s^n-\nu_s|^2 ]\,\mathrm ds \right)^{1/2},
\end{align*}
which tends to zero by Assumption~\ref{assump:strong-lift-convergence}(i). Since $g$ is a probability measure,
\begin{align*}
    \left| \log G_T^n-\log G_T \right| = \left| \int_{[0,T]} (X_t^n-X_t)\, g(\mathrm dt) \right|
    \leq \sup_{t\in[0,T]} |X_t^n-X_t|.
\end{align*}
Thus $G_T^n \longrightarrow G_T$ in probability. Likewise, $S_T^n=\exp(X_T^n) \longrightarrow \exp(X_T)=S_T$ in probability. Finally, note that every payoff of the form $h(G_T, S_T)$ has at most linear growth due to $|h(x,y)| \leq C_h(1+x+y)$ for $x,y>0$, which follows from $x^R \leq 1 + x$ combined with $x^R y^{1-R} \leq R x + (1-R)y$. Moreover, finiteness of $|\zeta_0|$ and $|\zeta_1|$ implies, by dominated convergence on the contour, that $h$ is continuous on $\mathbb R_+^2$. Hence $h(G_T^n,S_T^n) \longrightarrow h(G_T,S_T)$ holds in probability. 

Set $S_T^{*,n} := \sup_{t\in[0,T]}S_t^n$ and $S_T^* := \sup_{t\in[0,T]}S_t$. By Jensen's inequality, $G_T^n \leq \int_{[0,T]}S_t^n\,g(\mathrm dt) \leq S_T^{*,n}$, and similarly $G_T\leq S_T^*$. Thus $|h(G_T^n, S_T^n)| \leq C_h(1+S_T^{*,n})$. Let $p:=2+\delta$. Since $S^n$ is a nonnegative martingale, Doob's $L^p$ inequality and \eqref{eq:uniform-stock-moments-lift} give
\[
    \sup_{n\geq1} \mathbb E[ (S_T^{*,n})^{2+\delta}] <\infty.
\]
Consequently, the family $\bigl( h(G_T^n, S_T^n)^2 \bigr)_{n\geq1}$ is uniformly integrable. The assertion follows from Vitali's theorem. Since $H_t^{h,n}-H_t^h = \mathbb E[ h(G_T^n, S_T^n) - h(G_T, S_T) \mid\mathcal F_t ]$, Doob's $L^2$ inequality gives
\[
    \left\| \sup_{t\in[0,T]} |H_t^{h,n}-H_t^h| \right\|_{L^2(\Omega)} \leq 2 \|h(G_T^n, S_T^n) - h(G_T, S_T)\|_{L^2(\Omega)},
\]
which proves the assertion.
\end{proof}

\subsection{Proof of Theorem \ref{thm:lifted-hedge-stability}}\label{sec:proof-lifted-hedge-stability}

\begin{proof}[Proof of Theorem \ref{thm:lifted-hedge-stability}]
To simplify the notation, define $\Gamma_t^{h,n} := \vartheta_t^{*,h,n} S_t^n\sqrt{\nu_t^n}$ and $\Gamma_t^h := \vartheta_t^{*,h} S_t\sqrt{\nu_t}$. In view of Theorem~\ref{thm:lifted-general-hedge} and Theorem~\ref{thm:optimal-strategy-geo}, respectively, these coefficients are given by
\begin{align}
    \Gamma_t^{h,n} &= cS_t^n\sqrt{\nu_t^n} + bH_t^n(1,0)\sqrt{\nu_t^n}\, \beta_t^n(1,0) 
    \nonumber\\ &\quad + \int_{\mathcal S_R} H_t^n(z,0)\sqrt{\nu_t^n}\, \beta_t^n(z,0)\, \zeta_0(\mathrm dz) + \int_{\mathcal S_R} H_t^n(z,1-z)\sqrt{\nu_t^n}\,\beta_t^n(z,1-z)\, \zeta_1(\mathrm dz), \label{eq:Gamma-lifted}
\end{align}
and analogously for $\Gamma^h$. Likewise, define $\Lambda_t^{h,n} := \sqrt{1-\rho^2}\,\sigma \sqrt{\nu_t^n}\,Q_t^{h,n}$ and $\Lambda_t^h := \sqrt{1-\rho^2}\,\sigma \sqrt{\nu_t}\,Q_t^h$, where $Q^{h,n}$ and $Q^h$ are given by \eqref{eq:Qh-general-g-lifted} and \eqref{eq:Qh-general-g}. Then $L_t^{h,n} = \int_0^t \Lambda_u^{h,n}\, \mathrm dW_u^I$ and $L_t^h = \int_0^t \Lambda_u^h\, \mathrm dW_u^I$. In particular, the pricing martingales admit the representation
\begin{align*}
    H_t^{h,n} &= v^{*,h,n} + \int_0^t \Gamma_u^{h,n}\, \mathrm dW_u^S + \int_0^t \Lambda_u^{h,n}\, \mathrm dW_u^I,
    \\ H_t^h &= v^{*,h} + \int_0^t \Gamma_u^h\, \mathrm dW_u^S + \int_0^t \Lambda_u^h\,\mathrm dW_u^I,
\end{align*}
with $v^{*,h,n} = \mathbb E[h(G_T^n, S_T^n)]$ and $v^{*,h} = \mathbb E[h(G_T, S_T)]$. Taking the difference at maturity $t = T$ and noting that all terms are strongly orthogonal, we see that 
\begin{align*}
    \delta_n^2 &= |v^{*,h,n}-v^{*,h}|^2 + \mathbb E\left[ \int_0^T |\Gamma_t^{h,n}-\Gamma_t^h|^2 \,\mathrm dt \right] + \mathbb E\left[ \int_0^T |\Lambda_t^{h,n}-\Lambda_t^h|^2 \,\mathrm dt \right].
\end{align*}
This implies $|v^{*,h,n}-v^{*,h}| \leq \delta_n$, 
\begin{align}
    \left( \mathbb E\left[ \int_0^T |\Gamma_t^{h,n}-\Gamma_t^h|^2 \,\mathrm dt \right] \right)^{1/2} \leq \delta_n \qquad \text{ and } \qquad \left( \mathbb E\left[ \int_0^T |\Lambda_t^{h,n}-\Lambda_t^h|^2 \,\mathrm dt \right] \right)^{1/2} \leq \delta_n.
\end{align}
Note that  
\[
    (\vartheta^{*,h,n}\cdot S^n)_t - (\vartheta^{*,h}\cdot S)_t = \int_0^t (\Gamma_s^{h,n}- \Gamma_s^h) \,\mathrm dW_s^S.
\]
and similarly
\[
    L_t^{h,n}-L_t^h = \int_0^t (\Lambda_s^{h,n}-\Lambda_s^h)\,\mathrm dW_s^I.
\]
The assertion follows from the previous bounds combined with Doob's inequality. Finally,
\begin{align*}
    |\epsilon_{h,n}-\epsilon_h| &= \left| \|L_T^{h,n}\|_{L^2(\Omega)}^2 - \|L_T^h\|_{L^2(\Omega)}^2 \right|
    \\ &\leq \|L_T^{h,n}-L_T^h\|_{L^2(\Omega)} \left( \|L_T^{h,n}\|_{L^2(\Omega)} + \|L_T^h\|_{L^2(\Omega)} \right).
\end{align*}
Since the orthogonal residual is an orthogonal projection component,
\[
 \|L_T^{h,n}\|_{L^2(\Omega)} \leq \|h(G_T^n,S_T^n)\|_{L^2(\Omega)},
 \qquad
 \|L_T^h\|_{L^2(\Omega)} \leq \|h(G_T,S_T)\|_{L^2(\Omega)}.
\]
We conclude by combining these bounds with
\[
 \|h(G_T^n,S_T^n)\|_{L^2(\Omega)}
 \leq \|h(G_T,S_T)\|_{L^2(\Omega)}+\delta_n.
\]
\end{proof}

\subsection{Strong stability for regular Volterra kernels}\label{sec:strong-stability}

\begin{theorem}[Strong convergence of the lifted variance process]
\label{thm:strong-variance-convergence-regular}
Fix $T>0$. Suppose that $\mathcal K$ is completely monotone with $\mathcal K\in C^1([0,T+1])$. For each $n\geq1$, let
\[
    \mathcal K^n(t) = \sum_{i=1}^n w_i^n \mathrm e^{-x_i^n t}, \qquad w_i^n,x_i^n\geq0,
\]
and assume that $\|\mathcal K^n-\mathcal K\|_{L^2([0,T])} \longrightarrow 0$, and
\begin{align}
    \sup_{n\geq1} \left( \|\mathcal K^n\|_{L^\infty([0,T+1])} + \|(\mathcal K^n)'\| _{L^\infty([0,T+1])} \right) &<\infty. \label{eq:regular-kernel-uniform-C1}
\end{align}
Let $\nu$ and $\nu^n$ denote the nonnegative solutions to the associated Volterra equations driven by the same Brownian motion $W^\nu$, i.e.
\begin{align}
    \nu_t &= \nu_0 + \int_0^t \mathcal K(t-s)\kappa(\theta-\nu_s)\,\mathrm ds + \int_0^t \mathcal K(t-s)\sigma\sqrt{\nu_s}\,\mathrm dW_s^\nu,
    \label{eq:regular-limit-variance}
    \\ \nu_t^n &= \nu_0 + \int_0^t \mathcal K^n(t-s)\kappa(\theta-\nu_s^n)\,\mathrm ds + \int_0^t \mathcal K^n(t-s)\sigma\sqrt{\nu_s^n}\,\mathrm dW_s^\nu. \label{eq:regular-lifted-variance}
\end{align}
Then $\int_0^T \mathbb E[ |\nu_t^n-\nu_t|^2 ]\,\mathrm dt \longrightarrow0$.
\end{theorem}
\begin{proof}
We divide the proof into three steps.

\emph{Step 1: Uniform regularity and moment bounds.} By \eqref{eq:regular-kernel-uniform-C1}, there is a constant $C_T>0$ such that, for every $n\geq1$ and every $h\in(0,1]$,
\begin{align*}
    \int_0^h |\mathcal K^n(s)|^2\,\mathrm ds &\leq h\,\|\mathcal K^n\|_{L^\infty([0,T+1])}^2 \leq C_T h,
\end{align*}
while
\begin{align*}
    \int_0^T |\mathcal K^n(s+h)-\mathcal K^n(s)|^2\,\mathrm ds &\leq T h^2 \|(\mathcal K^n)'\| _{L^\infty([0,T+1])}^2 \leq C_T h.
\end{align*}
Hence the sequence $(\mathcal K^n)_{n\geq1}$ satisfies the kernel increment condition in Assumption~\ref{assump:K} uniformly in $n$, with exponent $\gamma=1$. The same is true for $\mathcal K$. Moreover, for every $q\geq2$, one can verify that
\begin{align}
    \sup_{n\geq1}\sup_{t\in[0,T]} \mathbb E[(\nu_t^n)^q] + \sup_{t\in[0,T]} \mathbb E[\nu_t^q] < \infty. \label{eq:uniform-variance-q-moments-regular}
\end{align}

\emph{Step 2: Convergence in probability.}
We show that $\int_0^T |\nu_t^n - \nu_t|^2\, \mathrm{d}t \longrightarrow 0$ in probability. Consider arbitrary subsequences $(n_k)_{k\geq1}$ and $(m_k)_{k\geq1}$. The pair $Z^k:=(\nu^{n_k},\nu^{m_k})$ solves the two-dimensional stochastic Volterra equation
\begin{align*}
    Z_t^k &= (\nu_0,\nu_0)^\top + \int_0^t \bm{\mathcal K}^k(t-s)b(Z_s^k)\,\mathrm ds + \int_0^t \bm{\mathcal K}^k(t-s)a(Z_s^k)\,\mathrm dW_s^\nu,
\end{align*}
where
\[
    \bm{\mathcal K}^k = \begin{pmatrix}
        \mathcal K^{n_k} & 0\\
        0 & \mathcal K^{m_k}
    \end{pmatrix},
\]
and, extending the square root to the whole real line by
$x\mapsto\sqrt{x^+}$,
\[
    b(x_1,x_2) = \kappa \begin{pmatrix} \theta-x_1\\ \theta-x_2 \end{pmatrix},
    \qquad 
    a(x_1,x_2) = \sigma \begin{pmatrix} \sqrt{x_1^+}\\ \sqrt{x_2^+} \end{pmatrix}.
\]
The coefficients are continuous and of linear growth. Moreover,
\[
    \bm{\mathcal K}^k \longrightarrow \begin{pmatrix} \mathcal K&0\\ 0&\mathcal K \end{pmatrix}
\]
in $L^2([0,T])$. Step~1 gives the required uniform kernel regularity. The weak stability theorem for stochastic Volterra equations therefore implies tightness of $(Z^k)_{k\geq1}$ in $L^2([0,T];\mathbb R^2)$, and every limit point $(\bar\nu^1,\bar\nu^2)$ is a weak solution of
\begin{align}
    \bar\nu_t^j &= \nu_0 + \int_0^t \mathcal K(t-s)\kappa(\theta-\bar\nu_s^j)\,\mathrm ds + \int_0^t \mathcal K(t-s)\sigma\sqrt{\bar\nu_s^j}\, \mathrm d\bar W_s^\nu, \qquad j=1,2, \label{eq:doubled-limit-variance}
\end{align}
with the same limiting Brownian motion $\bar W^\nu$ in both coordinates, see \cite[Thm.~7.5]{MR4350978}. 

Since $\mathcal K\in C^1([0,T])$, the limiting stochastic Volterra equation has pathwise uniqueness. Indeed, the drift $x\mapsto\kappa(\theta-x)$ is globally Lipschitz and $x\mapsto\sigma\sqrt{x^+}$ is globally $1/2$-H\"older, so the regular-kernel Yamada--Watanabe theorem for stochastic Volterra equations applies; see, e.g., \cite[Thm.~2.3 and Cor.~2.6]{MR4578147}. Consequently, $\bar\nu^1=\bar\nu^2$ a.s. Thus every weak limit of $(\nu^{n_k},\nu^{m_k})$ is concentrated on the diagonal of $L^2([0,T])\times L^2([0,T])$. By the Gy\"ongy--Krylov criterion, $(\nu^n)_{n\geq1}$ converges in probability in $L^2([0,T])$. The limit must be the unique strong solution $\nu$ of \eqref{eq:regular-limit-variance}, which proves the desired convergence in probability.

\emph{Step 3: Upgrade to $L^2(\Omega\times[0,T])$.}
Set
\[
    Y_n := \int_0^T |\nu_t^n-\nu_t|^2\,\mathrm dt.
\]
By Step 2, we obtain $Y_n\to0$ in probability. Choose any $q>2$ and put $r=q/2>1$. Jensen's inequality gives
\begin{align*}
    Y_n^r &\leq T^{r-1} \int_0^T |\nu_t^n-\nu_t|^{2r}\,\mathrm dt
    \\ &= T^{q/2-1} \int_0^T |\nu_t^n-\nu_t|^q\,\mathrm dt
    \leq 2^{q-1}T^{q/2-1} \int_0^T \bigl((\nu_t^n)^q+\nu_t^q\bigr)\,\mathrm dt.
\end{align*}
By \eqref{eq:uniform-variance-q-moments-regular}, $\sup_{n\geq1}\mathbb E[Y_n^r]<\infty$. Hence $(Y_n)_{n\geq1}$ is uniformly integrable. Since $Y_n\to0$ in probability, Vitali's theorem yields $\mathbb E[Y_n]\longrightarrow0$, which proves the assertion. 
\end{proof}

\bibliographystyle{abbrv}
\bibliography{volterra-heston-asia}

\end{document}